\documentclass[11pt]{book}
\usepackage[english]{babel}

\usepackage{amsmath}

\usepackage{amsthm}

\usepackage{amssymb}

\usepackage{color}

\usepackage{graphicx}

\newtheorem{theorem}{Theorem}[section]

\newtheorem{definition}[theorem]{Definition}

\newtheorem{problem}[theorem]{Problem}

\newtheorem{lemma}[theorem]{Lemma}

\newtheorem{proposition}[theorem]{Proposition}

\newtheorem{corollary}[theorem]{Corollary}

\newtheorem{remark}[theorem]{Remark}

\newtheorem{example}[theorem]{Example}

\newcommand{\rr}{\mathbb{R}}

\newcommand{\cc}{\mathbb{C}}

\newcommand{\upn}{^{^{{[N]}}}}

\newcommand{\ket}[1]{| #1 \rangle }

\newcommand{\bra}[1]{\langle #1 | }

\newcommand{\beq}{\begin{equation}}

\newcommand{\eeq}{\end{equation}}

\newcommand{\ra}{\rangle }

\newcommand{\la}{\langle }

\newcommand{ \ave}[2]{ \langle #1 | #2 | #1 \rangle }

\newcommand{\amp }[2]{\langle #1|#2 \rangle }

\newcommand{\weakv}[3]{ \frac{\langle #1|#2| #3\rangle}{\langle #1 | #3 \rangle }}

\newcommand{\nn}{\nonumber}

\newcommand{\Z}{\mathbb{Z}}

\newcommand{\beqa}{\begin{eqnarray}}
\newcommand{\eeqa}{\end{eqnarray}}
\newcommand{\beqar}{\begin{eqnarray*}}
\newcommand{\eeqar}{\end{eqnarray*}}

\def \la {\langle}
\def \ra {\rangle}

\def \P \mathcal{P}

\input colornam.sty
\input epsf.sty
\numberwithin{section}{chapter}
\numberwithin{equation}{chapter}

\title{\bf The mathematics of superoscillations}

\author{Y. Aharonov\footnote{Schmid College of Science and Technology, Chapman University, Orange 92866, CA, US.}, F. Colombo\footnote{Dipartimento di Matematica, Politecnico di Milano, Via Bonardi 9, 20133 Milano, Italy.}, I. Sabadini$^2$, D.C. Struppa$^1$, J. Tollaksen$^1$}
\date{to appear in {\em Memoirs of the American Mathematical Society }}
\makeindex

\begin{document}
\pagenumbering{arabic}

\maketitle

\tableofcontents

\newpage

In the past 50 years, quantum physicists have discovered, and experimentally demonstrated, a phenomenon which they termed
{\it superoscillations}. Aharonov and his collaborators showed that  superoscillations naturally arise when dealing with weak values, a notion that provides a fundamentally different way to regard measurements in quantum physics. From a mathematical point of view, superoscillating functions are a superposition of small Fourier components with a bounded Fourier spectrum,  which result, when appropriately summed, in a shift that can be arbitrarily large, and well outside the spectrum. Purpose of this work is twofold: on one hand we provide a self-contained survey of the existing literature, in order to offer a systematic mathematical approach to superoscillations; on the other hand, we obtain some new and unexpected results, by showing that  superoscillating sequences can be seen of as solutions to a large class of convolution equations and can therefore be treated within the theory of Analytically Uniform spaces. In particular, we will also discuss the persistence of the superoscillatory behavior when superoscillating sequences are taken as initial values
of  the Schr\"odinger equation and other equations.

%
%
\mainmatter
%
%


\chapter{Introduction}
As it is well known from the rudiments of Fourier Analysis, signals (whether time or space dependent) cannot display details that are smaller than the shortest period of their Fourier components.
This is true for all kind of signals: images, sounds, electrical, etc.
In the last 50 years, quantum physicists have discovered (and experimentally demonstrated) a puzzling phenomenon, which they termed {\it superoscillations} and that seemed to violate this principle.

The original insights, which eventually led to  this discovery,
appeared in 1964 in a paper by Aharonov, Bergman, and Lebowitz
\cite{abl}, where the authors show that, as a result of the uncertainty principle, the initial conditions of a quantum mechanical system can be selected independently of the final conditions. Subsequently it was demonstrated by Aharonov, Albert and  Vaidman \cite{aav}
 that if non-disturbing measurements are performed on such pre- and post-selected systems, then strange outcomes will be obtained during the intermediate time. Traditionally, it was believed that if a measurement interaction is weakened so that there is no disturbance on the system, then no information will be obtained.
 However, it has been shown that information can be obtained even if not a single particle (in an ensemble) was disturbed ~\cite{spie-nswm}. The outcomes of these measurements, what in \cite{aav} are called ``weak values", depend on both the pre- and the post-selection and can have values outside the allowed eigenvalue spectrum and  even complex values.  Using these unique properties of the so-called weak values, weak measurements have been used to discover new physical effects which could not be otherwise detected.

  Aharonov and his collaborators also showed that the weak values lead to another new phenomenon called  superoscillations ~\cite{superosc0}. It has been shown that such behavior has important applications in a variety of areas, including metrology, antenna theory, and a new theory of superresolution in optics. For example, superoscillations do not require a media-substrate (in contrast to evanescent waves) and can therefore be focused much deeper into the media than evanescent waves do ~\cite{b4}, \cite{natmat-2012}.
 The areas in engineering and technology to which superoscillations are being applied seem to be growing on a daily basis, and rather than attempting to offer our own survey, we would like to refer the reader to the work of Lindberg \cite{lindberg}, as well as the work of Berry and his coauthors, \cite{berry2}, \cite{berry}, \cite{berry-noise-2013}, \cite{dennis2008}, \cite{b1}, \cite{b6}, and \cite{b5}. Paper \cite{lindberg}, in particular, contains a wealth of fairly recent references that highlight the state of the art of applications in this area. It is because of the importance of these applications that experimental groups around the world are working to apply these ideas to build new imaging and measuring devices.

From a mathematical point of view, superoscillatory functions demonstrated that a superposition of small Fourier components with a bounded Fourier spectrum,  in modulus less than $1$, can nevertheless result in a shift by an arbitrarily large $a$, well outside the spectrum. They can be thought of as an approximation of $e^{iax}$ in terms of a sequence of the form
\[
 \left\{\sum_{j=0}^n C_j(n,a)e^{ik_j(n)x}\right\}_{n=0}^{\infty}.
\]
The example, which is usually considered prototypical, derives from the sequence of functions:
\begin{equation}
F_n(x,a)=\Big(\cos \Big(\frac{x}{n}\Big)+ia\sin \Big(\frac{x}{n}\Big)\Big)^n
=\left( {{1+a}\over 2}e^{ix/n} +{{1-a}\over 2}e^{-ix/n}\right)^n
\label{eq1}
\end{equation}
where $a\in\mathbb R$, $a>1$. By performing a binomial expansion, this sequence can be written as
$$
\sum_{j=0}^n C_j(n,a)e^{i(1-2j/n)x}
$$
for suitable coefficients $C_j(n,a)$ and thus we see that the largest wavelength in the expansion is $1$.
However, around $|x|<\sqrt n$, $F_n(x,a)$ can be approximated as $F_n(x,a)\approx e^{ia x}$, that is, with a wavelength much larger than one.  This phenomenon is very general and holds for a wide range of functions and coefficients.

The  literature on superoscillations has been growing rapidly in physics journals and, in recent times, in mathematics as well. For example, it is known that regions of superoscillations are typical in random fields ~\cite{dennis2008}.
From a mathematical point of view, we have recently offered the foundations for a rigorous treatment of such a phenomenon, \cite{acsst}, and a good survey of the ideas up to \cite{acsst} is given in \cite{lindberg}. Naturally,
a key component to the superoscillatory phenomenon is the extremely rapid oscillation in the coefficients $C_j$ and
since the regions of superoscillations are created at the expense of having the function grow exponentially in other regions,
it would be natural to conclude that the superoscillations would be quickly ``over-taken" by tails coming from the exponential regions and would thus be short-lived.  However, it has been shown that superoscillations are remarkably robust~\cite{b4} and can last for an arbitrarily long time \cite{ACSST3}, at least if we take the limit for $n$ going to infinity. In the process of establishing such results, we discovered some unexpected relations between the theory of superoscillations and convolutions in Analytically Uniform spaces.
From the perspective of communication theory, it has been shown that this relationship is also related to a trade-off between signal-to-noise and bandwidth~\cite{kempf1}, making it easier to engineer superoscillatory signals (see \cite{kempf2} and also its precursor \cite{superosc-misc2}). In this regard, we should also point out some early work on entire functions and bandwidth limited signals \cite{bondcahn}, \cite{requicha}.

There are many fundamental mathematical questions, such as the optimization of superoscillations, their longevity in time, their ubiquitous occurrence in a broad spectrum of different settings such as throughout group theory, as well as questions of numerical nature: for example, for practical applications, superoscillations (constructed out of a precise interference of non-superoscillatory waves) can be particularly sensitive to noise. If, for example, in \eqref{eq1} we take $n=10$, $a=4$, and we add a random phase noise of $10^{-4}$, the superoscillations quickly disappears, see \cite{berry-noise-2013}. Issues of numerical stability for superoscillations, for example, are analyzed in \cite{leeferreira2}.

In this memoir, we offer a comprehensive introduction to the mathematical theory of superoscillations, and prove a large number of  results describing their properties which are most helpful in making progress on these open questions.  We also show that the phenomenon actually arises in a much larger context than has been previously foreseen.

\bigskip

\noindent {\bf Overview of the memoir.}

\bigskip

In
Chapter 2 we give an introduction to the theory of weak measurements which led to superoscillatory behavior.  This theory is gaining increasing importance among theoretical and applied physicists as demonstrated by the number of experiments devoted to its clarification. This chapter, which is written with the notations used in the quantum physics community, can be easily skipped by the uninterested reader, but we have included it because we believe that the physical motivations may give important support to otherwise surprising results.

In Chapter 3, we offer a rigorous treatment of the superoscillatory phenomenon in terms of the Taylor and the Fourier coefficients of a superoscillating sequence. We use this treatment to deduce important properties of these functions.

Chapter 4 gives an overview of the main mathematical tool that we will use throughout this memoir, namely Ehrenpreis' theory of Analytically Uniform spaces (AU-spaces), and its applications to convolution equations.
The fundamental idea, here, was introduced first in \cite{acsst6}; essentially we noticed that superoscillating sequences can be thought of as solutions to very special cases of convolution equations and we discovered that the theory of AU-spaces and its extension due to Berenstein and Taylor \cite{bt} can be applied with success.
We will also show how to use Dirichlet series to construct further classes of superoscillating sequences.

Chapter 5 deals with the question of the permanence of the superoscillatory behavior when superoscillating sequences are taken as initial values
of  the Schr\"odinger equation for the free particle. In particular we show that if we evolve a superoscillating sequence according to the Schr\"odinger equation the outcome remains superoscillating for all values of $t$, as long as we take $n\to + \infty$.

In Chapter 6 we extend the ideas described in Chapter 5, and we show how the use of the  theory of formal solutions to the Cauchy problem in
the complex domain can be combined with the theory of AU-spaces to generalize the study of superoscillations longevity to the case in which they are taken as initial values for a wide class of differential equations.

It is physically interesting to replace the spatial variable $x$ in a superoscillating sequence with an operator (for example the momentum operator), so to obtain superoscillating sequences of operators. Chapter 7 is dedicated to this topic. As we pointed out earlier, superoscillations appear in rather unexpected settings. For example, if one considers the angular momentum for a system of particles, and computes its weak values, one ends up with a superoscillatory behavior. Thus, we see, in Chapter 8, that superoscillations can be discovered in classical groups, and we explicitly discuss the case of $SO(3)$. Once again this chapter will reverse to a more physical notation.
\vskip 0.5truecm
\noindent {\bf Acknowledgments}. The authors are grateful to Chapman University for the support while writing this memoir. This work was supported in part by Israel Science Foundation grant No. 1125/10, and The Wolfson Family Charitable Trust.
The authors are very grateful to Professors M. Berry, R. Buniy, S. Popescu, and S. Nussinov for their interest in this work and for the many critical (yet helpful) commentaries. Finally, the authors wish to express their gratitude to the anonymous referee, who made many significant suggestions to improve the manuscript, and pointed out some important references that had been added to the present version.


\chapter{Physical motivations}

\section{Overview}

In this chapter we offer an overview of the background from quantum physics which generated the notion of superoscillations. This chapter does not contain any original material and in fact we refer to \cite{aav}, \cite{superosc0}, \cite{aharonov_book}, \cite{townes}, and \cite{av2v} for more details. Its purpose is to help the reader who is not familiar with the physics related to superoscillations to get a sense of the general framework of our investigation, and it can be skipped by the reader interested only in the mathematical aspects. For consistency with the literature, we have used in this chapter the terminology and the notations used in the quantum physics community, while the rest of this memoir employs more traditional mathematical notations.

In a broader context, superoscillations are examples of unusual weak values \cite{PT-nov-2010} which can be obtained for pre- and post-selected quantum systems. The story of the weak value begins in 1964 with the advent of the two-time reformulation of quantum mechanics, by Aharonov, Bergmann and Lebowitz in \cite{abl}.  The usual formulation of quantum mechanics is given in terms of an initial wavefunction or quantum state, which is then propagated forward in time according to the Schr\"odinger equation.  Outcomes of experiments then occur randomly upon measurement with the probabilities given in terms of this forward evolved wavefunction.

Thus, while the Schr\"odinger equation is time reversal symmetric, the introduction of measurements and actual recorded events seems to spoil this feature. This time asymmetric view of quantum mechanics can be made symmetric by realizing that the process of preparation is actually a kind of filtering of results:  only one state of many possible states is chosen to begin with.  By introducing the concept of post-selection, i.e. filtering the final results by a selection criterion (just as one does in a preparation), then the theory can be made once more time-symmetric.

Once two boundary conditions are supplied, one in the past and one in the future, one can think of the past state moving forward in time, or equivalently the future state moving backwards in time (or both). While this is a natural interpretation of this picture, it is by no means required.

Even though the two-time approach to quantum physics can be applied to any situation that conventional quantum physics can, perhaps its most useful new consequence is the notion of {\it weak value}.   This idea was first presented in a 1988 paper (see \cite{aav}) written by Aharonov, Albert, and Vaidman, with the provocative title ``How the result of a measurement of a component of the spin of a spin-1/2 particle can turn out to be 100".   The idea is to take a pre- and post-selected average of the weak measurement results of an operator.  Here, a weak measurement consists in weakly coupling a meter to the system, usually taken to be an impulsive interaction with the meter prepared in a Gaussian state.  Without post-selection, the meter would be shifted either up or down by a small amount, depending on which eigenstate the system is prepared in.

However, in \cite{aav} the authors showed that with system post-selection, the meter can be deflected by an amount much larger (in principle arbitrarily large) than the shift without post-selection.
If  a system is pre-selected in an initial quantum state $\ket{\Psi_{\mathrm{in}}}$, and post-selected on a final state $\ket{\Psi_{\mathrm{fin}}}$, then the result of weakly measuring the operator ${A}$ is not its expectation value, but rather what Aharonov, Albert, and Vaidman called its weak value:
\beq
\langle A\rangle_w = \weakv { \Psi_{\mathrm{fin}} } {{A}} {\Psi_{\mathrm{in}}},
\label{eq: AAV}
\eeq
an object which, for a linear operator, can exceed the eigenvalue range, and even assume complex values.

Having weak values outside the spectrum of the operators involved has been discussed at length in the past and has most comprehensively been investigated for spin-1/2 systems~\cite{botero}, \cite{berry2}, \cite{b1}, \cite{b6}, \cite{b4}, \cite{b5}, \cite{new}.  In addition, Berry et al.~\cite{b5} looked at superweak statistics for much more general situations, and they proved that if the Hilbert space is sufficiently high in dimensions, and if the pre- or post-selection are, in a sense, `generic', then the existence of superweak values becomes common or typical.

\section{Von Neumann measurements}

In classical mechanics, a single particle is described by its position  and momentum. In quantum mechanics any system is described by its quantum state, namely by a vector in a Hilbert space which is, in general, infinite-dimensional.
From the experimental point of view, once a quantum system has been prepared in a particular eigenstate, one could then ascertain with certainty some measurable quantities since the state of the system is already in an eigenstate of the operator corresponding to the prepared state. Repeating the same measurement without any significant evolution of the quantum state will lead to the same result.
The values  for a measurement of non-commuting observables performed after the preparation are described by a probability distribution (either continuous or discrete),
depending on the quantity being measured.
The process by which a quantum state becomes one of the eigenstates of the operator corresponding to the measured observable is called "collapse", or "wavefunction collapse".

 The measurement is usually assumed to be ideally accurate, so the dynamic state of a system after measurement is assumed to collapse into an eigenstate of the operator corresponding to the measurement.
 It is a postulate of quantum mechanics that all measurements \index{measurement} have an associated operator (called an observable) such that:
\begin{enumerate}
\item[i)]
    The observable \index{observable} is a Hermitian (self-adjoint) operator ${A}: D( A)\subset \mathcal H \to \mathcal H$, where $D(A)$ is the domain of $A$ and $\mathcal H$ is a Hilbert space.
    \item[ii)]
    The observable's eigenvectors form an orthonormal basis spanning the state space in which that observable exists.
    Any quantum state can be represented as a superposition of the eigenstates of an observable.
   \item[iii)] The eigenstates of Hermitian operators have real eigenvalues.
  \end{enumerate}
Some examples of observables are:
the Hamiltonian operator which represents the total energy of the system, the momentum operator, the position operator.

The von Neumann measurement scheme describes measurements by assuming that  the measuring device is treated as a quantum object, see Chapter 7 in the book \cite{aharonov_book} or the original book by von Neumann \cite{vNBOOK}. Consider for example the Stern-Gerlach experiment, designed to measure the spin of a particle and which consists in sending a particle into a magnetic field ${B}$ which is non-homogeneous and varies, say, in the $z$ direction.

The Hamiltonian of the interaction of the spin ${S}$ of the particle with magnetic field ${B}$ is given by
$$
H_{\rm int}=-\mu {S} \cdot {B},
$$
where the Hamiltonian $H_{\rm int}$
is the potential energy of a magnetic dipole of momentum $\mu S$ in the magnetic field ${B}$.
When the particle passes through the field its momentum ${p}$ changes accordingly to the Heisenberg  equation
$$
\frac{d{p}}{dt}=\frac{i}{\hbar}[H_{\rm int},{p}]=\mu \nabla ({S} \cdot {B}).
$$

Assume now that ${B}$ is parallel to the $z$-axis, and its $z$-component is  $B_z$.
Then if the particle crosses the magnetic field in a time $T$, it acquires transverse momentum $\mu \frac{\partial B_z}{\partial z} S_zT$
proportional to the spin component $S_z$.
 A beam of particles entering the Stern-Gerlach apparatus therefore splits into beams for each spin component $S_z$.
 In this experiment the measurement interaction lasts a limited time $T$.
It produces a change, namely the deflection of the particle, that corresponds to the value of the
observable $S_z$. At all other times the particle and
Stern-Gerlach apparatus are distinct, independent systems.
Note that the measurement does not change the measured observable $S_z$ and so,
 in principle, the interaction time $T$ can be very small if  $\frac{\partial B_z}{\partial z}$ is very large. Indeed, it is important for the measurement to last a very short time, since otherwise the observable might change during the process.
Note also that the measurement is a quantum process, because we wrote
the Hamiltonian $H_{\rm int}$ for the measurement interaction and represented the measuring device as itself a quantum system.

\bigskip\noindent This example is a paradigmatic situation for all quantum measurements which can be described by the following properties:
\begin{itemize}
\item[i)]
 The measurement interaction lasts a limited time $T$.
\item[ii)]
The measurement produces a change (in the previous example the deflection of the particle) that corresponds to the value of the
observable.
\item[iii)]
The measurement  does not change the measured observable.
\item[iv)]
 The interaction time $T$ can be very small.
\item[v)] The measurement is a quantum process.
\end{itemize}

The criteria for quantum measurement i) to v) are collectively called the  von
Neumann model  for the measurement of an observable ${A}$ which, in the above example, is the spin $S_z$.

To satisfy v), we treat the measurement via an interaction Hamiltonian $H_{\rm int}$.
 According to iii), $H_{\rm int}$  and ${A}$ must commute. According
to ii), $H_{\rm int}$ must couple ${A}$ to something that yields an observable change,
like the deflection
of the particle in the measurement of $S_z$. The simplest coupling we can write is ${A}{P}_{\mathrm{md}}$,
where ${P}_{\mathrm{md}}$ is an otherwise independent observable, and the subscript $"\mathrm{md}"$ identifies this observable as being a measuring device. Since i) requires $H_{\rm int}$ to be effective only
during the measurement, we multiply ${A}{P}_{\mathrm{md}}$ by a coupling $g(t)$ that is different from zero only
in an interval $0 \leq t \leq T$, and with
$$
\int_0^Tg(t)dt =g_0
$$
and, finally, iv) implies that we can consider the limit $T\to 0$. In this limit the measurement is termed
impulsive. Thus our interaction Hamiltonian is
\begin{equation}\label{HintT}
H_{\rm int}(t) = g(t){A}{P}_{\mathrm{md}} ,
\end{equation}
and the total Hamiltonian, which includes the separate free Hamiltonians of the measuring device $H_{\mathrm{md}}$ and of the measured system
$H_s$, is
$$
H = H_{\mathrm{md}} + H_{\rm s} + H_{\rm int}.
$$



\section{Weak values and weak measurements - the main idea}
\label{abl-main-idea}
In order to more deeply appreciate the new quantum mechanical features and effects arising out of the notion of weak value (and as a consequence, superoscillations), we need to first take a step back and recall the history of ``time" in quantum mechanics. This section and the next two are largely taken from \cite{townes}.\\
The ``time-asymmetry" attributed to the standard formulation of Quantum Mechanics was inherited from classical mechanics where one can predict the future based on initial conditions: once the equations of motion are fixed in classical mechanics, then the initial and
final conditions are not independent, only one can be fixed arbitrarily.
In contrast, as a result of the uncertainty principle, the relationship between
initial and final conditions within Quantum Mechanics can be one-to-many: two ``identical" particles, prepared in exactly the same way, with identical environments
can subsequently exhibit different properties even when they are both subjected to completely identical measurements.
These subsequent identical measurements provide fundamentally new information about the system which could not in principle be obtained from the initial conditions.

Quantum Mechanics's ``time-asymmetry" is the assumption that measurements only have consequences {after} they are performed, i.e. towards the future.  Nevertheless, Aharonov, Bergmann and Lebowitz showed in ~\cite{abl} that the new information obtained from measurements is also relevant for the {past} of every quantum-system and not just the future.  This inspired the authors to reformulate Quantum Mechanics in terms of {\em Pre-and-Post-Selected ensembles} (PPS).
The traditional paradigm for ensembles is to simply prepare systems in a particular state and thereafter subject them to a variety of experiments.  For pre- and-post-selected-ensembles, we add one more step, a subsequent measurement or post-selection.
By collecting only a subset of the outcomes for this later measurement, we see that the ``pre-selected-only-ensemble"   can be divided into sub-ensembles according to the results of this subsequent ``post-selection-measurement."
Because pre- and post-selected ensembles are the most refined quantum ensemble, they are of fundamental importance and subsequently led to the {\em two-vector} or {\em Time-Symmetric reformulation of Quantum Mechanics} (TSQM)~\cite{av}, \cite{jmav}.  TSQM  provides a complete description of a quantum-system at a given moment by using two-wavefunctions, one evolving from the past towards the future (the one utilized in the standard paradigm) and a second one, evolving from the future towards the past.

In TSQM measurements occur at the present time $t$
 while the state is known both  at $t_{\mathrm{in}}<t$ (past) and at $t_{\mathrm{fin}}>t$ (future).
 To be more precise, we start at $t=t_{\mathrm{in}}$ with a  measurement of  a nondegenerate operator ${O}_{\mathrm{in}}$.   This gives as one potential outcome the state $\ket{\Psi_{\mathrm{in}}}$. In other words, we prepared the ``pre-selected" state
$\ket{\Psi_{\mathrm{in}}}$.  Then we consider a later time $t_{\mathrm{fin}}$, and we perform another measurement of a nondegenerate
operator ${O}_{\mathrm{fin}}$ which yields one possible outcome: the post-selected state $\ket{\Psi_{\mathrm{fin}}}$.
 At an intermediate time $t\in [t_{\mathrm{in}},t_{\mathrm{fin}}]$, we measure a nondegenerate observable ${A}$ (for simplicity), with eigenvectors $\{ \ket{a_j} \}$. Our goal is to
determine the conditional probability of $a_j$, given that we know both boundary conditions,
$\ket{\Psi_{\mathrm{in}}}$ and $\ket{\Psi_{\mathrm{fin}}}$.

To this purpose, we use the time displacement operator:
   $$
   U_{t_{\mathrm{in}}\rightarrow t}=\exp\{-iH(t-t_{\mathrm{in}})\}
   $$
   where $H$ is the Hamiltonian for the free system.  For simplicity, we assume that $H$ is time
   independent and we set $\hbar=1$.
The standard
theory of Von Neumann measurements states that the system collapses into an eigenstate $\ket{a_j}$
{\it after} the  measurement at $t$ with an amplitude $\textcolor{BlueViolet}{\bra{a_j} U_{t_{\mathrm{in}}\rightarrow t}\ket{\Psi_{\mathrm{in}}}}$.
The amplitude for our series of events is defined as
$$
\alpha_j = \textcolor{BlueViolet}{\bra{\Psi_{\mathrm{fin}}}U_{t\rightarrow t_{\mathrm{fin}}}|a_j\rangle\langle a_j|U_{t_{\mathrm{in}}\rightarrow t}\ket{\Psi_{\mathrm{in}}}}
$$
and is illustrated in Figure \ref{timerev}.a.
This translates into the fact that the conditional probability to measure $a_j$ (given that $\ket{\Psi_{\mathrm{in}}}$ is pre-selected
   and $\ket{\Psi_{\mathrm{fin}}}$ will be post-selected) is given by the so-called ABL formula, see ~\cite{abl}:
\begin{equation}
{\rm Prob}(a_j,t|\Psi_{\mathrm{in}},t_{\mathrm{in}}; \Psi_{\mathrm{fin}},t_{\mathrm{fin}})  =\textcolor{BlueViolet}{\frac{ |\bra{\Psi_{\mathrm{fin}}} U_{t\rightarrow t_{\mathrm{fin}}}\ket{a_j}\bra{a_j} U_{t_{\mathrm{in}}\rightarrow t}\ket{\Psi_{\mathrm{in}}}|^2 }{\sum_{n} |\bra{\Psi_{\mathrm{fin}}} U_{t\rightarrow t_{\mathrm{fin}}}\ket{a_n}\bra{a_n} U_{t_{\mathrm{in}}\rightarrow t}\ket{\Psi_{\mathrm{in}}}|^2}}.
\label{ablnts0}
\end{equation}
\begin{remark}{\rm The ABL formula is intuitive:  $|\bra{a_j} U_{t_{\mathrm{in}}\rightarrow t}\ket{\Psi_{\mathrm{in}}}|^2$ is the probability to obtain $\ket{a_j}$ having started with $\ket{\Psi_{\mathrm{in}}}$.  If $\ket{a_j}$ was obtained, then the system collapsed to $\ket{a_j}$ and $|\bra{\Psi_{\mathrm{fin}}} U_{t\rightarrow t_{\mathrm{fin}}}\ket{a_j}|^2$ is then the probability to obtain $\ket{\Psi_{\mathrm{fin}}}$.  The probability to obtain $\ket{a_j}$ and $\ket{\Psi_{\mathrm{fin}}}$ then is $|\alpha_j|^2$.  This is not yet the conditional probability since the post-selection may yield outcomes other than $\bra{\Psi_{\mathrm{fin}}}$. The probability to obtain $\ket{\Psi_{\mathrm{fin}}}$ is
$$\sum_j|\alpha_j|^2=|\bra{\Psi_{\mathrm{fin}}}\Psi_{\mathrm{in}}\ra|^2<1.$$
The question being investigated concerning probabilities of $a_j$ at $t$ assumes we are successful in obtaining
 the post-selection and therefore requires the denominator in (\ref{ablnts0}), $\sum_j|\alpha_j|^2$, which is a
 re-normalization to obtain a proper probability.
}
\end{remark}

As a first step toward understanding the underlying time-symmetry in the ABL formula, we consider the time-reverse of the numerator of (\ref{ablnts0}) and, as a consequence, the time reverse of Figure \ref{timerev}.a.   Firstly,  we
apply $U_{t\rightarrow t_{\mathrm{fin}}}$ on $\bra{\Psi_{\mathrm{fin}}}$ instead of on $\bra{a_j}$.  We then note that $$\bra{\Psi_{\mathrm{fin}}}U_{t\rightarrow t_{\mathrm{fin}}}|=\langle U_{t\rightarrow t_{\mathrm{fin}}}^\dag\Psi_{\mathrm{fin}}|$$ by using the well known Quantum Mechanics symmetry
\[
\begin{split}
&U_{t\rightarrow
      t_{\mathrm{fin}}}^\dag={\left\{e^{-iH(t_{\mathrm{fin}}-t)}\right\}}^\dag\\
      &=e^{iH(t_{\mathrm{fin}}-t)}=e^{-iH(t-t_{\mathrm{fin}})}\\
      &=U_{t_{\mathrm{fin}}\rightarrow t}.
         \end{split}
      \]
      We also apply $U_{t_{\mathrm{in}}\rightarrow t}$ on $\bra{a_j}$ instead of on  $\ket{\Psi_{\mathrm{in}}}$ which yields the time-reverse reformulation of the numerator of (\ref{ablnts0}) $$\textcolor{RedViolet}{\langle U_{t_{\mathrm{fin}}\rightarrow t}\Psi_{\mathrm{fin}}|a_j\rangle\langle U_{t\rightarrow t_{\mathrm{in}}} a_j|\Psi_{\mathrm{in}}\rangle}$$ as illustrated in Figure \ref{timerev}.b.

\begin{figure}[h!]
\centering
  \includegraphics[width=5in,height=5.9in,keepaspectratio]{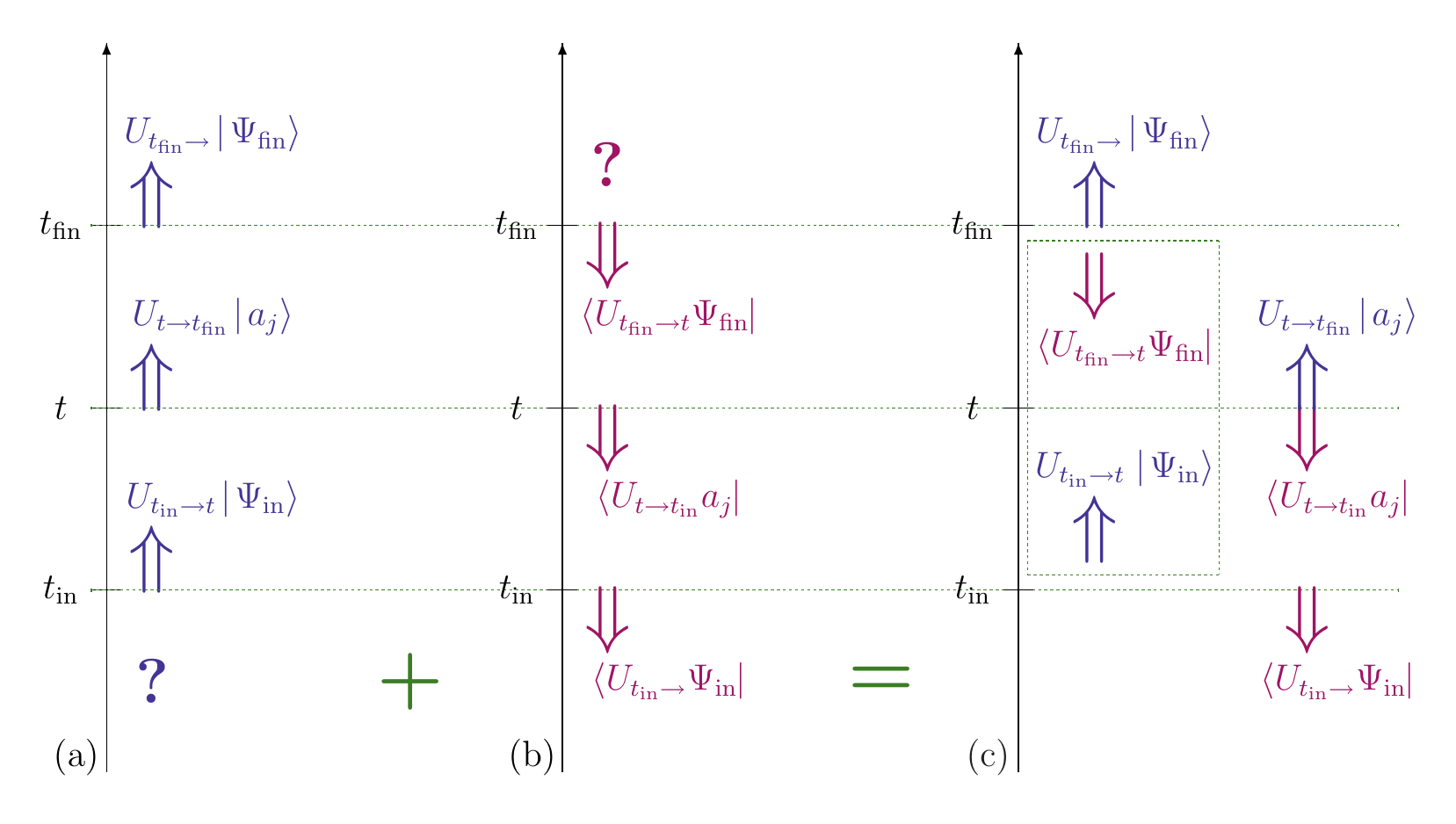}
\caption[Time-reversal symmetry in probability amplitudes. From \cite{townes}]
{\small Time-reversal symmetry in probability amplitudes. From \cite{townes}.}
\label{timerev}
\end{figure}

 To formulate what we mean by the two-vector in TSQM more work is needed.  For example, if
we want to compute  the probability for possible outcomes of $a_j$ at time $t$, we must consider
both $\textcolor{BlueViolet}{| U_{t_{\mathrm{in}}\rightarrow t}\mid\!\Psi_{\mathrm{in}}\rangle}$ and $\textcolor{RedViolet}{\la U_{t_{\mathrm{fin}}\rightarrow t}\Psi_{\mathrm{fin}}\!\!\mid}$. In fact, these expressions propagate the pre- and post-selection to the present time $t$ (see the conjunction of both figures \ref{timerev}.a and \ref{timerev}.b giving \ref{timerev}.c; these two-vector are not just the time-reverse of each other).  This represents the basic idea behind the  Time-Symmetric reformulation of Quantum Mechanics and gives
\begin{equation}
{\rm Prob}(a_j,t|\Psi_{\mathrm{in}},t_{\mathrm{in}}; \Psi_{\mathrm{fin}},t_{\mathrm{fin}})  =\frac{ |\textcolor{RedViolet}{\langle U_{t_{\mathrm{fin}}\rightarrow t}\Psi_{\mathrm{fin}}}\ket{a_j}\bra{a_j} \textcolor{BlueViolet}{U_{t_{\mathrm{in}}\rightarrow t}\ket{\Psi_{\mathrm{in}}}}|^2 }{\sum_{n} |\textcolor{RedViolet}{\langle U_{t_{\mathrm{fin}}\rightarrow t}\Psi_{\mathrm{fin}}}\ket{a_n}\bra{a_n} \textcolor{BlueViolet}{U_{t_{\mathrm{in}}\rightarrow t}\ket{\Psi_{\mathrm{in}}}}|^2}.
\label{ablnts}
\end{equation}
While this mathematical manipulation clearly proves that time symmetric reformulation of quantum mechanics is consistent with the standard approach to Quantum Mechanics, it leads to a very different interpretation.  To give an example, the action of $U_{t_{\mathrm{fin}}\rightarrow t}$ on $\bra{\Psi_{\mathrm{fin}}}$ (i.e.
$\textcolor{RedViolet}{\langle U_{t_{\mathrm{fin}}\rightarrow t}\Psi_{\mathrm{fin}}}|$)
 can be interpreted to mean that the time displacement operator $U_{t_{\mathrm{fin}}\rightarrow t}$ sends $\bra{\Psi_{\mathrm{fin}}}$ back in time from the time $t_{\mathrm{fin}}$ to the present, $t$.
A number of new categories of states
are suggested by the TSQM formalism and have proven useful in a large variety of situations.

One of the simplest, and yet interesting, examples of pre- and post-selection is to pre-select a spin-1/2 system with $\ket{\Psi_{\mathrm{in}}}=|{\sigma}_x=+1\ra=\vert\!\!\uparrow_x\!\rangle $ at time $t_{\mathrm{in}}$.
After the pre-selection, spin measurements in the direction perpendicular  to $x$ yields complete uncertainty in the result,
so if we post-select at time $t_{\mathrm{fin}}$ in the $y$-direction, we obtain $\ket{\Psi_{\mathrm{fin}}}=|{\sigma}_{\mathrm{y}}= +1\ra=\vert\uparrow_y\rangle$ one-half of the times.  Since the particle is free, the spin is conserved in time and thus for any $t\in [t_{\mathrm{in}},t_{\mathrm{fin}}]$, an ideal measurement of either ${\sigma}_x$ or ${\sigma}_y$, yields the value
   $+1$ for this pre- and post-selection.

The fact that two non-commuting observables are known with certainty is a most surprising property which no pre-selected only ensemble could possess.

We now ask another question, slightly  more complicated, about the spin in the
direction $\xi=45^{\circ}$ relative to the $x-y$ axis.  This yields:
\beq
{\sigma}_{\xi}={\sigma}_x \cos 45^{\circ} +{\sigma}_y\sin 45^{\circ}=\frac{{\sigma}_x +{\sigma}_y}{\sqrt{2}}.
\label{spin45}
\eeq
Since we know that ${\rm Prob}({\sigma}_{x}=+1)=1$ and ${\rm Prob}({\sigma}_{y}=+1)=1$, one might wonder why we could not  insert both values, ${\sigma}_{{x}}= +1$ and ${\sigma}_{{y}}= +1$ into (\ref{spin45}) and obtain ${\sigma}_{\xi}=\frac{1+1}{\sqrt{2}}=\frac{2}{\sqrt{2}}=\sqrt{2}$ (see figure \ref{2XorY}).

\ \  \vskip 1cm
\begin{figure}[h!]
\begin{picture}(400,90)(0,0)
\put(40,10){\line(1,0){40}}
\put(40,90){\line(1,0){40}}
\color{BlueViolet}
\put(60,10){\vector(0,1){20}}
\color{RedViolet}
\put(60,90){\vector(0,-1){20}}
\color{Red}
\put(40,15){\dashbox{1}(40,70)}
\put(100,15){\dashbox{1}(40,70)}
\put(240,15){\dashbox{1}(140,70)}
\color{OliveGreen}

\color{Black}
\put(100,10){\line(1,0){40}}
\put(100,90){\line(1,0){40}}
\color{BlueViolet}
\put(120,10){\vector(0,1){20}}
\color{RedViolet}
\put(120,90){\vector(0,-1){20}}
\color{Black}
\put(240,10){\line(1,0){40}}
\put(240,90){\line(1,0){40}}
\color{BlueViolet}

\put(260,10){\vector(0,1){20}}
\color{RedViolet}

\put(260,90){\vector(0,-1){20}}
\color{OliveGreen}

\color{BlueViolet}
\put(260,0){\makebox(0,0){$|{\sigma}_x=1\ra$}}
\put(120,0){\makebox(0,0){$|{\sigma}_x=1\ra$}}
\put(60,0){\makebox(0,0){$|{\sigma}_x=1\ra$}}

\color{RedViolet}
\put(260,100){\makebox(0,0){$\la{\sigma}_y=1|$}}
\put(120,100){\makebox(0,0){$\la{\sigma}_y=1|$}}
\put(60,100){\makebox(0,0){$\la{\sigma}_y=1|$}}
\color{Black}


\color{Black}
\put(30,10){\makebox(0,0){$t_{\mathrm{in}}$}}
\put(30,90){\makebox(0,0){$t_{\mathrm{fin}}$}}
\put(30,50){\makebox(0,0){$t$}}
\put(90,120){\makebox(0,0){(a)}}
\put(290,120){\makebox(0,0){(b)}}
\put(60,50){\makebox(0,0){${\sigma}_x\!=\!1$}}
\put(310,50){\makebox(0,0){${\sigma}_{45^{\circ}}=\frac{{\sigma}_x +{\sigma}_y}{\sqrt{2}}=\frac{1+1}{\sqrt{2}}=\sqrt{2}?$}}


\put(120,50){\makebox(0,0){${\sigma}_y\!=\!1$}}

\color{OliveGreen}

\end{picture}

\caption{\footnotesize A spin-1/2 particle is pre-selected at time $t_{\mathrm{in}}$ to be $|{\sigma}_x=1\ra$, and post-selected at $t_{\mathrm{fin}}$ to be $\la{\sigma}_y=1|$. (a) During the intermediate time $t\in [t_{\mathrm{in}},t_{\mathrm{fin}}]$, ABL formula gives that an ideal measurement of either ${\sigma}_x$ or ${\sigma}_y$ yields
   $+1$ with certainty, suggesting that such a particle has well defined values of the two noncommuting spin components.
 (b) It would seem to follow that
the spin component ${\sigma}_{45^{\circ}}$ would
have to be $\sqrt{2}$ which is not an allowed eigenvalue.}
\label{2XorY}
\end{figure}
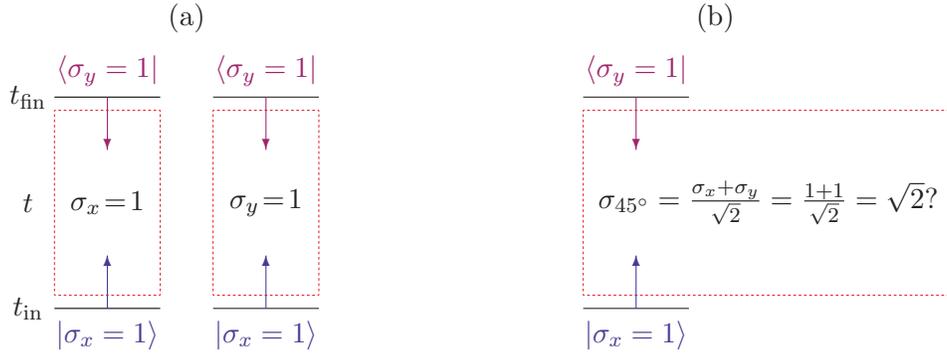

But such a result  cannot be correct for an ideal measurement, in fact the eigenvalues of any spin operator, including ${\sigma}_{\xi}$, must be $\pm 1$.
The inconsistency can also be seen by noting that $${\left(\frac{\sigma_x+\sigma_y}{\sqrt{2}}\right)}^2=\frac{\sigma_x^2+\sigma_y^2+\sigma_x\sigma_y+\sigma_y\sigma_x}{2}=\frac{1+1+0}{2}=1.$$
 By the previous argument, we would instead expect $${\left(\frac{\sigma_x+\sigma_y}{\sqrt{2}}\right)}^2={\left(\frac{1+1}{\sqrt{2}}\right)}^2=2\neq 1.$$
  The replacement of ${\sigma}_{{x}}= +1$ {and} ${\sigma}_{{y}}= +1$ in (\ref{spin45}) can only be done if ${\sigma}_{{x}}$ and ${\sigma}_{{y}}$ commute, which would allow both values  simultaneously to be definite.  When both statements are combined together simultaneously (as we attempted to do in suggesting that ${\sigma}_{45^{\circ}}$ might equal $\frac{1+1}{\sqrt{2}}=\frac{2}{\sqrt{2}}=\sqrt{2}$),  then the statements, namely $P({\sigma}_{x}=+1)=1$ and ${\rm Prob}({\sigma}_{y}=+1)=1$, are said to be ``counter-factuals", see \cite{at2}.

Although it appears we have reached the end-of-the-line with this argument,
nevertheless, it still seems that there should be some sense in which
both  $${\rm Prob}({\sigma}_{x}=+1)=1 \qquad{\rm and} \qquad{\rm Prob}({\sigma}_{y}=+1)=1$$ manifest themselves simultaneously to produce ${\sigma}_{\xi}=\sqrt{2}$.

\section{Weak values and weak measurements - mathematical aspects}
\label{weakennopost}
  Let $\la{A}\ra = \ave{\Psi}{{A}}$, and $|\Psi\rangle$ be any vector in a Hilbert space. Set $$\Delta A^2 = \ave{\Psi}{({A} - \la{A}\ra)^2},$$ and let $\ket{\Psi_\perp}$ be a state such that $\amp {\Psi}{\Psi_\perp} = 0$.
Then we have:
 \begin{theorem} For every observable $A$ and a normalized state $|\Psi\rangle$, the formula
 \begin{equation}\label{thm1}
 A|\Psi\rangle=\langle  A\rangle|\Psi\rangle+\Delta A|\Psi_\perp\rangle
 \end{equation}
 holds for some state $|\Psi_\perp\rangle$ which is orthogonal to $|\Psi\rangle$.
   \end{theorem}
   \begin{proof}
To prove the statement, we write
$$A|\Psi\rangle=\langle A\rangle|\Psi\rangle+A|\Psi\rangle-\langle A\rangle|\Psi\rangle$$
   now, we set:
$|\widetilde{\Psi}_\perp\rangle =A|\Psi\rangle-\langle A\rangle|\Psi\rangle$,   so:
$$\langle \widetilde{\Psi}_\perp |\Psi\rangle =(\langle\Psi|A-\langle\Psi|\langle
      A\rangle)|\Psi\rangle=\langle\Psi| A|\Psi\rangle - \langle A\rangle\langle\Psi|\Psi\rangle = 0.$$
Now we set $|\Psi_\perp\rangle=b|\widetilde{\Psi}_\perp\rangle$, where $|\Psi_\perp\rangle$ is normalized and $b$
   real (note that $\langle\Psi|\Psi_\perp\rangle=0$).
   So we have
$A|\Psi\rangle=\langle A\rangle|\Psi\rangle+b|\Psi_\perp\rangle$.
   Now we multiply from the left by $\langle\Psi_\perp|$, and we get:
   $\langle\Psi_\perp|A|\Psi\rangle=b$.   Now we can see that:
\[
\begin{split}
\langle\Psi|A^2|\Psi\rangle&=\langle\Psi|A(\langle A\rangle|\Psi\rangle+b|\Psi_\perp\rangle)\\
&=  \langle\Psi|({\langle A\rangle}^2|\Psi\rangle+b\langle A\rangle|\Psi_\perp\rangle+bA|\Psi_\perp\rangle)\\
&= {\langle A\rangle}^2+b\langle\Psi|A|\Psi_\perp\rangle
   \end{split}     \]
        so $$\langle A^2\rangle-{\langle A\rangle}^2=b\langle\Psi|A|\Psi_\perp\rangle=b^2$$
   which means that $$b=\sqrt{\langle A^2\rangle-{\langle A\rangle}^2}=\Delta A$$
   and the result
   $$A|\Psi\rangle=\langle  A\rangle|\Psi\rangle+\Delta
   A|\Psi_\perp\rangle$$  is proved.
\end{proof}
We now use this result to show how to perform measurements which do not disturb either the pre- or post-selections. The interaction $H_{\mathrm{int}}\!=\!-g(t){Q}_{\mathrm{md}}{A}$ is weakened
 by minimizing $g_0 \Delta Q_{\mathrm{md}}$.  For simplicity, we consider $g_0\ll 1$ (assuming without lack of generality that the state of the measuring device is a Gaussian with spreads $\Delta P_{\mathrm{md}}\!=\!\Delta Q_{\mathrm{md}}\!=\!1$).   We may then set $e^{ -i g_0  {Q}_{\mathrm{md}} {A}
 }\!\approx\! 1-ig_0 {Q}_{\mathrm{md}} {A}$ and use \eqref{thm1}.
This shows that before the post-selection, the system state is:
   \begin{equation}
   \begin{split}
      \exp({ -i  g_0{Q}_{\mathrm{md}} {A}})|\Psi_{\mathrm{in}}\ra\!&=\!
(1\!-\!ig_0 {Q}_{\mathrm{md}}{A})|\Psi_{\mathrm{in}}\rangle\!\\
&=\!(1\!-i\!g_0 {Q}_{\mathrm{md}}\langle{A}\rangle)|\Psi_{\mathrm{in}}\rangle\!-i\!g_0 {Q}_{\mathrm{md}}\Delta{A}|\Psi_{\mathrm{in}\perp}\rangle .
\end{split}
   \end{equation}

\noindent Computing the norm of this state
$$
{\parallel (1-ig_0 {Q}_{\mathrm{md}}{A})|\Psi_{\mathrm{in}}\rangle
      \parallel}^2=1+{g_0 ^2{Q}_{\mathrm{md}}^2}\langle {A}^2\rangle,
      $$
 the probability to leave $|\Psi_{\mathrm{in}}\rangle$ unchanged after the measurement is:
   \begin{equation}
      \frac{1+{g_0 ^2{Q}_{\mathrm{md}}^2}{\langle{A}\rangle}^2}
   {1+{g_0 ^2{Q}_{\mathrm{md}}^2}\langle
   {A}^2\rangle}\longrightarrow 1\,\,\,\,\,\,{\rm when\ }g_0 \rightarrow 0,
   \end{equation}
while the probability to disturb the state (i.e. to obtain $|\Psi_{\mathrm{in}\perp}\rangle$) is:
   \begin{equation}
      \frac{{g_0 ^2{Q}_{\mathrm{md}}^2}{\Delta{A}}^2}
   {1+{g_0 ^2{Q}_{\mathrm{md}}^2}\langle
   {A}^2\rangle}\longrightarrow 0\,\,\,\,\,\,{\rm when\ }g_0 \rightarrow 0.
\label{collprob}
   \end{equation}
The final state of the measuring device is now a superposition of many substantially overlapping Gaussians with probability distribution given by ${\rm Prob}(P_{\mathrm{md}})=\sum_i |\la a_i|\Psi_{\mathrm{in}}\ra|^2 \exp\left\{{-\frac{(P_{\mathrm{md}}-g_0 a_i)^{2}} {2\Delta P_{\mathrm{md}}^{2}}}\right\} $.
This sum  can be approximated by a single Gaussian
 $$\tilde{\Phi}^{\mathrm{fin}}_{\mathrm{md}}(P_{\mathrm{md}})\approx\langle P_{\mathrm{md}}|e^{-ig_0 {Q}_{\mathrm{md}}\la{A}\ra}|\Phi^{\mathrm{in}}_{\mathrm{md}}\rangle\approx\exp\left\{-{{(P_{\mathrm{md}}-g_0\la
 A\ra)^2}\over{\Delta P_{\mathrm{md}}^2}}\right\}$$ centered at $g_0\la{A}\ra$.

Formula (\ref{collprob}) shows that the probability for a collapse  decreases as $O(g_0 ^2)$, but the measuring device's shift grows as $O(g_0)$, so $\delta P_{\mathrm{md}}=g_0 a_i$~\cite{spie-nswm}.
For a sufficiently weak interaction (e.g. $g_0\ll 1$), the probability for a collapse can be made arbitrarily small, while the measurement still yields information but becomes less precise because
the shift in the measuring device is much smaller than its uncertainty $\delta P_{\mathrm{md}}\ll\Delta P_{\mathrm{md}}$.
If we perform this measurement on a single particle, then
two non-orthogonal states will be indistinguishable.
If this were possible, it would violate unitarity because these states could time evolve into orthogonal states $|\Psi_1\ra|\Phi^{\mathrm{in}}_{\mathrm{md}}\ra \longrightarrow |\Psi_1\ra|\Phi^{\mathrm{in}}_{\mathrm{md}}(1)\ra$ and
$|\Psi_2\ra|\Phi^{\mathrm{in}}_{\mathrm{md}}\ra \longrightarrow |\Psi_2\ra|\Phi^{\mathrm{in}}_{\mathrm{md}}(2)\ra$, with $|\Psi_1\ra|\Phi^{\mathrm{in}}_{\mathrm{md}}(1)\ra$ orthogonal to $|\Psi_2\ra|\Phi^{\mathrm{in}}_{\mathrm{md}}(2)\ra$.  With weakened measurement interactions, this does not happen because the measurement of these two non-orthogonal states causes a  shift in the measuring device smaller than its uncertainty.  We conclude that the shift $\delta P_{\mathrm{md}}$ of the measuring device is  a measurement error because $ \tilde{\Phi}^{\mathrm{fin}}_{\mathrm{md}}(P_{\mathrm{md}})=\langle
      P_{\mathrm{md}}-g_0 \la{A}\ra|\Phi^{\mathrm{in}}_{\mathrm{md}}\rangle\approx\langle P_{\mathrm{md}}|\Phi^{\mathrm{in}}_{\mathrm{md}}\rangle$ for $g_0 \ll 1$.
Nevertheless, if a large ($N\geq\frac{N'}{g_0 }$) ensemble of  particles is used, then the shift of all the measuring devices ($\delta P^{tot}_{\mathrm{md}}\approx g_0 \la{A}\ra\frac{N'}{g_0}=N'\la{A}\ra$) becomes distinguishable (because of repeated integrations), while  the collapse probability   still goes to zero.
That is, for a large ensemble of particles which are all either $|\Psi_2\ra$ or
   $|\Psi_1\ra$, this measurement can distinguish between them even if $|\Psi_2\ra$ and
   $|\Psi_1\ra$ are not orthogonal, because the scalar product $\langle\Psi_1\upn|\Psi_2\upn\rangle=\cos^N\theta\longrightarrow 0$.

The fact of having a new measurement paradigm, namely information gain without disturbance, is fruitful to inquire whether this type of measurement reveals new values or new properties.
With weak measurements (which involve adding a post-selection to this ordinary, but weakened, von Neumann measurement),  the measuring device registers a new value, the so-called weak value.
As an indication of this, we insert a complete set of states $\{ \ket{\Psi_{\mathrm{fin}}}_j \}$ into the outcome of the weak interaction and we calculate the expectation value as follows
\beq
 \la{A}\ra =  \bra{\Psi_{\mathrm{in}}} { \left[\sum_j  \ket{\Psi_{\mathrm{fin}}}_j\bra{\Psi_{\mathrm{fin}}}_j\right]{A}} \ket{\Psi_{\mathrm{in}}}
= \sum_j |\langle \Psi _{\mathrm{fin}} \!\mid_j \!\Psi _{\mathrm{in}}
\rangle|^2\
{ {\langle \Psi _{\mathrm{fin}}\! \mid_j {A} \mid \!\Psi _{\mathrm{in}}
\rangle} \over {\langle \Psi _{\mathrm{fin}} \!\mid_j \!\Psi _{\mathrm{in}}
\rangle}}.
\label{expweak}
\eeq
If we interpret the states $ \ket{\Psi_{\mathrm{fin}}}_j $ as the outcomes of a final ideal measurement on the system (i.e. a post-selection)
then performing  a weak measurement (e.g. with $ g_0\Delta Q_{\mathrm{md}}\rightarrow 0$) during the intermediate time $t\in [t_{\mathrm{in}},t_{\mathrm{fin}}]$, provides the coefficients for $|\langle \Psi_{\mathrm{fin}}\! |_j \Psi_{\mathrm{in}}
\rangle |^2$ which gives  the probabilities ${\rm Prob}(j)$
for obtaining a pre-selection of $\bra{\Psi_{\mathrm{in}}}$ and a post-selection of  $ \ket{\Psi_{\mathrm{fin}}}_j $.  The intermediate weak measurement does not disturb these states
and the  quantity
$$A_{{w}}(j) \equiv { {\langle \Psi _{\mathrm{fin}} \!\mid_j {A} \mid \!\Psi _{\mathrm{in}}
\rangle} \over {\langle \Psi _{\mathrm{fin}} \!\mid_j \!\Psi _{\mathrm{in}}
\rangle}}$$
  will be defined as the weak value of ${A}$ given a particular final post-selection $\langle \Psi _{\mathrm{fin}} \!\mid_j$, see Definition \ref{weakval}.
Thus, from the formula
$\la{A}\ra = \sum_j {\rm Prob}(j)\,  A_{{w}}(j)$,
 one can think of  $\la{A}\ra$ for the whole ensemble as being constructed out of sub-ensembles of pre- and post-selected-states in which the weak value is multiplied by a probability for a post-selected-state.

The weak value arises naturally from a weakened measurement with post-selection. In fact, let us  take
 $ g_0 \ll 1$; then the final state of measuring device in the momentum representation becomes
\begin{equation}\label{post_selected}
\begin{split}
\bra{P_{\mathrm{md}}} \bra{\Psi_{\mathrm{fin}}}e^{ -i  g_0  {Q}_{\mathrm{md}} {A}
 }&\ket{\Psi_{\mathrm{in}}}\ket{\Phi_{\mathrm{in}}^{\mathrm{md}}} \approx
\bra{P_{\mathrm{md}}} \bra{\Psi_{\mathrm{fin}}}1+i g_0 {Q}_{\mathrm{md}} {A}\ket{\Psi_{\mathrm{in}}}\ket{\Phi_{\mathrm{in}}^{\mathrm{md}}}\\
&\approx \bra{P_{\mathrm{md}}}\langle\Psi_{\mathrm{fin}}\!\mid
\Psi_{\mathrm{in}} \rangle \lbrace 1+i g_0 {Q} \weakv {\Psi_\mathrm{fin}}{{A}}{\Psi_\mathrm{in} }\rbrace\ket{\Phi_{\mathrm{in}}^{\mathrm{md}}}\\
&\approx  \langle\Psi_{\mathrm{fin}}\ket{\Psi_{\mathrm{in}}}
\bra{P_{\mathrm{md}}}
e^{ -i  g_0  {Q}A_{{w}}
 }\ket{\Phi_{\mathrm{in}}^{\mathrm{md}}}\\
&\rightarrow  \langle\Psi_{\mathrm{fin}}\ket{\Psi_{\mathrm{in}}}\exp\left\{{-{{(P_{\mathrm{md}}- g_0 \,
 A_{{w}})^2}
}}\right\},\\
\end{split}
 \end{equation}
where $A_w$ is as in the following definition.

\begin{definition}\label{weakval} Let $A$ be a Hermitian operator and let $ |\Psi_{\mathrm{in}}\rangle$,  $|\Psi_{\mathrm{fin}}\rangle$ denote a PPS.
We call
$$
A_{{w}}=\weakv {\Psi_\mathrm{fin}}{{A}}{\Psi_\mathrm{in} }.
$$
the {\em weak value} of $A$ on the PPS.
The weak value will also be denoted by $\langle A\rangle_w$.
\end{definition}
The final state
 of the measuring device is almost unentangled with the
 system; it is shifted by a very unusual quantity, the weak value, $A_{{w}}$,
which is not, in general, an eigenvalue of ${A}$.
\bigskip
From the definition of weak value it immediately follows that
\begin{theorem} For any pairs of Hermitian  operators
$$
\langle A + B \rangle_w=\langle A\rangle_w+\langle B\rangle_w.
$$
\end{theorem}
\begin{proof}
From the linearity of $A$ and $B$ it immediately follows that
$$
{ {\langle \Psi _{\mathrm{fin}} \mid A+B \mid \!\Psi _{\mathrm{in}}
\rangle} \over {\langle \Psi _{\mathrm{fin}} \mid \!\Psi _{\mathrm{in}}
\rangle}}= { {\langle \Psi _{\mathrm{fin}} \mid A \mid \!\Psi _{\mathrm{in}}
\rangle} \over {\langle \Psi _{\mathrm{fin}} \mid \!\Psi _{\mathrm{in}}
\rangle}}+ {{\langle \Psi _{\mathrm{fin}} \mid B \mid \!\Psi _{\mathrm{in}}
\rangle} \over {\langle \Psi _{\mathrm{fin}} \mid \!\Psi _{\mathrm{in}}
\rangle}}.
$$
\end{proof}
\begin{remark}{\rm
We are now interested in considering what happens when we consider the product of two observables, for example, $A^1 A^2$ and we apply them to product states of the form $|\Phi_1\rangle \, |\Phi_2\rangle$, (sometimes indicated as $|\Phi_1\rangle \otimes | \Phi_2\rangle$). We note that
$$
|A^1 A^2 |\Phi_1\rangle \, |\Phi_2\rangle =| A^1 \Phi_1\rangle \, |  A^2  \Phi_2\rangle.
$$
As it is well known, the weak value of a product of observable is not, in general, the product of the weak values (see e.g. \cite{at2}), but the situation is different when we consider initial and final states which are product states. In this case we obviously have
\[
\begin{split}
\langle A^1A^2\rangle_w&=\frac{\langle \Psi_2|\langle\Psi_1| A^1A^2|\Phi_1\rangle \, |\Phi_2\rangle}{\langle \Psi_2|\langle\Psi_1|\Phi_1\rangle \, |\Phi_2\rangle}\\
 &= \frac{\langle \Psi_2|A^2|\Phi_2\rangle \langle \Psi_1| A^1|\Phi_1\rangle}{\langle \Psi_2|\Phi_2\rangle \langle \Psi_1|\Phi_1\rangle }\\
 &=\langle A^1\rangle_w\langle A^2\rangle_w.
\end{split}
\]
where the weak values for $A^1, A^2$ and $A^1 A^2$ are calculated with respect to different pairs of initial and final states.
}
\end{remark}

If we now consider two orthonormal bases of eigenstates ${\Phi}_i$, $\Psi_i$, the weak value can be rewritten as
\begin{equation}\label{for40}
\langle A\rangle_w = \frac{\sum_{i=1}^N \alpha_i \langle \Psi_i | {A} |\Phi_i\rangle}{\sum_{i=1}^N \alpha_i \langle \Psi_i |\Phi_i\rangle}.
\end{equation}
We now decompose the operator $ A$ in terms of the projector operators
$$
P_{A=a_i}= \sum_{j} |\Phi_{i,j}\rangle\, \langle \Phi_{i,j} | ,
$$
where $|\Phi_{i,j}\rangle$ is a complete set of eigenstates with eigenvalue $a_i$ and therefore the spectral decompositions of $A$ and $I$ are
\begin{equation}\label{specdec}
A=\sum_i a_iP_{A=a_i}, \qquad I=\sum_i P_{A=a_i}.
\end{equation}
Finally, with this notation, the probability of finding $a=a_n$ is given by
\begin{equation}\label{for23}
{\rm Prob}({A=a_n})=\frac{|\sum_{i=1}^N \alpha_i \langle \Psi_i| P_{A=a_n}| \Phi_i\rangle|^2 }{\sum_k|\sum_{i=1}^N \alpha_i \langle \Psi_i| P_{A=a_k}| \Phi_i\rangle|^2}.
\end{equation}

\begin{theorem} If a strong measurement of an observable $A$
  yields an outcome $a$ with probability one, then  the weak value $\langle A\rangle_w$ gives the same outcome $a$.
\end{theorem}
\begin{proof}
Consider the state $\sum_{i=1}^N \alpha_i \langle \Psi_i| P_{A=a}| \Phi_i\rangle $ be such that the probability of finding the result $A=a$ in a strong measurement of $A$ is one. Then by (\ref{for23}) we have
$$
{\rm Prob}({A=a})=\frac{|\sum_{i=1}^N \alpha_i \langle \Psi_i| P_{A=a}| \Phi_i\rangle|^2 }{\sum_k|\sum_{i=1}^N \alpha_i \langle \Psi_i| P_{A=a_k}| \Phi_i\rangle|^2}=1.
$$
This implies that
\begin{equation}\label{for45}
\sum_{i=1}^N \alpha_i \langle \Psi_i|P=a_k| \Phi_i\rangle =0, \qquad \forall k \ {\rm such\ that\ } a_k\not= a.
\end{equation}
Consider now the weak value of $A$ which can be calculated by using the spectral decompositions of $A$
and $I$, see \eqref{specdec}, and replacing them into (\ref{for40}). By using (\ref{for45}) this gives us
$$
\langle A\rangle_w = \frac{\sum_{i=1}^N \alpha_i \langle \Psi_i | {A} |\Phi_i\rangle}{\sum_{i=1}^N \alpha_i \langle \Psi_i |\Phi_i\rangle}=\frac{\sum_{i=1}^N \alpha_i \langle \Psi_i| \sum_{k}a_k P_{A=a_k}| \Phi_i\rangle }{\sum_{i=1}^N \alpha_i \langle \Psi_i|\sum_k P_{A=a_k}| \Phi_i\rangle}=a.
$$
This concludes the proof.
\end{proof}

\begin{theorem} If the weak value of a dichotomic operator equals one of its eigenvalues, then the outcome of a strong measurement of the operator is that same eigenvalue with probability one.
\end{theorem}
\begin{proof}

Assume that the operator $A$ has two distinct eigenvalues $a_1$, $a_2$ and assume that its weak value is $\langle A\rangle_w=a_1$. Then, by (\ref{for40}) we obtain that
$$
\langle A\rangle_w=\frac{\sum_{i=1}^N \alpha_i \langle \Psi_i | a_1 P_{A=a_1}| \Phi_i\rangle + \sum_{i=1}^N \alpha_i \langle \Psi_i | a_2 P_{A=a_2}| \Phi_i\rangle }{\sum_{i=1}^N \alpha_i \langle \Psi_i | P_{A=a_1}| \Phi_i\rangle + \sum_{i=1}^N \alpha_i \langle \Psi_i | P_{A=a_2}| \Phi_i\rangle }=a_1.
$$
An immediate computation shows that $\sum_{i=1}^N \alpha_i \langle \Psi_i | P_{A=a_2}| \Phi_i\rangle =0$.
Therefore the probability that when we strong measure $A$, we obtain $a_2$ is zero and so, by (\ref{for23}), the probability of the strong measurement of $A$ be $a_1$ is one.
\end{proof}

\bigskip

\begin{remark}{\rm  While the definition of weak value is inspired by physical considerations, it has been recently shown in \cite{cheon} that weak values can be defined in a very natural way without any recourse to a physical setting. Because of the importance of this notion in our book we think it is useful to recall the arguments given in \cite{cheon}. Let $A$ be an hermitian operator acting on $\mathbb C^n$. The idea of weak value consists in representing $A$ through two different orthonormal bases $|\Psi_j\rangle$ and $|\Phi_j\rangle$, $j=1,\ldots ,n$, such that $\langle \Phi_i|\Psi_j\rangle\not=0$ for all $i,j=1,\ldots, n$. Then the weak value of $A$ with respect to these bases is given by
$$
(A_{ij})_w=\frac{\langle \Phi_i|A|\Psi_j\rangle}{\langle \Phi_i|\Psi_j\rangle}.
$$
It is clear that once we have fixed an initial and final state of the system $\Psi_{\rm in}$, $\Psi_{\rm fin}$ we can express  $\Psi_{\rm in}$ in terms of the basis $\Psi_j$ and $\Psi_{\rm fin}$ in terms of the basis $\Phi_j$ and so the weak value $A_w$ can be written in terms of the components $(A_{ij})_w$.
}
\end{remark}

\bigskip

\section{Large weak values and superoscillations}
\label{spin100}
The weak value for the spin-1/2 system that we considered previously (and which was confirmed experimentally for an analogous observable, the polarization \cite{RSH}) is
$
({\sigma}_{\xi=45^\circ})_{{w}} = \sqrt{2},
$
in contrast with the expected eigenvalues $\pm 1$. Incidentally, we note that the weak values, even well outside the eigenvalue spectrum, can be obtained by post-selecting states which are more anti-parallel to the pre-selection: for example, if we post-select the $+1$ eigenstate of $(\cos\alpha)\sigma_x + (\sin\alpha)\sigma_z$, then $({\sigma}_z)_{{w}}= g_0\tan \frac{\alpha}{2}$, which gives arbitrarily large values such as spin-100.
 To obtain this result,  we post-select ${\sigma}_y=1$   instead of post-selecting ${\sigma}_x=1$; ${\sigma}_y=1$  will be satisfied in one-half the trials (see Figure \ref{multparticlesprepost}).\footnote{If a post-selection does not satisfy ${\sigma}_y=+1$, then that member of the sub-ensemble must be discarded.  This highlights a fundamental difference between pre- and post-selection due to the macrosopic arrow-of-time: in contrast to post-selection, if the pre-selection does not satisfy the criterion, then a subsequent unitary transformation can transform to the proper criterion.}

\begin{figure}[h]
\vskip 1cm
\begin{picture}(400,90)(0,0)
\put(90,10){\line(1,0){40}}
\put(90,90){\line(1,0){40}}
\color{BlueViolet}
\put(35,0){\makebox(0,0){all $\mid\!{\sigma}_x =+1\!\rangle$}}
\put(110,10){\vector(0,1){20}}
\color{RedViolet}
\put(110,90){\vector(0,-1){20}}
\put(35,113){\makebox(0,0){either}}
\put(35,100){\makebox(0,0){$\mid\!{\sigma}_y =+1\!\rangle$}}
\put(35,87){\makebox(0,0){or $\mid\!{\sigma}_y =-1\!\rangle$}}

\put(110,100){\makebox(0,0){$\mid\!{\sigma}_y =-1\!\rangle$}}

\put(170,100){\makebox(0,0){$\mid\!{\sigma}_y =+1\!\rangle$}}

\put(230,100){\makebox(0,0){$\mid\!{\sigma}_y =-1\!\rangle$}}

\put(370,100){\makebox(0,0){$\mid\!{\sigma}_y =+1\!\rangle$}}

\color{OliveGreen}

\put(170,50){\oval(70,120)}
\put(370,50){\oval(70,120)}

\color{Black}

\put(110,0){\makebox(0,0){particle 1}}
\put(150,10){\line(1,0){40}}
\put(150,90){\line(1,0){40}}
\color{BlueViolet}

\put(170,10){\vector(0,1){20}}
\color{RedViolet}

\put(170,90){\vector(0,-1){20}}

\color{Black}
\put(170,0){\makebox(0,0){particle 2}}
\put(210,10){\line(1,0){40}}
\put(210,90){\line(1,0){40}}
\color{BlueViolet}

\put(230,10){\vector(0,1){20}}
\color{RedViolet}

\put(230,90){\vector(0,-1){20}}

\color{Black}
\put(230,0){\makebox(0,0){particle 3}}

\put(260,50){\circle*{3}}
\put(280,50){\circle*{3}}
\put(300,50){\circle*{3}}
\put(320,50){\circle*{3}}
\put(350,10){\line(1,0){40}}
\put(350,90){\line(1,0){40}}
\color{BlueViolet}

\put(370,10){\vector(0,1){20}}
\color{RedViolet}

\put(370,90){\vector(0,-1){20}}
\color{Black}
\put(370,0){\makebox(0,0){particle N}}
\put(80,10){\makebox(0,0){$t_{in}$}}
\put(80,90){\makebox(0,0){$t_{fin}$}}
\put(35,55){\makebox(0,0){weak measurement}}
\put(35,45){\makebox(0,0){of ${\sigma}_{\xi=45^\circ}$ at time $t$}}

\color{Orange}
\bezier{500}(90,40)(105,60)(120,40)

\color{WildStrawberry}
\bezier{35}(100,40)(115,60)(130,40)

\color{Orange}
\bezier{500}(150,40)(165,60)(180,40)
\color{WildStrawberry}
\bezier{35}(160,40)(175,60)(190,40)

\color{Orange}
\bezier{500}(210,40)(225,60)(240,40)
\color{WildStrawberry}
\bezier{35}(220,40)(235,60)(250,40)

\color{Orange}
\bezier{500}(350,40)(365,60)(380,40)
\color{WildStrawberry}
\bezier{35}(360,40)(375,60)(390,40)

\end{picture}

\caption[Complete correlation between N particles]{Statistical weak measurement ensemble. From \cite{townes}.}
{\small }
\label{multparticlesprepost}
\end{figure}
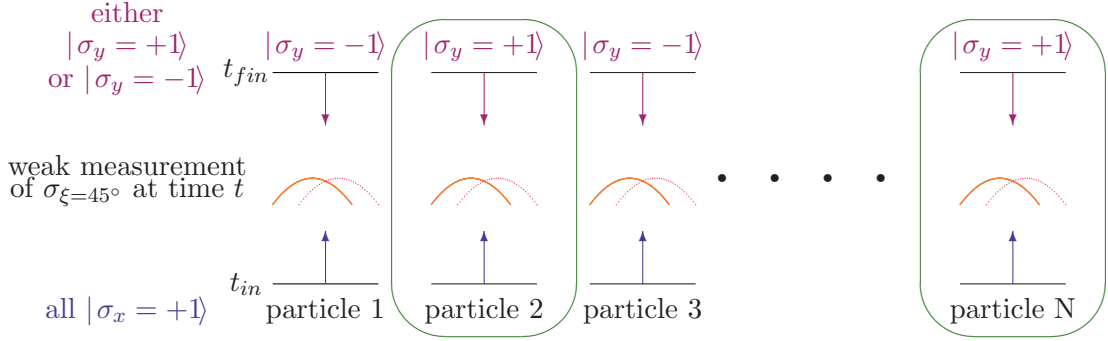

To show this in an actual calculation, we
use (\ref{post_selected}) and the post-selected state of the quantum system in the  $\sigma_{\xi}$ basis ($ \vert\uparrow_y\rangle\equiv \cos (\pi/8) \vert\uparrow_{\xi}\rangle - \sin (\pi/8)
\vert\downarrow_{\xi}\rangle$), the measuring device probability
distribution is:
\[
{\rm Prob}(P_{\mathrm{md}})=N^2 \left[\cos ^2(\pi/8) e^{ -(P_{\mathrm{md}}-1)^2 /{\Delta} ^2}
- \sin ^2(\pi/8) e^{ -(P_{\mathrm{md}}+1)^2 /{\Delta} ^2}\right]^2.
\]
With a strong or an ideal measurement, $\Delta\ll1$, the distribution is localized again around the eigenvalues $\pm1$, as illustrated in figures \ref{wmspin2}.a and
\ref{wmspin2}.b.
What is different, however, is that when the measurement is weakened, i.e. $\Delta$
is made
larger, then the distribution changes to one single distribution centered around
$\sqrt{2}$, the weak value, as illustrated in Figure \ref{wmspin2}.c-f, (the width again is reduced with an ensemble \ref{wmspin2}.f).
Using  ~(\ref{expweak}), we can see that the weak value is just the pre- and post-selected sub-ensemble arising from within the pre-selected-only ensembles.
\begin{figure}
  \centering
\includegraphics[width=12cm,height=8cm,keepaspectratio=true]{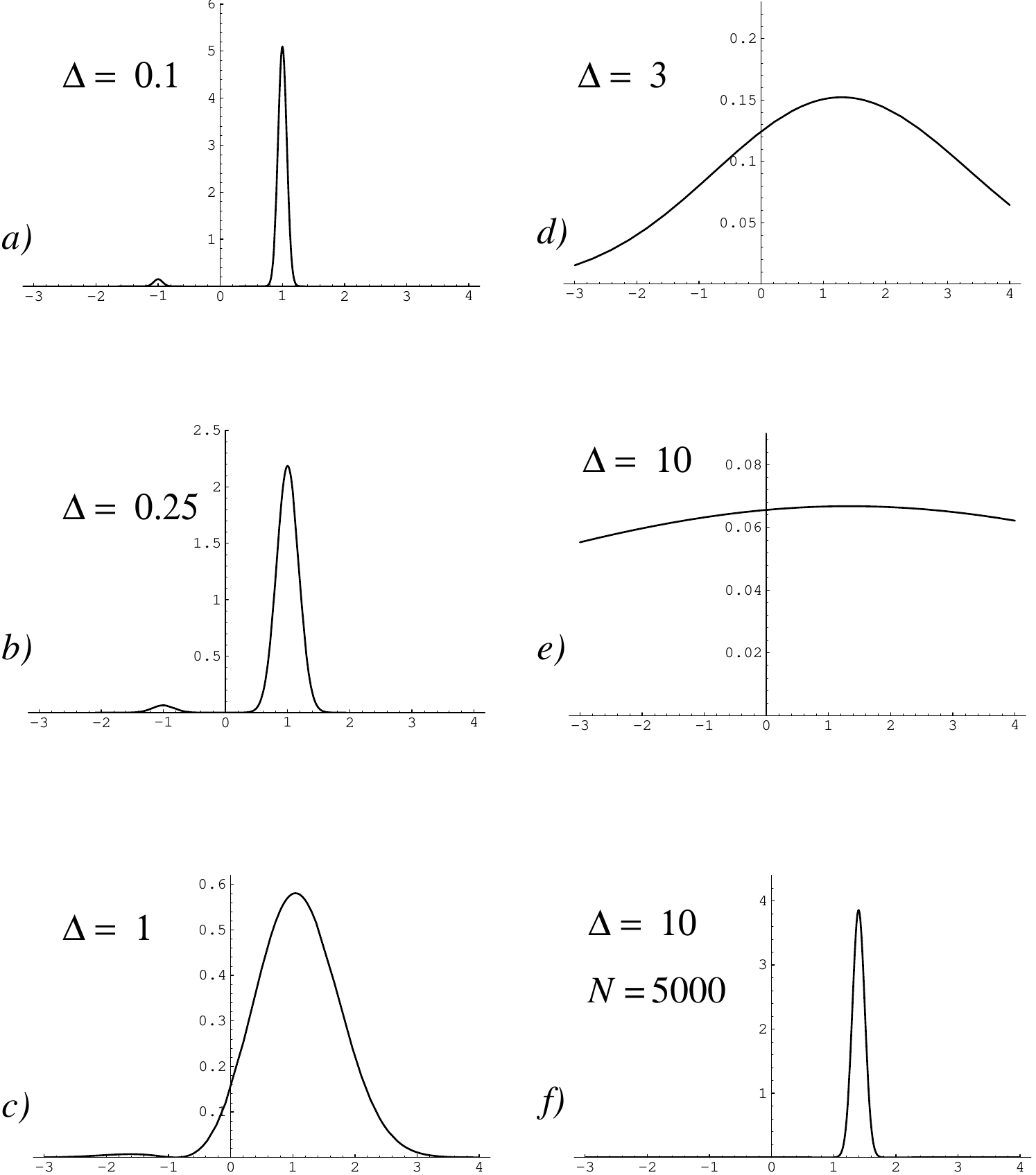}
\caption[Measurement on pre- and post-selected ensemble of single spins]{\small {\bf
~ Measurement on pre- and post-selected ensemble.}~
 ``Probability distribution of the pointer variable for measurement of
$\sigma_\xi$ when the particle is pre-selected in the state $\vert
{\uparrow_x} \rangle$ and post-selected in  the state $\vert
{\uparrow_y} \rangle$. The  strength of the
measurement  is
parameterized by the width of the distribution $\Delta$.
 ($a$) $\Delta = 0.1$; ($b$) $\Delta = 0.25$; ($c$) $\Delta =
1$; ($d$) $\Delta = 3$; ($e$) $\Delta = 10$.
  ($f$) Weak measurement on the ensemble
of 5000 particles; the original width of the peak, $\Delta = 10$, is reduced to
$10/\sqrt 5000 \simeq 0.14$. In the strong measurements ($a$)-($b$)
the pointer is localized
around the eigenvalues $\pm1$, while in the weak measurements ($d$)-($f$)
the peak of the distribution is located in the weak value
$(\sigma_\xi)_w = \langle {\uparrow_y}|
\sigma_\xi |{\uparrow_x} \rangle/\langle {\uparrow_y}|{\uparrow_x} \rangle
= \sqrt2$. The outcomes of the weak measurement on the ensemble
of 5000 pre- and post-selected particles, ($f$), are clearly outside the
range of the eigenvalues, (-1,1)." From \cite{av2v}.}
\label{wmspin2}
\end{figure}
\noindent

 Instead of considering an ensemble of spin-1/2 particles, we now consider ``particles" which are composed of many $N$ spin-1/2 particles,
and perform  a weak measurement of the collective observable ${\sigma}_\xi\upn \equiv \frac{1}{N}
\sum_{\mathrm{i=1}}^{N} {\sigma}_\xi^i$ in the $45^{\circ}$-angle to the $x-y$ plane. Using
$H_{\mathrm{int}} = -{{ g_0  \delta(t)}\over N}  {Q}_{\mathrm{md}} \sum_{\mathrm{i=1}}^N  {\sigma}^i_\xi$, a particular pre-selection of $|{\uparrow_x} \rangle$ (i.e. $|\Psi_{\mathrm{in}}\upn\rangle = \prod_{\mathrm{j=1}}^N |{\uparrow_x} \rangle_j$) and post-selection
 $|{\uparrow_y}\rangle$
(i.e.
$\langle\Psi_{\mathrm{fin}}\upn|  = \prod_{\mathrm{k=1}}^N \langle{\uparrow_y}|_k=\prod_{\mathrm{n=1}}^N \left\{\langle{\uparrow_z} |_n+i\langle{\downarrow_z}
|_n\right\}$).
The final state of the measuring device is then:
\beq
|\Phi_{\mathrm{fin}}^{\mathrm{md}}\ra=\prod_{j=1}^N \langle{\uparrow_y}|_j \exp\left\{{{ g_0}\over N}  {Q}_{\mathrm{md}} \sum_{\mathrm{k=1}}^N  {\sigma}^k_\xi\right\}  \prod_{i=1}^N|{\uparrow_x} \rangle_i |\Phi_{\mathrm{in}}^{\mathrm{md}}\ra .
\label{bigspinb}
\eeq
Since the spins do not interact with each other, we can calculate one of the products and take the result to the $N$-th power:
\[
\begin{split}
|\Phi_{\mathrm{fin}}^{\mathrm{md}}\ra&=\prod_{j=1}^N \langle{\uparrow_y}|_j \exp\left\{{{ g_0}\over N}  {Q}_{\mathrm{md}}  {\sigma}^j_\xi\right\} |{\uparrow_x} \rangle_j |\Phi_{\mathrm{in}}^{\mathrm{md}}\ra\\
&=\left\{\langle{\uparrow_y}| \exp\left\{{{ g_0}\over N}  {Q}_{\mathrm{md}}  {\sigma}_\xi\right\} |{\uparrow_x} \rangle\right\}^N\!\!\!\! |\Phi_{\mathrm{in}}^{\mathrm{md}}\ra .
\end{split}
\]
We now use the identity:
   $$\exp\left\{{i\alpha{\sigma}_{\vec{n}}}\right\}=\cos\alpha+i{\sigma}_{\vec{n}}\sin\alpha.$$
   This identity is easily proven using the fact that for any integer $k$ one has
$\sigma_{{n}}^{2k}=I$ and $\sigma_{{n}}^{2k+1}=\sigma_{{n}}$.
Thus it follows that:
\[
\begin{split}
e^{i\alpha\sigma_{{n}}}&=\sum_{k=0}^\infty\frac{(i\alpha)^k\sigma_{{n}}^k}{k!}\\
&=\sum_{k=0}^\infty\frac{(i\alpha)^{2k}}{(2k)!}+\sigma_{{n}}\sum_{k=0}^\infty\frac{(i\alpha)^{2k+1}}{(2k+1)!}\\
      &=e^{i\alpha\sigma_{{n}}}=\cos\alpha+i\sigma_{{n}}\sin\alpha.
\end{split}
\]
As a consequence we obtain:
\begin{eqnarray}
\ket{\Phi_{\mathrm{fin}}^{\mathrm{md}}}&=&\left\{\langle{\uparrow_y}| \left[\cos \frac{{ g_0} {Q}_{\mathrm{md}}}{N}-i{\sigma}_\xi\sin \frac{{ g_0} {Q}_{\mathrm{md}}}{N}\right]  |{\uparrow_x} \rangle\right\}^N |\Phi_{\mathrm{in}}^{\mathrm{md}}\ra\nonumber\\
&=&
{\left[\langle{\uparrow_y}|{\uparrow_x} \rangle\right]^N }
\left\{\cos \frac{{ g_0} {Q}_{\mathrm{md}}}{N}-i\alpha_w\sin \frac{{ g_0} {Q}_{\mathrm{md}}}{N}\right\}^N  |\Phi_{\mathrm{in}}^{\mathrm{md}}\ra
\label{bigspina}
\end{eqnarray}
where we have substituted $\alpha_w\equiv({\sigma}_\xi)_w=\weakv {\uparrow_y}{{\sigma}_\xi}{\uparrow_x }$.    We consider only the second part. In fact the first bracket, a number, can be neglected since it does not depend on ${Q}$ and thus can only affect the normalization:
\beq
\ket{\Phi_{\mathrm{fin}}^{\mathrm{md}}}=
\left\{1-\frac{{ g_0}^2 ({Q}_{\mathrm{md}})^2}{N^2} - \frac{i{ g_0} \alpha_w{Q}_{\mathrm{md}}}{N} \right\}^N|\Phi_{\mathrm{in}}^{\mathrm{md}}\ra\approx e^{i g_0\alpha_w  {Q}_{\rm md} } |\Phi_{\mathrm{in}}^{\mathrm{md}}\ra .\nonumber\\
\label{bigspinc}
\eeq
The last approximation for $N\rightarrow\infty$, is obtained by using
$(1+\frac{a}{N})^N=(1+\frac{a}{N})^{\frac{N}{a}a} \approx e^a$.\\
When we project onto $P_{\mathrm{md}}$, i.e. the pointer, we see that the pointer is robustly shifted by the
the same weak value, i.e. $\sqrt{2}$:
\beq
  ({\sigma}_\xi)_{{w}} = {{\prod_{k=1}^N \langle{\uparrow_y}|_k ~ \sum_{\mathrm{i=1}}^N
\left\{{\sigma}^i_x + {\sigma}^i_y\right\} ~ \prod_{\mathrm{j=1}}^N |{\uparrow_x} \rangle_j}
\over { \sqrt 2 ~ N(\langle{\uparrow_y} |{\uparrow_x} \rangle)^N}}=
\sqrt2 \pm O(\frac{1}{\sqrt N}).
\label{wvlargespin}
\eeq
A single experiment is now sufficient to
determine the weak value with great precision and there is no longer
any need to average over results obtained in multiple experiments as we did in the previous section.
Moreover, by repeating the experiment with different measuring devices, we see that  each measuring device
shows the very same weak values, up to an insignificant spread of $\frac{1}{\sqrt
{ N}}$
and the information from {\it both} boundary conditions, i.e.
$$|\Psi_{\mathrm{in}}\rangle = \prod_{\mathrm{i=1}}^N |{\uparrow_x} \rangle_i\qquad
 {\rm and}\qquad
\langle\Psi_{\mathrm{fin}}|  = \prod_{\mathrm{i=1}}^N \langle{\uparrow_y}|_i,
$$
describes the entire interval of time between pre- and post-selection.
For example, following \cite{av2v}, we consider $N=20$.  The probability distribution of the measuring device after the post-selection is:
\begin{equation}
  \label{probfin}
 {\rm Prob} (Q_{\rm md}^{(N)}) ={N}^2 \Bigl (\sum_{i=1}^N (-1)^i \bigl(\cos  ^2
(\pi/8)\bigr)^{N-i} \bigl(\sin  ^2 (\pi/8)\bigr)^i
e^{-(Q_{\rm md}^{(N)}-{{(2N-i)}\over N})^2/{2\Delta^2}}\Bigr)^2 .
\end{equation}
and is drawn for different values of $\Delta$ in Figure \ref{wmspin3}. While this result is rare, we have recently shown~\cite{at3} how any ensemble can yield robust weak values like this in a way that is not rare and for a much stronger regime of interaction. Thus, our discussion shows that weak values are a general property of every pre- and post-selected ensemble.

\begin{figure}
  \centering
  \includegraphics[width=15cm,height=11cm,keepaspectratio=true]{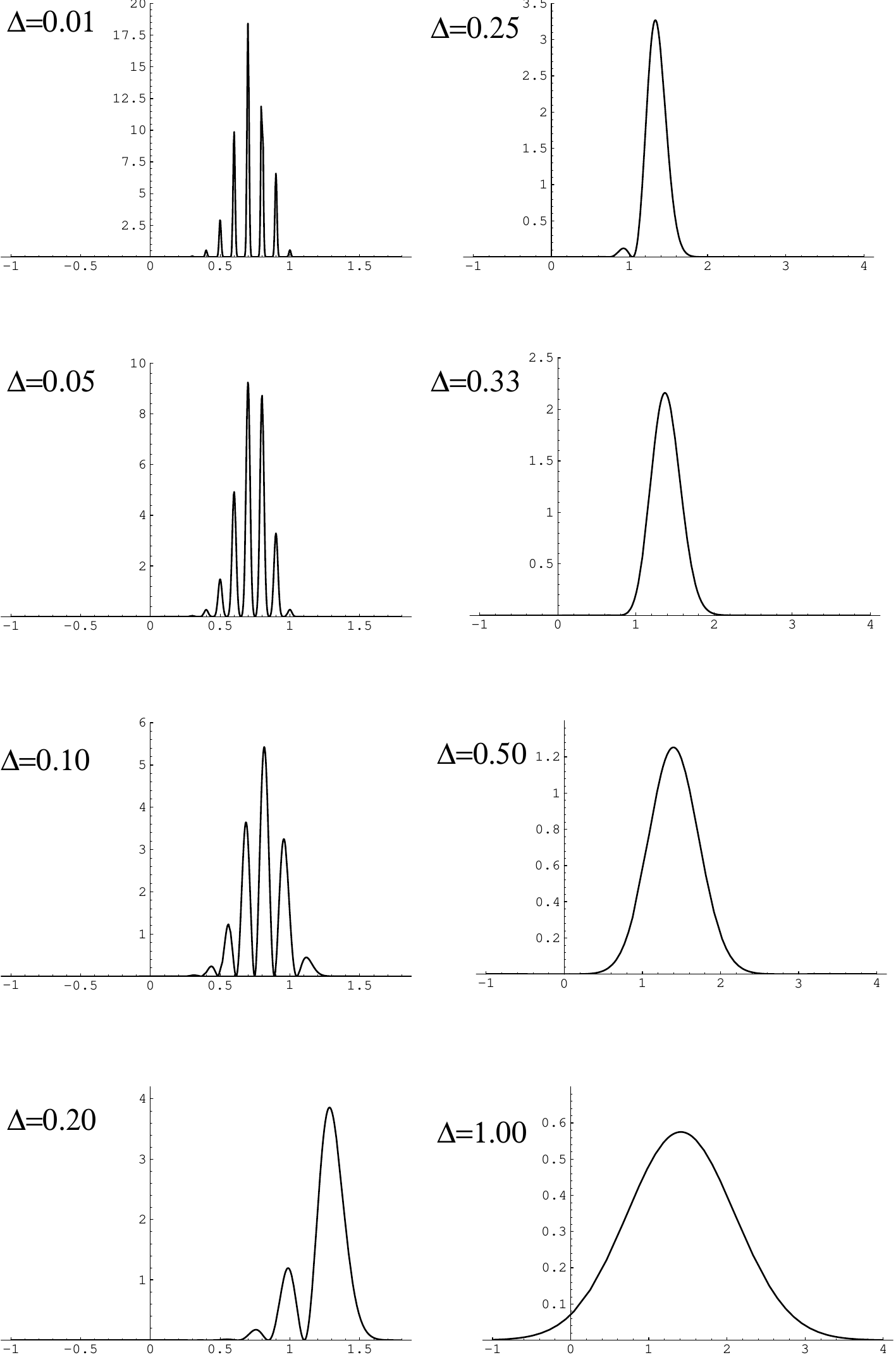}
\vskip 1cm
\caption[Measurement on pre- and post-selected ensemble of a single large spin]{\small{\bf ~ Measurement on a single system.}~ ``Probability
  distribution of the pointer variable for the measurement of $A
  =(\sum_{i=1}^{20} (\sigma_i)_\xi)/20$ when the system of 20 spin-$1\over 2$
  particles is pre-selected in the state $|\Psi_1\rangle =
  \prod_{i=1}^{20} |{\uparrow_x} \rangle _i$ and post-selected in the
  state ${|\Psi_2\rangle = \prod_{i=1}^{20} |{\uparrow_y} \rangle _i}$.
  While in the very strong measurements, $\Delta = 0.01-0.05$, the
  peaks of the distribution located at the eigenvalues, starting from
  $\Delta = 0.25$ there is essentially a single peak at the location
  of the weak value, $A_w = \sqrt 2$." From \cite{av2v}.}
\label{wmspin3}
\end{figure}
\noindent

As an example, consider again ~(\ref{bigspina}):
\begin{eqnarray}
|\Phi_{\mathrm{fin}}^{\mathrm{md}}\ra&=&\left\{\cos \frac{{\lambda} {Q}_{\rm md}}{N}-i\alpha_w\sin \frac{{\lambda} {Q}_{\rm md}}{N}\right\}^N|\Phi_{\mathrm{in}}^{\mathrm{md}}\ra\nonumber\\
&=&\left\{\frac{\exp(\frac{i\lambda Q_{\rm md}}{N})+\exp(-\frac{i{\lambda} {Q}_{\rm md}}{N})}{2}+\alpha_w
\frac{\exp(\frac{i\lambda Q_{\rm md}}{N})-\exp(-\frac{i{\lambda} {Q}_{\rm md}}{N})}{2}
\right\}^N|\Phi_{\mathrm{in}}^{\mathrm{md}}\ra\nonumber\\
&=&\underbrace{\left\{\exp\left({\frac{i{\lambda} {Q}_{\rm md}}{N}}\right)\frac{(1+\alpha_w )}{2}+\exp\left({-\frac{i{\lambda} {Q}_{\rm md}}{N}}\right)\frac{(1-\alpha_w )}{2}\right\}^N}_{\equiv \Psi(x)}|\Phi_{\mathrm{in}}^{\mathrm{md}}\ra \nonumber
\end{eqnarray}
We already saw how this could be approximated as $\exp({i\lambda\alpha_w  {Q}_{\rm md} } ) |\Phi_{\mathrm{in}}^{\mathrm{md}}\ra$ which produced a robust-shift in the measuring device by the weak value $\sqrt{2}$.
However, we can also view
$$
\Psi (x)=\left\{\exp\left({\frac{i{\lambda} {Q}_{\mathrm{md}}}{N}}\right)\frac{(1+\alpha_w )}{2}+\exp\left({-\frac{i{\lambda} {Q}_{\mathrm{md}}}{N}}\right)\frac{(1-\alpha_w )}{2}\right\}^N
$$
in a different way, by performing a binomial expansion:
\begin{eqnarray}
&\Psi (x)=\sum_{n=0}^N \frac{(1+\alpha_w)^n(1-\alpha_w)^{N-n}}{2^N} \frac{N!}{n!(N-
n)!}\exp\left(\frac{in\lambda{Q}_{\mathrm{md}}}{N}\right)\exp\left(\frac{-i\lambda{Q}_{\mathrm{md}} (N-n)}{ N}\right) &
\nonumber\\
&=\sum_{n=0}^N c_n \exp\left(\frac{i\lambda{Q}_{\mathrm{md}} (2n-N)}{ N}\right)=\sum_{n=0}^N c_n \exp\left(\frac{i\lambda{Q}_{\mathrm{md}}\lambda _n}{ N}\right).
\label{binom}
\end{eqnarray}
From this computation, we see that this wavefunction is a
superposition of waves with small wavenumbers $\vert k\vert\leq 1$ (because
$-1<\frac{2n-N}{N}<1$).
For a small region (which can include several wavelengths $2\pi /\alpha_w$, depending on
how large one chooses $N$), $\Psi (x)$
appears to have a very large momentum, since $\alpha_w$ can be arbitrarily large, i.e. a superoscillation.

\begin{figure}[tbph!]
\vskip -.4cm \centering
\includegraphics[width=4.54in,height=2.0in,keepaspectratio]{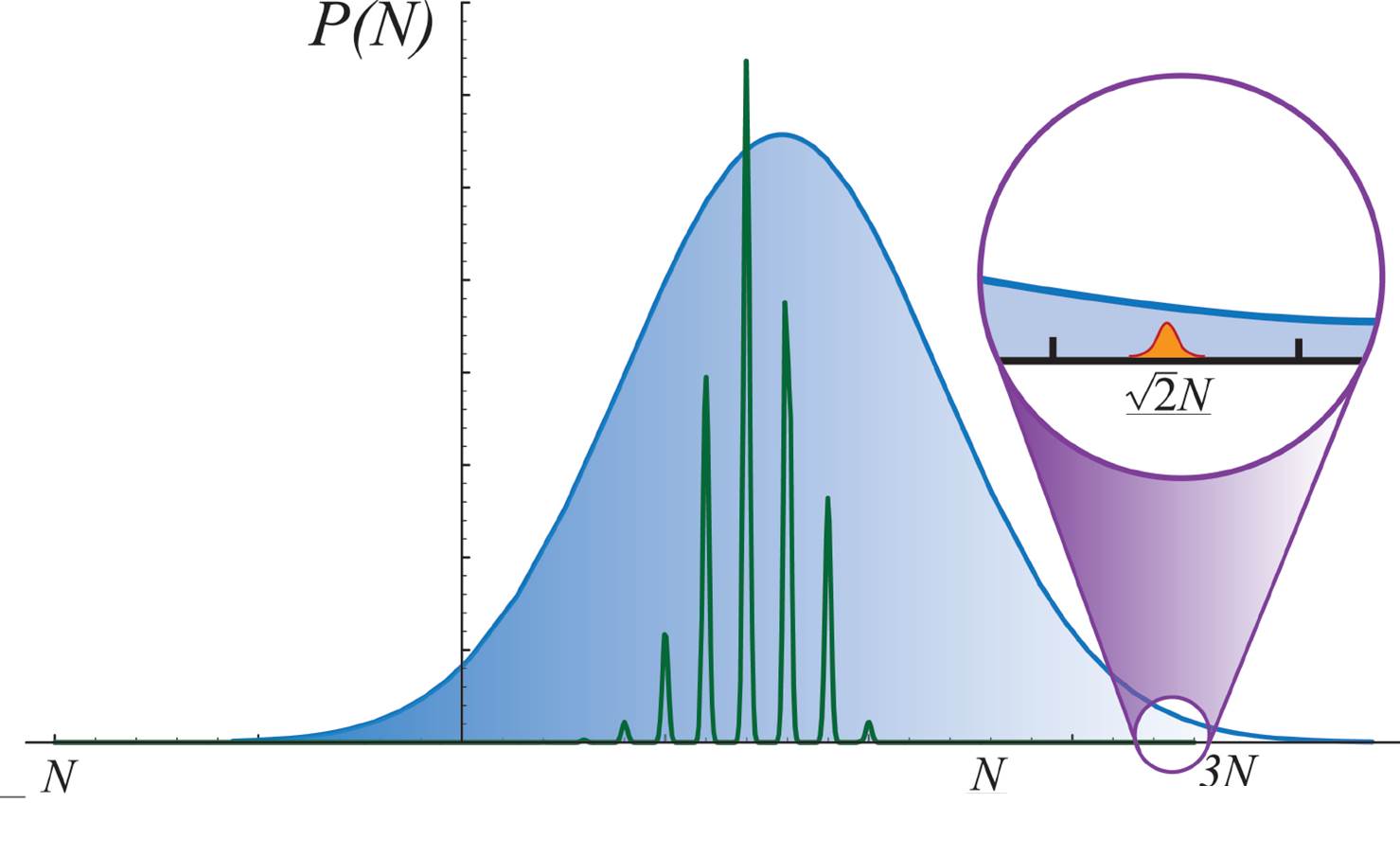}
\caption{{\footnotesize Combining the results from Figure \ref{wmspin3}, we draw here the probability distributions for strong and weak measurements of the observable $\sum_{i=1}^{20} (\sigma_i)_{45^{\circ}})/20$ for a system of $20$ spin-1/2 particles
pre-selected in the state $|\Psi_{\rm in}\rangle =
  \prod_{i=1}^{20} |{\uparrow_x} \rangle _i$. Before the post-selection is performed, the spikes in the distribution  (colored in green) represent the possile measurement outcomes (which are eigenvalues) for an ideal measurement.
The wider curve (colored in blue) represents the probabilities for a weak measurement. After the post-selection is performed for the very unlikely state ${|\Psi_{\rm fin}\rangle = \prod_{i=1}^{20} |{\uparrow_y} \rangle _i}$,
a single peak is left (colored in red) way out in the tail at the ``impossible" location
  of the weak value, $A_w = N\sqrt 2$. Adaptation based on \cite{townes}.}}
\label{bigspin-fig3}
\end{figure}

To summarize,
TSQM is a {\it reformulation} of the standard approach to Quantum Mechanics, and therefore it must be possible to view the novel effects from the traditional single-vector perspective.  This is precisely what super-oscillations teach us.  In summary,
there are  two ways to understand weak values:
\begin{itemize}
\item the measuring device is registering the weak value as a property of the system as characterized by TSQM.
\item the weak value is a result of a complex interference effect in the measuring device, i.e. a superoscillation; the system continues to be described with a single-vector pursuant to the standard approach to Quantum Mechanics.
\end{itemize}

As can be seen in Figure \ref{bigspin-fig3}, the probability to obtain the weak value as an error of the measuring device is greater than the probability to actually obtain the weak value.  This is essential to preserve causality. More importantly however, the weak value is not a random error. The weak value is precisely what we expect to happen. Furthermore, it is a highly predictable property because it always occurs whenever we obtain a given post-selection.
 Finally, it is an appropriate description for a broad range of physical situations, since  any weakened interaction, and not just measurements, experience the weak value.

\begin{remark}
The notions of weak value and weak measurement have turned out to be very applicable and we conclude this Chapter by pointing out a number of facts which lead to experiments.
We begin by observing that
limited
 disturbance measurements have been used to explore many
 paradoxes such as Hardy's paradox ~\cite{at2}, \cite{har92}, \cite{lun09}, \cite{yok09}, the three-box problem \cite{jmav}, \cite{Kolenderski}, \cite{Goerge}, \cite{steinber2004} and other paradoxes \cite{at}, \cite{ccat}, \cite{nt}, \cite{jt}, \cite{at3}.  A number of experiments have been performed to test the predictions made by weak measurements and results have proven to be in very good agreement with
 theoretical predictions \cite{Ahnert}, \cite{Pryde}, \cite{RSH}, \cite{Wiseman}.  Since eigenvalues or expectation values can be derived from weak values~\cite{ab}, it is clear the reason why the weak value is of fundamental importance in Quantum Mechanics.
Paper \cite{aav}, and the discovery of superoscillations have proven to be extremely useful tools in quantum information science and technology.
For example, \cite{aav} and superoscillations directly led to the notion of the quantum random walk~\cite{adz}.
It was shown that implementation of the quantum random walk would lead to a universal quantum computer as well as a quantum
simulator (to study, e.g., phase transitions). Recent experimental realizations of the quantum random walk have been successful (for example: trapped atom with optical
lattice and ion trap; photons in linear optics). Moreover, the most recent proposal for a quantum algorithm which yields a quantum speed-up for a quantum computer was based on the quantum random walk \cite{qrw-alg}.
In addition, the quantum random walk offers one of the most successful connections with topological phases \cite{qrw-phase}.
This has also already been implemented experimentally. The quantum random walk has thus resulted in the first demonstration of topological phases in one dimension by using linear optics \cite{qrw-phase2}.

Another example of the impact of \cite{aav} on quantum information science and technology is a novel metrological technique.
 Aharonov, Albert and Vaidman first pointed out that the anomalously large deflection of the beam in a Stern-Gerlach apparatus was controlled by the overlap between pre- and post-selected states, as well as the size of the magnetic field.  Consequently, a small change in the magnetic field would result in a large change in the beam displacement.  Therefore, this effect could also be used to measure small changes in a magnetic field.
After the original suggestion in 1988 in \cite{aav}, nothing happened until paper \cite{at3} appeared in 2007, which resulted in the first experimental use of weak values for metrology in an experiment by Hosten and Kwiat \cite{kwiat}. They measured
the optical spin Hall effect, an effect where different polarizations of light are shifted spatially in different, polarization dependent directions when the beam is incident on a glass interface. The effect theoretically corresponded to a spatial shift by 1 Angstrom, which is much smaller than the width of the optical beam.  In order to measure this shift, they utilized \cite{aav}, pre- and post-selected the polarization,
and measured the  deflection (amplified by $10^4$) on a position sensitive detector (at the cost of a reduced intensity).

Furthermore, Brunner and Simon \cite{brunner} introduced an
``interferometric scheme based on a purely imaginary weak value, combined
with a frequency-domain analysis, which may have potential to outperform
standard interferometry by several orders of magnitude", see \cite{lix11}.

One difficulty in the polarization based experiments is the fact that the source of deflection must be polarization dependent.   The Rochester group was able to generalize this by switching over to an interferometer based system \cite{townes}, where the different paths corresponded to different directions of deflection.  The pre-selection and post-selection corresponded to the optical beam entering and leaving different ports of the interferometer, while the weak measurement corresponded to a moving mirror slightly  misaligning the interferometer.  A piezo activated mirror moved the mirror slightly back and forth by a known amount, and the test was to see how small the interferometric weak value technique could measure it.  With an hour of integration time, the group reported 500 frad resolution, and later found a signal-to-noise ratio at the standard quantum limit.  This was done with milliwatts of power in an open air experiment \cite{howell}.
Turner et al., see \cite{tur11}, adjusted the scheme of  \cite{howell} for the use in torsion balance experiments in gravity research; they demonstrated picoradian accuracy of deflection measurements.
Hogan et al. \cite{hog11} included a folded optical lever into the scheme of \cite{howell} and achieved a record angle sensitivity of 1.3 prad/$\sqrt{\rm Hz}$; their scheme is potentially applicable for gravitational wave detection.

Many more experiments have been performed utilizing weak PPS measurements. We mention
\cite{bru04}, \cite{cho10}, \cite{howell}, \cite{gog11}, \cite{kwiat}, \cite{how10}, \cite{lun09}, \cite{mir07}, \cite{par98}, \cite{Pryde}, \cite{steinber2004}, \cite{RSH}, \cite{sol04}, \cite{sta09}, \cite{sta10}, \cite{sut95}, \cite{sut93}, \cite{wan06}, \cite{yok09}.
Experimentally, it was implemented on a large variety of systems, types of couplings, and experimental configurations. Most of the experiments were optical, one \cite{sut93} utilized nuclear magnetic resonance, and others utilized a solid state setup. Foe example, \cite{gefen} showed ``a significant amplification even in the
presence of finite temperatures, voltage, and external noise."
 Some of the optics experiments were performed on single photons \cite{bru04}, \cite{cho10}, \cite{howell}, \cite{kwiat}, \cite{how10}, \cite{mir07}, \cite{steinber2004}, \cite{RSH}, \cite{sol04}, \cite{par98}, \cite{sta09}, \cite{sta10}, \cite{sut95},
could be interpreted both classically and from a quantum perspective.
Even from a purely classical perspective, the study of weak values has led to a variety of ``new" phenomena. For example, the experiment proposed by Knight et al.  \cite{kni90} and performed by Parks et al.  \cite{par98} demonstrated an enhanced shift of the beam in either coordinate or momentum space.  Morover, the weak PPS methods have been implemented to produce  beam-deflection measurements \cite{hog11}, \cite{sta09}, \cite{tur11}, as well as phase and frequency \cite{sta10} measurements.
\end{remark}

\chapter{Basic mathematical properties of superoscillating sequences}
\label{sec3}

\section{Superoscillating sequences}

In this chapter we study the mathematical properties of superoscillating sequences.

\begin{definition}\label{ipsilon enne}
We call {\em generalized Fourier sequence}\index{generalized Fourier sequence}
a sequence of the form
\begin{equation}\label{basic_sequence}
Y_n(x,a):= \sum_{j=0}^n C_j(n,a)e^{ik_j(n)x}
\end{equation}
where $a\in\mathbb R^+$,  $n\in\mathbb N$, $C_j(n,a)$ and $k_j(n)$ are real valued functions.
\end{definition}
\begin{remark}{\rm
The sequence of partial sums of a Fourier expansion is a particular case of
this notion with $C_j(n,a)=C_j\in\mathbb R$ and $k_j(n)=k_j\in\mathbb R$ are multiples of a real number.}
\end{remark}
\begin{definition}\label{superoscill}
Let $a,\alpha\in\mathbb{R}^+$.
A generalized Fourier sequence
$$Y_n(x,a)=\sum_{j=0}^n C_j(n,a)e^{ik_j(n)x}$$ is said to be a {\em superoscillating sequence} if: \index{superoscillating sequence}
\begin{itemize}
 \item[i)]
  $|k_j(n)|< \alpha$\ \ {\rm for all} $n$ {\rm and} $j\in \mathbb{N}\cup\{0\}$;
\item[ii)] there exists a compact subset of $\mathbb R$, which will be called a {\em superoscillation set}, on which $Y_n$ converges uniformly to $e^{ig(a)x}$ where $g$ is a continuous real-valued function such that $|g(a)| > \alpha$.
\end{itemize}
\end{definition}
The usual Fourier sequence of a function is obviously not superoscillating because it violates i).\\

Indeed, one can consider a somewhat more general situation described by the following definition.
\begin{definition}\label{generalized_super}
 Given $f(x)=\sum_{j=0}^{\infty} d_j e^{i a_j x}$, for $a_j$, $d_j\in \mathbb{R}$, we say that the sequence
$$S_n(x)=\sum_{j=0}^n C_j(n)e^{ik_j(n)x},$$ where $C_j(n)$ are real valued functions, is $f$-superoscillating  if there exists an index $J$ such that $\sup_j |k_j(n)|<a_J$ and the sequence $S_n$ converges uniformly to $f$ on some compact subset of $\mathbb R$.
\end{definition}

\begin{remark}{\rm
 Definition \ref{generalized_super} can be considered in a different framework. Assume $f$ to be a $\mathcal C^\infty$ function. Then the representation $f(x)=\sum_{j=0}^{\infty} d_j e^{i a_j x}$ is equivalent (under suitable conditions) to requiring that $f$ is a solution of the convolution equation
$$
\mu * f=0
$$
where $\mu$ is a compactly supported distribution whose Fourier transform $\hat\mu$ vanishes exactly at the points $ia_j$ with multiplicity $1$. On the other hand, every partial sum
$$
S_n(x)=\sum_{j=0}^n C_j(n)e^{ik_j(n)x}
$$
is clearly the solution of an ordinary differential equation
$$
\wp_n\left(\frac{d}{dx}\right)S_n= 0
$$
where the polynomials $\wp_n$ vanish at least at the points $i k_j(n)$.
Therefore, just like entire functions can be seen as "limits" of polynomials of increasing degrees, so convolution operators are limits of ordinary differential operators of increasing order. Of course not every holomorphic function is the symbol of a convolution operator since the Paley-Wiener-Schwartz theorem \cite{ehrenpreis} imposes very specific requests on the growth of an entire function $F$ in order for $F$ to be the Fourier transform of a compactly supported distribution $\mu$, and therefore the symbol of the convolutor $\mu * \cdot$.}
\end{remark}

\begin{remark}{\rm
It is worth making a different kind of argument as well. Instead of considering the class of $\mathcal{C}^\infty$ functions, one might be tempted to consider the space of analytic functions. This would be reasonable since the functions we are considering are convergent series of exponentials, and it would be possible to impose suitable conditions on the coefficients $d_j$ to ensure the analyticity of the sum of the series. However, the space $\mathcal A$ of real analytic functions is not an Analytically Uniform space (see Chapter 4 and \cite{bd}) and the theory of convolution equations in it is not yet well understood. We will instead consider the space of those real analytic functions which can be extended  as entire functions of a complex variable, when the real variable $x$ is replaced by the complex variable $z$.  In this case it is possible to give explicit conditions for $\sum_{j=0}^\infty d_j \exp(ia_j z)$ to be convergent to an entire function in $\mathbb C$, see \cite{bt1}, and the main advantage which we will fully exploit in Chapters 5 and 6 is the fact that such series can be considered as solutions to suitable convolution equations $\mu * f=0$, where now the convolutor $\mu$ is an analytic functional whose Fourier-Borel transform is also an entire function and has exponential type bounds, see Chapter 3. An important special case of this situation occurs when the analytic functional $\mu$ is supported at the origin. In this case
the expansion $\sum_{j=0}^\infty d_j \exp(ia_j z)$ may be interpreted as being the solution of an infinite order differential equation, and
$a_j$ satisfy the condition $|a_j-a_\ell|\geq c |j-\ell|$ for some constant $c$.}
\end{remark}

The primary example of superoscillating sequence was already discussed in the introduction, and is given by the following sequence:
\begin{equation}\label{ef}\index{superoscillating sequence}
F_n(x,a)=\Big(\cos \Big(\frac{x}{n}\Big)+ia\sin \Big(\frac{x}{n}\Big)\Big)^n,
\end{equation}
where  $a>1$, $n\in \mathbb{N}$, and $x\in\mathbb{R}$.

\begin{proposition}\label{pro3.1.6}
 Consider the sequence (\ref{ef}).  Then we have
\begin{itemize}
\item[(1)]
For every $x_0\in \mathbb{R}$
$$
\lim_{n\to \infty}F_n(x_0,a)=e^{ia{x_0}}.
$$
\item[(2)]
The functions  $F_n(x,a)$ can be written in terms of their Fourier coefficients $C_j(n,a)$ as
$$
F_n(x,a)=\sum_{j=0}^nC_j(n,a)e^{i(1-2j/n) {x}} ,
$$
where
$$
C_j(n,a):=\frac{(-1)^j}{2^n}{n\choose j}(a+1)^{n-j}(a-1)^j.
$$
\item[(3)] For every $p\in \mathbb{N}$ the following relation
$$
F_n^{(p)}(0,a)=\sum_{j=0}^nC_j(n,a) \left[i\left(1-\frac{2j}{n}\right)\right]^p
$$
between the Taylor and the Fourier coefficients of (\ref{ef}) holds.
\end{itemize}
\end{proposition}
\begin{proof}
Point (1) follows from standard computation. Point (2) is a consequence of the Newton binomial formula.
Point (3) follows by taking the derivatives of $$\sum_{j=0}^nC_j(n,a)e^{i(1-2j/n) {x}} $$ and computing them at the origin.
\end{proof}

We observe that
$$
F_n(x,a)\to e^{{i a}x},
$$
as $ n\to \infty$,
which follows from either representation of
$F_n(x,a)$.

\bigskip
\begin{theorem}\label{carnot}
Let $M>0$ be a fixed real number. Then for every $x$ such that $|x|\leq M$
the sequence $F_n(x,a)$ converges uniformly to $e^{ia{x}}$. Thus $F_n(x,a)$ is a superoscillating sequence.
\end{theorem}
\begin{proof}
We  have to show that for every $x$ such that $|x|\leq M$ we have
$$
\sup_{|x|\leq M}  |F_n(x,a) -e^{ia{x}}|\to 0\ \ \ \ {\rm as} \ \ \ n\to \infty.
$$
To this purpose, we will compute an estimate for the modulus of the function $F_n(x,a) -e^{ia{x}}$.
Let us set
$$
w=F_n(x,a) \ \ \ \ \ {\rm and}\ \ \ \ \ z=e^{ia{x}},
$$
and observe that the modulus and the angles associated with $w$ and $z$ are, respectively,
$$
\rho_w=\Big(\cos^2 \Big(\frac{x}{n}\Big)+a^2\sin^2 \Big(\frac{x}{n}\Big)\Big)^{n/2},
\ \ \ \ \
\theta_w=n\arctan \Big(a \tan \Big(\frac{x}{n}\Big)\Big)
$$
and
$$
\rho_z=1,
\ \ \ \ \ \ \ \
\theta_z=a{x}.
$$
The Carnot theorem for triangles gives
$$
|w-z|^2=1+\rho_w^2-2\rho_w \cos(\theta_w -\theta_z)
$$
so that we obtain
\begin{equation}\label{differenceEn}
\begin{split}
 &|F_n(x,a) -e^{ia{x}}| ^2=1+\Big(\cos^2 \Big(\frac{x}{n}\Big)+a^2\sin^2 \Big(\frac{x}{n}\Big)\Big)^{n}\\
&-2\Big(\cos^2 \Big(\frac{x}{n}\Big)+a^2\sin^2 \Big(\frac{x}{n}\Big)\Big)^{n/2} \cos\Big[n\arctan \Big(a \tan \Big(\frac{x}{n}\Big)\Big) -a{x}\Big].
\end{split}
\end{equation}
Let us set
\begin{equation}\label{errorf}
\mathcal{E}_n^2(x,a):= |F_n(x,a) -e^{ia{x}}| ^2
\end{equation}
and observe that for any $x$ such that  $|x|\leq M$ we have
$$
\Big(\cos^2 \Big(\frac{x}{n}\Big)+a^2\sin^2 \Big(\frac{x}{n}\Big)\Big)^{n}\to 1\ \ \ {\rm as} \ \ \ n\to \infty
$$
and
$$
\cos\Big[n\arctan \Big(a \tan \Big(\frac{x}{n}\Big)\Big) -a{x}\Big]\to 1 \ \ \ {\rm as} \ \ \ n\to \infty.
$$
Using \eqref{differenceEn}, we deduce that $\mathcal{E}_n^2(x,a)\to 0$.
Note that $\mathcal{E}_n^2(x,a)$, as a function of $x$, is continuous on the compact set $[-M,M]$ for any
$
n > \frac{2M}{\pi}
$
so $\mathcal{E}_n^2(x,a)$ has maximum.
Set
$$
\varepsilon(n,a)=\max_{x\in [-M,M]} \mathcal{E}_n(x,a).
$$
It is now easy to see that $\varepsilon(n,a)\to 0$ as $n\to \infty$ and since
$$
\sup_{|x|\leq M}  |F_n(x,a) -e^{ia {x}}|=\varepsilon(n,a)
$$
the convergence is  uniform in $[-M,M]$.
\end{proof}
The next result however shows that on $\mathbb R$ the sequence $F_n(x,a)$ does not converge uniformly.
\begin{proposition}
The sequence $F_n(x,a)$ does not converge uniformly to $e^{iax}$ on $\mathbb R$.
\end{proposition}
\begin{proof}
Uniform convergence on $\mathbb R$ would be equivalent to
$$
\sup_{x\in \mathbb{R}}  |F_n(x,a) -e^{ia{x}}|\to 0\ \ \ \ {\rm as} \ \ \ n\to \infty.
$$
If $x=0$ obviously $F_n(0,a) -e^{0}=0$, however, if we consider the points
$x_n=j\pi n$ for $j\in \mathbb{Z}\setminus \{0\}=\{\pm 1, \pm2,...\}$ we have
$$
|F_n(x_n,a) -e^{ia{x_n}}|=| (\pm 1)^n-e^{ia{x_n}}|\not\to 0\ \ \ \ {\rm as} \ \ \ n\to \infty
$$
if $a\in \mathbb{R}\setminus (\mathbb{Z}\setminus \{0\})$, so the convergence cannot be uniform.\\
If $a\in \mathbb{Z}\setminus \{0\}$, we reason in the same way by taking $x_n=n\frac{\pi}{2}$.
\end{proof}

\begin{remark}
{\rm Proposition \ref{pro3.1.6} shows that, under reasonable assumptions, every function which can be represented as a Fourier series of exponentials, can in fact be seen as the limit of a suitable superoscillating sequence, though in general it may be difficult to explicitly compute its terms. Indeed, let $f(x)=\sum_{j=0}^{\infty}\alpha_j e^{ia_jx}.$ We know from the previous result, that for each $j$ there exists a superoscillating sequence $F_n(x,a_j)$ which converges to $e^{ia_jx}.$ Then we have that
$$
f(x)=\sum_{j=0}^{\infty}\alpha_j e^{ia_jx}=\sum_{j=0}^{\infty}\alpha_j \lim_{n \to \infty}F_n(x,a_j).
$$
Since the convergence of $F_n(x,a_j)$ to $e^{ia_jx}$ is uniform on compact sets, and by imposing suitable conditions on the growth of the sequence $\alpha_j$, we can exchange the limit and the series to obtain
$$
f(x)=\lim_{n \to \infty} \sum_{j=0}^{\infty} \alpha_j F_n(x,a_j)=\lim_{n \to \infty}\mathcal{Q}_n(x),
$$
where $\mathcal{Q}_n(x)=\sum_{j=0}^{\infty} \alpha_j F_n(x,a_j).$ One can easily see that $\mathcal{Q}_n(x)$ is superoscillating as well. The specific request on the growth of the sequence $\{\alpha_j\}$ is a consequence of the estimate on the error $\mathcal{E}_n^2(x,a)$, see Remark \ref{estimate}.}
\end{remark}
\begin{remark}{\rm
Theorem \ref{carnot} explains the mathematical behavior of superoscillations in terms of the
Taylor and Fourier coefficients. Indeed,
for every $a>1$ and for every $n\in \mathbb{N}$
a direct computation of $F_n'$  gives
$
F_n'(0,a)={i a}.
$
By Proposition \ref{pro3.1.6}, Point (3), we have
$$
F_n'(0,a)=\sum_{j=0}^nC_j(n,a) i\Big(1-\displaystyle\frac{2j}{n}\Big),
$$
and, more in general,
\begin{equation}\label{Fnp}
F_n^{(p)}(0,a)=\sum_{k=0}^nC_k(n,a)[i(1-2k/n)]^p.
\end{equation}
Therefore, by taking $p=0,1$ in this last formula,
 we obtain
$$
\sum_{k=0}^nC_k(n,a)=1,
$$
and
$$
{a}=\sum_{j=0}^nC_j(n,a) \left(1-\frac{2j}{n}\right).
$$
}
\end{remark}
One can also obtain additional combinatorial identities by directly calculating the higher derivatives
of $F_n(x,a)$ in its original expression, see \cite{kyle lee}. More specifically, we have:
\begin{proposition}
For any value of  $a,x\in\mathbb R$, $n\in\mathbb N$,
\begin{equation}\label{multinomial}
F_n^{(p)}(x,a)=\sum_{\substack{k_1,k_2,\ldots ,k_{n}=0\\k_1 + k_2 + \cdots + k_{n}=p}}^p \frac{p!}{k_1! k_2! \ldots k_{n-1}!k_n!} g^{(k_1)}g^{(k_2)}\ldots g^{(k_{n-1})}g^{(k_n)}, \quad p\in\mathbb N.
\end{equation}
\end{proposition}
\begin{proof}
We will prove this result by induction on $p.$ Assume that (\ref{multinomial}) is true for a given $p$ (and obviously it is true for $p=0$). Then the $(p+1)-$th derivative of $F_n$ is given by
\begin{equation}
F_n^{(p+1)}(x,a)=\sum_{\substack{k_1,k_2,\ldots ,k_{n}=0\\k_1 + k_2 + \cdots + k_{n}=p}}^p \frac{p!}{k_1! k_2! \ldots k_{n-1}!k_n!} (g^{(k_1+1)}g^{(k_2)}\ldots g^{(k_{n-1})}g^{(k_n)}+
\end{equation}
$$
+g^{(k_1)}g^{(k_2+1)}\ldots g^{(k_{n-1})}g^{(k_n)}+ \cdots +g^{(k_1)}g^{(k_2)}\ldots g^{(k_{n-1})}g^{(k_n+1)}),
$$
and the result now follows from the fundamental identity for multinomial coefficients (see \cite{gallier}, Chapter IV), which states that
$$
\frac{(p+1)!}{h_1! h_2! \ldots h_{n-1}!h_n!}=\frac{p!}{(h_1-1)! h_2! \ldots h_{n-1}!h_n!}+\frac{p!}{h_1! (h_2-1)! \cdots h_{n-1}!h_n!}
$$
$$
+\ldots +\frac{p!}{h_1! h_2! \ldots (h_{n-1}-1)!h_n!}+\frac{p!}{h_1! h_2! \ldots h_{n-1}!(h_n-1)!}
$$
\end{proof}

In particular, the expression for the derivative can be given a very compact form when evaluated at the origin.

\begin{corollary}
For any value of $a\in\mathbb R$,  $n\in\mathbb N$,
$$
F_n^{(p)}(0,a)=\left(\frac{i}{n}\right)^p\sum_{\substack{k_1,k_2,\ldots ,k_{n}=0\\k_1 + k_2 + \ldots + k_{n}=p}}^p \frac{p!}{k_1! k_2! \ldots k_{n-1}!k_n!} a^{\varepsilon}, \qquad p\in\mathbb N
$$
$\varepsilon$ is the number of odd integers in $\{k_1, k_2, \ldots , k_n\}$.
\end{corollary}

\begin{proof}
It follows immediately from evaluating different orders of differentiation of $g$ at the origin.
$$
g^{(p)}(0) =
\begin{cases}
 (\frac{i}{n})^p, & \text{if }p\text{ is even} \\
 a(\frac{i}{n})^p , & \text{if }p\text{ is odd}
 \end{cases}
 $$
\end{proof}

Comparing this last expression with \eqref{Fnp}, we obtain a new family of nontrivial identities.
\begin{corollary}
For any value of $a\in\mathbb R$,  $n,p\in\mathbb N$, we have
\begin{equation}\label{IDENTITIES}
\sum_{\substack{k_1,k_2,\ldots ,k_{n}=0\\k_1 + k_2 + \cdots + k_{n}=p}}^p \frac{p!}{k_1! k_2! \ldots k_{n-1}!k_n!} a^{\varepsilon} = \sum_{k=0}^nC_k(n,a)(n-2k)^p.
\end{equation}
\end{corollary}

Note that when $a=1$ the only non--vanishing coefficient $C_k(n,a)$ is $C_0(n,1)=1$, and therefore the previous corollary simply gives the very well known formula
$$
\sum_{\substack{k_1,k_2,\ldots,k_{n}=0\\k_1 + k_2 + \ldots + k_{n}=p}}^p \frac{p!}{k_1! k_2! \cdots k_{n-1}!k_n!}=n^p.
$$

\begin{remark}\label{estimate}
{\rm
Recall that in the proof of Theorem \ref{carnot} the term
$$|F_n(x,a)\! -\!e^{ia{x}}| ^2=\mathcal{E}_n^2(x,a)$$ is given by formula (\ref{differenceEn}).
 We can give a first approximation of  $\mathcal{E}_n^2(x,a)$  by considering the principal part of the infinitesimum $\mathcal{E}_n^2(x,a)$
for $|x|\leq M$.  If we can find
two  constants $j$ and $\beta\in \mathbb{R}^+$ such that
$$
\mathcal{E}_n^2(x,a)=\beta \,\Big(\frac{x}{n}\Big)^j+o\Big(\frac{x}{n}\Big)^{j}, \ \ \ {\rm as}\  \ \ \ n \to \infty,
$$
we can choose
$$
\mathcal{E}_n^2(x,a)\approx\beta\,\Big(\frac{x}{n}\Big)^j\  \ \ {\rm as}\  \ \ \ n \to \infty,
$$
as first approximation of $\mathcal{E}_n^2(x,a)$.
With some computations we have
$$
\beta\,\Big(\frac{x}{n}\Big)^j=\frac{3}{2}(a^2-1)\Big(\frac{x}{n}\Big)^2,
$$
and so
$$
\mathcal{E}_n(x,a)\approx \frac{x}{n} \sqrt{\frac{3}{2}\Big(a^2-1\Big)}.
$$
This proves that also $\mathcal{E}_n^2(x,a)\to 0$ uniformly over the compact sets.
}
\end{remark}

We now offer a simple estimate for the first and second derivatives of the functions $F_n.$
\begin{proposition}\label{sjhdjhgfns}
Let $F_n$ be the function defined in (\ref{ef}). Then for $x$ in a compact set in $\mathbb{R}$ we have:
$$
\lim_{n\to\infty}|F'_n(x,a)| = a \qquad  \lim_{n\to\infty} |F''_n(x,a)| = a^2.
$$
\end{proposition}

\begin{proof}
In this proof, it is convenient to directly compute the derivatives
of the function $F_n(x,a)=g_n^n(x)$ where
$$
g_n(x)=\cos \Big(\frac{x}{n}\Big) +ia \sin \Big(\frac{x}{n}\Big).
$$
(Note that, to simplify the notation, we do not make it explicit that $g$ depends on $a$).
It is immediate to see that
$$
F'_n(x,a)=ng^{n-1}_n (x) g_n'(x)=n\frac{g_n'(x)}{g_n(x)}F_n(x,a),
$$
and
$$
F''_n(x,a)=\left(-\frac{1}{n}+n(n-1)\Big(\frac{g'_n(x)}{g_n(x)
}\Big)^2\right)F_n(x,a).
$$
We are now interested in the asymptotic behavior of $F_n'$ and $F_n''$ when $n$ goes to infinity and $x$ is limited to a compact subset of $\mathbb{R}$. We know from Theorem \ref{carnot} that $|F_n(x,a)|$ converges to $1$ on every compact subset of $\mathbb{R}$, so it is enough to compute the asymptotic behavior of $g_n'/g_n$.
We obtain:
$$
\frac{g'_n(x)}{g_n(x)}=\frac{1}{n} \ \frac{-\sin (\frac{x}{n}) +ia \cos (\frac{x}{n})}{\cos (\frac{x}{n}) +ia \sin (\frac{x}{n})}\sim \frac{1}{n}ia
$$
for $n$ large and $x$ in a compact set.
The statement follows.
\end{proof}

\begin{remark}{\rm
We point out that the material in this chapter
 is based on a precise definition of superoscillation phenomenon in terms of the
 uniform convergence of sequences of functions.
 Here we follow a different approach  with respect to the one in \cite[Section 2]{b4}
  where superoscillations are not studied in terms of the
  uniform convergence of functions.
  In \cite{b4} the authors treat  a different case by describing superoscillations with wavenumbers different from $a$, in the region away from the origin when $n$ is large but finite.
  Consequently, they consider the sequence
  $$
  G_n(x,a):=(\cos x+ia\sin x)^n
  $$
that converges only at $x=0$, and therefore does not fit in our setting.
}
\end{remark}

\medskip
\par\noindent
The following remark further clarifies the behavior of the sequence $F_n$.
\begin{remark}
{\rm  Consider a point $x_0$ and an increment $\delta x$.
The superoscillating sequence
$$
F_n(x,a):=\Big(\cos \Big(\frac{x}{n}\Big)+ia\sin \Big(\frac{x}{n}\Big)\Big)^n
$$
is such that $F_n(x_0,a)\to e^{iax_0}$ and $F_n(x_0+\delta x,a)\to e^{ia(x_0+\delta x)}$.
This can also be seen in a different way, which sheds some light on the behavior of the superoscillatory sequence. Indeed,
 \[
 \begin{split}
F_n(x_0+&\delta x,a)=\Big(\cos \Big(\frac{x_0+\delta x}{n}\Big)+ia\sin \Big(\frac{x_0+\delta x}{n}\Big)\Big)^n\\
&=\Big(\cos \Big(\frac{x_0}{n}\Big)\cos \Big(\frac{\delta x}{n}\Big)-\sin \Big(\frac{x_0}{n}\Big) \sin \Big(\frac{\delta x}{n}\Big)\\
&+ia\Big[\sin \Big(\frac{x_0}{n}\Big)\cos \Big(\frac{\delta x}{n}\Big)+\cos \Big(\frac{x_0}{n}\Big)\sin \Big(\frac{\delta x}{n}\Big)\Big]\Big)^n\\
&=\Big\{\cos \Big(\frac{\delta x}{n}\Big)\Big[\cos \Big(\frac{x_0}{n}\Big) +ia \sin \Big(\frac{x_0}{n}\Big) \Big]
+ ia \sin \Big(\frac{\delta x}{n}\Big)\\
&\times
 \Big[ \frac{i}{a}\sin \Big(\frac{x_0}{n}\Big) +\cos \Big(\frac{x_0}{n}\Big) \Big] \Big\}^n\\
&=\Big\{\cos \Big(\frac{\delta x}{n}\Big)+ia \sin \Big(\frac{\delta x}{n}\Big) \left(\displaystyle\frac{\cos \Big(\frac{x_0}{n}\Big) +\frac{i}{a} \sin \Big(\frac{x_0}{n}\Big)}{\cos \Big(\frac{x_0}{n}\Big) +ia \sin \Big(\frac{x_0}{n}\Big)} \right) \Big\}^n
\\
&\times
\Big[\cos \Big(\frac{x_0}{n}\Big) +i a \sin \Big(\frac{x_0}{n}\Big) \Big]^n\\
&=\Big\{\cos \Big(\frac{\delta x}{n}\Big)+i\tilde a_n \sin \Big(\frac{\delta x}{n}\Big) \Big\}^n \Big[\cos \Big(\frac{x_0}{n}\Big) +i a \sin \Big(\frac{x_0}{n}\Big) \Big]^n\\
\end{split}
\]
where
$$
\tilde a_n= a\displaystyle\frac{\Big[ \cos \Big(\frac{x_0}{n}\Big) +\frac{i}{a} \sin \Big(\frac{x_0}{n}\Big) \Big]  \Big[ \cos \Big(\frac{x_0}{n}\Big) -ia \sin \Big(\frac{x_0}{n}\Big)\Big]}{\cos^2 \Big(\frac{x_0}{n}\Big) + a^2 \sin^2 \Big(\frac{x_0}{n}\Big)}
$$
which can also be written as
$$
\tilde a_n= a\displaystyle\frac{ 1 - \frac{i}{2} \frac{a^2-1}{a} \sin \Big(2\frac{x_0}{n}\Big) }{\cos^2 \Big(\frac{x_0}{n}\Big) + a^2 \sin^2 \Big(\frac{x_0}{n}\Big)};
$$
since $\tilde a_n\to a$ as $n\to\infty$, we reobtain the previous result and we see that the modulus of the limit function
 now grows as $a$ grows. The amplitude of the superoscillations decreases when $a$ increases.  We see that $a=1$ is a fixed point while if $a$ is large then we obtain large variations.
}
\end{remark}


\section{Test functions and their Fourier transforms}

In this section (which the reader may skip with no loss of continuity) we will review some fundamental facts about test functions and their Fourier transforms.

\begin{definition}\label{d1.2.14}\index{Fourier transform}
Let $\Omega\subseteq\rr$ be an open set, and let $\mathcal{C}^\infty(\Omega)$  be the set of complex valued infinitely differentiable functions on $\Omega$.
Let $\mathcal{E}(\Omega)$\index{$\mathcal{E}$} be
the topological linear space over $\mathbb{C}$ which, as a set, is
$\mathcal{C}^\infty(\Omega)$, and whose topology is given by \index{$\mathcal{C}^\infty (\Omega )$}
the family of seminorms defined as follows:
for any compact set $K\subset\Omega$ and for any $m\in\mathbb{N}\cup\{0\}$ we set
$$
p_{m}(K)(u):=\sup_{k\leq m} \sup_{x\in K}| u^{(k)}(x)|,\ \ \
u\in \mathcal{C}^\infty(\Omega).
$$
\end{definition}
\begin{remark}

\noindent (1)
The topology $\mathcal{T}_{p_{m}(K)}$ generated by the family of seminorms
$p_{m}(K)$ makes $(\mathcal{E}(\Omega),\mathcal{T}_{p_{m}(K)})$
\index{$(\mathcal{E}(\Omega),\mathcal{T}_{p_{m}(K)})$}
 a metrizable space; in fact
it is possible to prove that the family $p_{m}(K)$ is equivalent to a countable
family of seminorms $p_{m}(K_j)$, for $j\in \mathbb{N}$, and suitable compact sets $K_j$. Moreover,
with the topology generated by the seminorms
$p_{m}(K_j)$, the space $\mathcal{E}(\Omega)$ turns out to be complete.
Since the space $\mathcal{E}(\Omega)$ is metrizable and complete, it is a Fr\'echet space.

\noindent (2)
With the topology $\mathcal{T}_{p_{m}(K_j)}$,
the notion of convergence of a sequence in $\mathcal{E}(\Omega )$
to an element $u\in \mathcal{E}(\Omega )$ is the following:
\begin{quote}
$u_j\to u$ if $u_j$ converges to $u$ uniformly with all its derivatives on compact sets in $\mathbb{R}$.
\end{quote}
\end{remark}
\begin{definition}
The space $\mathcal{S}(\mathbb{R})$ of rapidly decreasing functions
\index{rapidly decreasing function}
\index{$\mathcal{S}$}
is defined as the space of functions $u\in \mathcal{C}^\infty(\mathbb{R})$ such that
$$
\sup_{x\in \mathbb{R}} \big( 1+|x|^2\big)^{m/2} |u^{(k)}(x)| <\infty ,
 \ {\rm for\ any}\ m\in \mathbb{N}\cup\{0\}
  \ {\rm and}\  {\it for\ any} \ k\in \mathbb{N}.
$$
The topology $\mathcal{T}_{p_m}$ on $\mathcal{S}(\mathbb{R})$
is given by
the family of seminorms
$$
p_{m}(u):=\sup_{k\leq m}
\sup_{x\in \mathbb{R}}\big( 1+|x|^2\big)^{m/2} |u^{(k)}(x)| <\infty ,
\ \ m\in \mathbb{N}\cup \{0\}.
$$
\end{definition}
\begin{remark}{\rm
The topology $\mathcal{T}_{p_m}$ makes
$\mathcal{S}(\mathbb{R})$
\index{$(\mathcal{S}(\mathbb{R}), \mathcal{T}_{p_m(u)})$}
 a Fr\'echet space in which the notion of convergence (we will limit ourselves to the case of convergence to
 $0$), is described as follows:
\begin{quote}{\rm
$u_j\to 0$ in $\mathcal{S}(\mathbb{R})$ if and only if
$
\big( 1+|x|^2\big)^{m/2} u_j^{(k)}(x)\to 0,\ \ {\rm as }\  j\to \infty$,
uniformly  in $\mathbb{R}$,
for any $k\in \mathbb{N}$ and any $ m\in\mathbb{N}\cup\{0\}$.}
\end{quote}
}
\end{remark}


As we shall see below, one of the most important reasons to define
the space $\mathcal{S}(\mathbb{R})$ is that the Fourier transform
maps this space into itself and the inverse operator is also well
defined. Moreover, since $\mathcal{S}(\mathbb{R})$ is dense both in
$\mathcal L^2(\mathbb{R})$ and in $\mathcal{S}'(\mathbb{R})$,
we can extend the Fourier transform to those spaces.
We will only prove the results of particular interest
in the context of this memoir,
and for the missing proofs of the results
stated in the following pages we refer the reader to  \cite{rudin}, \cite{vladimirov1},
\cite{yosida}.

\noindent

Since $\mathcal{S}(\mathbb{R})
\subset \mathcal L^1(\mathbb{R})$ we can give the following definition.

\begin{definition}
Let $\phi \in \mathcal{S}(\mathbb{R})$.
The linear integral operator $\mathcal{F}: \phi\to \hat\phi$ defined by
$$
\mathcal{F}[\phi ](\lambda)=\hat{\phi}(\lambda) :=
\int_{\mathbb{R}}\phi(x) e^{-i\lambda x}\, dx,\ \ \ i=\sqrt{-1}
$$
is called the Fourier transform.\index{Fourier transform} The function $\hat \phi$ is said to be the Fourier transform of $\phi$.
\end{definition}

The inverse of the Fourier transform is defined, for every
$\psi\in \mathcal{S}(\mathbb{R})$, by
\begin{equation}\label{derfora}
\mathcal{F}^{-1}[\psi ](x) :=\frac{1}{(2\pi)}\int_{\mathbb{R}}\psi(\lambda )
e^{i\lambda x}\, d\lambda.
\end{equation}
We observe that the operator $\mathcal{F}^{-1}$ is well defined on
$\mathcal{S}(\mathbb{R})$ and
for every $\psi \in \mathcal{S}(\mathbb{R})$ we have
\begin{equation}\label{derforb}
\psi=\mathcal{F}[\mathcal{F}^{-1}[\psi ]]=\mathcal{F}^{-1}[\mathcal{F}[\psi ]].
\end{equation}
\noindent
A fundamental result is the following:
\begin{theorem}\label{fouS}
The Fourier transform  $\mathcal{F}$ is an isomorphism between
$\mathcal{S}(\mathbb{R})$ and itself.
\end{theorem}

Thanks to Theorem \ref{fouS}  and  to the density of $\mathcal{S}(\mathbb{R})$
in the space of tempered distributions $\mathcal{S}'(\mathbb{R})$, it is possible
to define the operator $\mathcal{F}$ also on $\mathcal{S}'(\mathbb{R})$.
First, we consider a regular distribution $T_f$ defined by a function
$f\in \mathcal L^1_{{\rm loc}}(\mathbb{R})$ and we define
\begin{eqnarray*}
\langle \mathcal{F}[T_f], \phi\rangle
&=&\int_{\mathbb{R}}  \mathcal{F}[T_f ](\lambda) \phi(\lambda)\, d\lambda
\\
&=&\int_{\mathbb{R}}  \mathcal{F}[f](\lambda) \phi(\lambda)\, d\lambda,\ \ \
{\rm for \ any}\  \phi \in \mathcal{S}(\mathbb{R}).
\end{eqnarray*}
By Fubini's theorem, we can change the order of integration to get
\begin{equation}\label{trasfluno}
\langle \mathcal{F}[T_f], \phi\rangle =\langle T_f, \mathcal{F}[\phi]\rangle,\ \ \
{\rm for \ any}\  \phi \in \mathcal{S}(\mathbb{R}), \ \ f\in\mathcal  L^1_{{\rm loc}}(\mathbb{R}) .
\end{equation}
\noindent
Since the relation (\ref{trasfluno}) holds for regular distributions,
we can assume the following definition of Fourier transform for tempered distributions.
\begin{definition}
Let $T\in \mathcal{S}'(\mathbb{R})$.
We define the Fourier transform on the space of tempered distributions
\index{Fourier transform!of tempered distributions} by the equality
$$
\langle \mathcal{F}[T], \phi\rangle =\langle T, \mathcal{F}[\phi]\rangle,\ \ \
\ \ \ {\rm for \ any}\ \ \  \phi \in \mathcal{S}(\mathbb{R}).
$$
\end{definition}

The Fourier transform on the space of tempered
 distributions is well defined. In fact, since
 $\mathcal{F}$ is an isomorphism from $\mathcal{S}(\mathbb{R})$
to itself the map $\phi\to \mathcal{F}[\phi]$ is linear and continuous from
$\mathcal{S}(\mathbb{R})$ to $\mathcal{S}(\mathbb{R})$, the functional
$\langle T, \mathcal{F}[\phi]\rangle$ represents a distribution in
$\mathcal{S}'(\mathbb{R})$ and moreover the map $T\to \mathcal{F}[T]$ is linear and
continuous from $\mathcal{S}'(\mathbb{R})$ to $\mathcal{S}'(\mathbb{R})$.
\\
The inverse of the Fourier
transform in the space of tempered distributions is defined as
\begin{equation}\label{invftsp}
\langle \mathcal{F}^{-1}[T], \phi\rangle =\frac{1}{2\pi}\langle T, \mathcal{F}^{-1}[\phi]\rangle,\ \ \
\ \ \ {\rm for \ any}\ \ \  \phi \in \mathcal{S}(\mathbb{R}).
\end{equation}

\begin{theorem}\label{fouSp}
The Fourier transform  $\mathcal{F}$ is an isomorphism between
 $\mathcal{S}'(\mathbb{R})$ and itself.
\end{theorem}

\begin{example} {\rm
One can easily verify that, if $a\in \mathbb{R}$,
\begin{equation}\label{fdelta}
\mathcal{F}[\delta(x- a)]=e^{-i\lambda a}
\end{equation}
holds in the sense of distributions. In particular, one has
$$
\mathcal{F}[\delta(x)]=1
$$
and
$$
\mathcal{F}[1]=(2\pi) \delta(x).
$$
}
\end{example}

We now state some properties of the Fourier transform in $\mathcal{S}'(\mathbb{R})$, see e.g. \cite{vladimirov1}.
\begin{theorem}\label{FTS}
Suppose that $T\in \mathcal{S}'(\mathbb{R})$ and
$k\in \mathbb{N}$  and $a\in \mathbb{R}$. Then
the following properties hold:
\begin{itemize}
\item[(1)] Differentiation of the Fourier transform
$$
D^k \mathcal{F}[T]=\mathcal{F}[(-ix)^k T].
$$
\item[(2)] Fourier transform of the derivatives
$$
\mathcal{F}[T^{(k)}]=(i\lambda )^k \mathcal{F}[T].
$$
\item[(3)] Translation of the Fourier transform
$$
\mathcal{F}[T](\lambda +a)=  \mathcal{F}[e^{-iax}T](\lambda).
$$
\end{itemize}
\end{theorem}

\begin{remark}\label{utileext}{\rm
As  immediate consequences of Theorem \ref{FTS} we obtain
the Fourier transforms of polynomials
 $$
\mathcal{F}[x^k ]=
(2\pi)\, i^{k} \delta^{(k)}(x)
$$
and keeping in mind that
$\mathcal{F}[\delta(x)]=1$, we also obtain
$$
\mathcal{F}[ \delta^{(k)} ]=
(i\lambda)^k \mathcal{F}[\delta]=(i\lambda)^k.
$$
}
\end{remark}

\section{Approximations of functions in $\mathcal{S}(\mathbb{R})$ }
\label{sec5}
As it is well known, $\mathcal C^\infty$ functions cannot, in general, be represented by their Taylor expansion. Similarly, the value of a $\mathcal C^\infty$ function $f$ at a point $a$ cannot be obtained if one knows infinitely many values $f(x)$ for $x$ near the origin.
In one word, $\mathcal C^\infty$ functions are too flexible. There are two possible ways to add rigidity to their structure. On one hand, we could consider analytic functions, which are fully represented by the Taylor expansion; on the other hand, as we will do in this section, we could ask for the function $f$ to be a band limited test function. As we will show, it will be then possible to represent $f(x+a)$ through infinitely many values of $f$  at a distance at most $1$ from $x$.
\\
We now give a definition, whose meaning will be clarified later.
\begin{definition}[Standard approximating sequence]
Let  $\psi\in \mathcal{S}(\mathbb{R})$, $n\in \mathbb{N}$ and $a>1$. Let
$$
C_j(n,a)=\frac{(-1)^j}{2^n}{n\choose j} (a+1)^{n-j}(a-1)^{j}.
$$
We call
$$
\phi_{\psi, n,a}(x):=\sum_{j=0}^nC_j(n,a)\psi\left(x+\left(1-\frac{2j}{n}\right)\right)
$$
{\rm standard approximating sequence} of $\psi$.
\end{definition}
\noindent
In order to show that $\phi_{\psi, n,a}$ is in fact approximating the function $\psi$, we need the following lemma.
\begin{lemma}[Integral representation of the standard approximating sequence]\label{intrep}
Suppose that $\psi\in\mathcal{S}(\mathbb{R})$. Then we have
$$
\phi_{\psi, n,a}(x)=\frac{1}{2\pi}\int_{\mathbb{R}}F_n(\lambda,a)\hat{\psi}(\lambda) e^{i\lambda x}\ d\lambda,
$$
where $F_n(\lambda,a)=\sum_{j=0}^n C_j(n,a)e^{i(1-2j/n)\lambda}$.
\end{lemma}
\begin{proof} Taking the Fourier transform of
$$
\phi_{\psi, n, a}(x)=\sum_{j=0}^nC_j(n,a)\psi\left(x+\left(1-\frac{2j}{n}\right)\right)
$$
and using Theorem \ref{FTS}, one obtains
$$
\hat{\phi}_{\psi, n,a}(\lambda)= \sum_{j=0}^nC_j(n,a) e^{i\lambda(1-2j/n)}\hat{\psi}(\lambda).
$$
Observing that
$$
\hat{\phi}_{\psi, n,a}(\lambda)=F_n(\lambda,a)\hat{\psi}(\lambda),
$$
and taking the inverse Fourier transform, the required result follows.
\end{proof}
To state our next result we recall a definition.
\begin{definition}
A function $\psi\in\mathcal S(\mathbb R)$ is said to be band limited if $\hat\psi$ is compactly supported by some compact $K\subset\mathbb R$.
\end{definition}
\begin{theorem}\label{bandlim} For any band limited $\psi\in\mathcal S(\mathbb R)$ the limit
$$
\lim_{n\to \infty}\phi_{\psi, n,a}(x)=
\psi(x+a)
$$
is uniform on every compact set in $\mathbb R$.
Moreover
$$
|\phi_{\psi, n,a}(x)- \psi(x+a)|\leq \varepsilon_n\int_{\mathbb{R}}  \,|\hat{\psi}(\lambda)| \ d\lambda,
$$
where
$$
\varepsilon_n:=\sup_{\lambda\in K}\mathcal{E}_n(\lambda,a).
$$
\end{theorem}
\begin{proof}
Since $\hat\psi$ is compactly supported by $K\subset\mathbb R$, the Lebesgue dominated convergence theorem implies that
\[
\begin{split}
\lim_{n\to \infty}\phi_{\psi, n,a}(x)&=\lim_{n\to \infty}\frac{1}{2\pi}\int_{\mathbb{R}}F_n(\lambda,a)\hat{\psi}(\lambda) e^{i\lambda x}\ d\lambda
\\
&
=\frac{1}{2\pi}\int_{\mathbb{R}}e^{ia\lambda }\hat{\psi}(\lambda) e^{i\lambda x}\ d\lambda
\\
&=
\psi(x+a).
\end{split}
\]
Lemma \ref{intrep} allows to estimate the error:
\[
\begin{split}
|\phi_{\psi, n,a}(x)- \psi(x+a)|&\leq
\frac{1}{2\pi}\int_{\mathbb R}|F_n(\lambda,a)-e^{ia\lambda}|d\lambda
\\
&
\leq
\sup_{\lambda\in K}\mathcal{E}_n(\lambda,a)  \int_{\mathbb{R}}  \,|\hat{\psi}(\lambda)| \ d\lambda
\\
&
=
\varepsilon_n\int_{\mathbb{R}}  \,|\hat{\psi}(\lambda)| \ d\lambda.
\end{split}
\]

\end{proof}
\begin{corollary}\label{3:cor1}
For any band limited $\psi \in\mathcal S (\mathbb R)$  we have that
$$
\lim_{n\to \infty}\phi^{(k)}_{\psi, n,a}(x)=
\psi^{(k)}(x+a)
$$
for any $k\in\mathbb N$, uniformly on every compact set in $\mathbb R$.
Moreover
$$
|\phi^{(k)}_{\psi, n,a}(x)- \psi^{(k)}(x+a)|\leq \varepsilon_{n,k}\int_{\mathbb{R}}  \,|\hat{\psi}(\lambda)| \ d\lambda,
$$
where
$$
\varepsilon_{n,k}:=\sup_{\lambda\in K}|\lambda|^k\mathcal{E}_n(\lambda,a).
$$

\end{corollary}
\begin{proof}
 This is an immediate consequence of the fact that if $\psi$ is  a band limited element of $\mathcal S(\mathbb R)$ also $\psi^{(k)}$ is band limited.
\end{proof}
\begin{corollary}\label{3:cor2}
 Let $a_1,\ldots , a_p$ be real numbers such that $|a_k|>1$ and let $\delta_{-a_k}(x) =\delta(x+a_k)$ represent the delta centered at the point $a_k$.
Then
$$
\lim_{n\to \infty} \sum_{k=1}^p \phi_{\psi, n,a_k}(x)=
\sum_{k=1}^p \delta_{-a_k} * \psi (x)
$$
uniformly on every compact set in $\mathbb R$.
Moreover, we have
$$
|\sum_{k=1}^p \phi_{\psi, n,a_k}(x)-
\sum_{k=1}^p \delta_{-a_k} * \psi (x)|\leq p\varepsilon_{n}\int_{\mathbb{R}}  \,|\hat{\psi}(\lambda)| \ d\lambda.
$$
\end{corollary}

The following result is a direct consequence of Lemma \ref{intrep}.
\begin{proposition}\label{6.5}  Let us fix $\gamma >0$,  $a>1$ and $n\in \mathbb{N}\!$. Suppose that $\psi\in\mathcal{X}_\gamma(\mathbb{R})$ where
$$
\mathcal{X}_\gamma(\mathbb{R}):=\{\, u\in \mathcal{S}(\mathbb{R})  \ : \ \int_{\mathbb{R}}| \hat u(\lambda)| \ d\lambda  \leq  \gamma\, \}.
$$
Then
$$
|\phi_{\psi, n,a}(x)- \psi(x+a)|\leq (1+a^n) \gamma.
$$
\end{proposition}
\begin{remark}{\rm
It is clear that one can generalize Proposition \ref{6.5} to the case of the derivatives of $\psi$ as well as to the case of finite sums of its translates, just like we did in the two corollaries of Theorem \ref{bandlim}.
}
\end{remark}

Even though the above inequality is not helpful when $n\to\infty$ since $a>1$, it is still interesting because it shows a uniform distance between $\phi_{\psi, n,a}(x)$ and $\psi (x+a)$ for all $x\in\mathbb{R}$. It is however helpful for some families of functions as the next example shows.

\begin{example}\label{exutil}
{\rm
It is well known that the family of functions $$u_\alpha(x)= xe^{-\alpha x^2}$$ belongs to $ \mathcal{S}(\mathbb{R})$ for every $\alpha >0$, see for example \cite{hormander}.
Their Fourier transforms are given by
$$
\hat{u}_\alpha(\lambda)=-\frac{\lambda  i}{2\alpha}\sqrt{\frac{\pi}{\alpha}}e^{-{\lambda^2}/{4\alpha}}.
$$
Now observe that
$$
\int_{\mathbb{R}}|\hat{u}_\alpha(\lambda)| d\lambda=2\sqrt{\frac{\pi}{\alpha}},
$$
and so the number $\gamma$ appearing in Proposition \ref{6.5} equals $2\sqrt{\dfrac{\pi}{\alpha}}$.
Thus
$$
|\phi_{u_\alpha, n}(x)-u_\alpha(x+a)|\leq \frac{1}{\pi}(1+a^n) \sqrt{\frac{\pi}{\alpha}}
$$
where
$$
\phi_{u_\alpha, n}(x):=\sum_{j=0}^nC_j(n,a)u_\alpha(x+(1-2j/n)).
$$
Fix $a\in \mathbb{R}$, $n\in \mathbb{N}$ and $\varepsilon >0$.  The elements of the family $u_\alpha(x)= xe^{-\alpha x^2}$ such that
$$
|\phi_{u_\alpha, n}(x)-u_\alpha(x+a)|\leq \varepsilon \ \ \ {\rm for\ every \ } \ \ \ x \in \mathbb{R},
$$
namely the elements of the family for which Theorem \ref{bandlim} holds,
are those for which $\alpha >\alpha_0$ where $\alpha_0$ satisfies the equality
\begin{equation}\label{alfa0}
(1+a^n) \sqrt{\frac{1}{\pi\alpha_0}}= \varepsilon.
\end{equation}
}
\end{example}
\begin{remark}{\rm If we had considered the gaussian $v_\alpha(x)=e^{-\alpha x^2}$ instead of $u_{\alpha}(x)$,
the constant $\gamma$ appearing in Proposition \ref{6.5} cannot be chosen arbitrarily small. Indeed, $\alpha$ cannot be chosen such that
$$
|\phi_{v_\alpha, n}(x)-v_\alpha(x+a)|\to 0
$$
when $\alpha$ varies in $\mathbb R^+$.
}
\end{remark}
\newpage

\chapter{Function spaces of holomorphic functions with growth}
The superoscillating functions and sequences that we have considered so far are real analytic functions defined on $\mathbb{R}$, which, because of their specific shape, can actually be extended as entire functions of a complex variable if we (formally) replace the real variable $x$ with the complex variable $z$. This is possible because the superoscillating sequences
we are studying are {\it finite} sums of exponentials, and their complex extensions are not simply holomorphic in a neighborhood of $\mathbb{R}$, but are actually entire. As it will be apparent in the remainder of this monograph, the fact that we can now use the powerful theory of holomorphic functions is central to many of the results we will obtain.

Before recalling, and reframing, the set-up for this analysis, we need to point out that while the exponentials that appear in the superoscillating sequences have frequencies bounded by $1$, since all the exponentials that appear are of the form $e^{i \lambda x}$, with $|\lambda|\leq 1$, and $x$ real, this feature disappears when we consider the entire extension of such functions. Indeed, functions of the form $e^{i \lambda z}$ are entire and of exponential type (and order $1$). Thus, we are naturally led to the study of entire functions with growth conditions at infinity. As we will show later, the study of superoscillating sequences will entail studying certain convolution operators on these spaces, and the appropriate setting for such a study is the theory of Analytically Uniform spaces, introduced by Ehrenpreis in the sixties, and fully developed in \cite{ehrenpreis}. Ehrenpreis' theory was originally designed to be a tool for the study of systems of linear constant coefficients partial differential equations, but Berenstein, Taylor, and their coauthors were able to show \cite{bd}, \cite{bt1}, \cite{bt}, \cite{struppamemoirs} that convolution equations could be treated as well within this framework: the reader is invited to consult \cite{bs1} for a rather comprehensive review on this topic. For these reasons, we will devote Section $4.1$ to a rather detailed description of the theory of Analytically Uniform spaces, and then Section 4.2 to the study of convolution operators, and convolution equations, in such spaces.

\section{Analytically Uniform spaces}

In this section we will study some spaces of (generalized) functions that were introduced by Ehrenpreis, and that he called Analytically Uniform spaces  (AU-spaces in short).  Before we offer the formal (and somewhat complicated) definitions, we should describe the philosophy of such a notion. Without aiming for completeness, one could say that a locally convex topological vector space $X$ is said to be an AU-space if its strong dual $X'$ is topologically isomorphic to a space of entire functions which satisfy suitable growth conditions at infinity.

 Let us offer both some clarification, as well as the rationale for this particular approach. First of all the nature of the growth conditions that are acceptable in this framework is fairly complex, will be fully described below, but essentially requires that the resulting spaces are algebras of functions closed under differentiation. Second, the topological isomorphism between $X'$ and the appropriate space of entire functions with growth conditions is usually a variation of the Fourier (or Fourier-Borel) transform, and the reason for this is apparent when one considers the (historically) first use of this idea.

\bigskip
Suppose we consider a linear constant coefficients partial differential operator $\wp(D)$ acting on the space ${\mathcal E}(\Omega)$ of infinitely differentiable functions on an open convex set in $\mathbb{R}^n.$ Suppose now we want to consider the surjectivity of such an operator on ${\mathcal E}(\Omega).$ The approach, independently pioneered by Ehrenpreis, Malgrange, and Palamodov \cite{e1}, \cite{e2}, \cite{e3}, \cite{ehrenpreis}, \cite{malgrange}, \cite{palamodov} consists in looking at the map
$$
\wp(D):{\mathcal E}(\Omega) \to {\mathcal E}(\Omega)
$$
and notice that by standard functional analysis results, such a map is surjective if and only if its adjoint
$$
\wp(-D):({\mathcal E}(\Omega))' \to ({\mathcal E}(\Omega))'
$$
is injective and has closed range (this is why the topologies are important). However, one now notes that the Fourier transform of $\wp(-D)$ is simply the multiplication operator by the polynomial $\wp(-z)$, and that the Fourier transform acts as an isomorphism between $({\mathcal E}(\Omega))'$ and the space
$$
\mathcal{A}_{\Omega}:=\widehat{({\mathcal E}(\Omega))'}=\{F\in {\mathcal H}(\mathbb{C}^n) \, : \, |F(z)| \leq A \exp (H_K(z))\}
$$
of entire functions whose growth is bounded at infinity by the supporting function $H_K(z):=\sup_{w\in K} (w\cdot z)$ of any compact subset $K$ of $\Omega$, where $w\cdot z$ denotes the scalar product of $w$ and $z$. So, see \cite{malgrange}, the surjectivity of $\wp (D)$ on the space of infinitely differentiable functions is equivalent to the injectivity of multiplication by $\wp(-z)$ on $\mathcal{A}_{\Omega}$ (which is trivial) and to the fact that the ideal generated by $\wp$ is closed in $\mathcal A_\Omega$. That this is the case, for $\Omega$ convex, was proved by Malgrange in \cite{malgrange}, where he explicitly gives the condition on the pair $(\wp,\Omega)$ that will generate surjectivity when $\Omega$ is not necessarily convex.

This important example shows both the necessity of identifying $\mathcal E'(\Omega)$ with a space of entire functions satisfying suitable conditions, and the necessity of having this discussion within the framework of topological vector spaces.

After the following technical definition we introduce the notion of AU-space:
\begin{definition}
Let $\mathcal{K}$ be a nonempty set of positive continuous functions on
$\mathbb{C}^n$. We say that
$\mathcal{K}$ is an analytically uniform structure (AU-structure)
if for
 any $k\in \mathcal{K}$ there exists $k'\in \mathcal{K}$ such that
$$
k'(z+z')(2+|z|^2)\leq k(z),\ \ \ {\rm for\ all}\ z, z'\in \mathbb{C}^n, \ |z'|\leq 1.
$$
\end{definition}

\begin{definition}
Let $X$ be a locally convex space. Assume that there exists an AU-structure $\mathcal{K}$ and a componentwise continuous bilinear form
$\langle\cdot ,\cdot \rangle$ on $W\times A_{\mathcal{K}}$ such that the map
from $X$ to the strong dual of  $A_{\mathcal{K}}$ given by
$$
\omega \to \langle\omega ,\cdot \rangle
$$
is a topological isomorphism. Then we call $X$ an AU-space.
\end{definition}

\begin{example}\label{differentiable}
{\rm Let $\mathcal{E}$ denote the space of infinitely differentiable functions on $\rr^n$ with the usual topology of uniform convergence on compact subsets of $\rr^n$. Then $\mathcal{E}$ is an AU-space and its dual is the space $\mathcal{E}'$ of compactly supported distributions. By the Paley-Wiener Theorem, this space is topologically isomorphic to the space of entire functions of exponential type, which grow like polynomials on the real axis. Specifically, the space of entire functions $F$ such that, for some positive constants $A, B$,
$$
|F(z)| \leq A(1+|z|)^B \exp(B|{\rm Im} (z)|).
$$
}
\end{example}

\begin{example}
{\rm Let $\mathcal{D}$ denote the space of infinitely differentiable functions with compact support on $\rr^n.$ Let $\mathcal{D}'$ denote its dual, namely the space of Schwartz distributions (note that $\mathcal{D}$ is properly contained in $\mathcal{E}$, while $\mathcal{D}'$ properly contains $\mathcal{E}'$). Then $\mathcal{D}'$ is an AU-space and its dual $\mathcal{D}$ is, by the Paley-Wiener-Schwartz Theorem, isomorphic to the space of entire functions of exponential type, which, on the real axis, go to zero faster than the inverse of any polynomial. Namely those entire functions $F$ such that, for some positive constants $A, B$, and for all positive constants $C$,
$$
|F(z)| \leq A (1+|z|)^{-C}\exp(B|{\rm Im}( z)|).
$$}
\end{example}

\begin{example}\label{entire}
{\rm Let $\mathcal{H}$ denote the space of entire functions on $\cc^n$. Then $\mathcal{H}$ is an AU-space and its dual is the space $\mathcal{H'}$ of linear analytic functionals, which, by another formulation of the Paley-Wiener Theorem, is isomorphic to the space of entire functions of exponential type, i.e. the entire functions $F$ satisfying, for some positive constants $A, B$,
$$
|F(z)| \leq A\exp B|z|.
$$
This space is interesting per se, but also because it exemplifies an interesting property of AU-spaces. Specifically, if $X$ is an AU-space, then the subspace of $X$ obtained as the kernel in $X$ of a system of linear constant coefficients partial differential operators, is still an AU-space. In this case, $\mathcal{H}(\cc^n)$ is the kernel in $\mathcal{E}(\rr^{2n})$ of the Cauchy-Riemann system of differential operators.}
\end{example}

\begin{example}
{\rm In the three examples above, we have considered functions defined on the entire Euclidean space $\rr^n$ or $\cc^n.$ However, a modification of the weights allows us to show that if we consider an open convex set $\Omega$ in $\rr^n$ or in $\cc^n$, then the corresponding spaces $\mathcal{E}(\Omega),  \mathcal{D'}(\Omega),$ and $\mathcal{H}(\Omega)$ are also AU-spaces. This is what we used to discuss the surjectivity of $\wp(D):\mathcal E(\Omega) \to \mathcal E(\Omega)$. If $\Omega$ is not convex, then the corresponding spaces are not AU-spaces anymore. }
\end{example}

While these examples are very important, it is equally important to note that many significant spaces in the theory of differential equations are not AU-spaces. For example, none of the $\mathcal{L}^p$ spaces is an AU-space, and neither are $\mathcal{S}$ and $\mathcal{S'}$, the spaces of Schwartz test functions and its dual, the space of tempered distributions. Finally, another space that (quite surprisingly) fails to be an AU-space is the space of real analytic functions, see \cite{bd}.

In some of the situations described in the examples above, the growth in the space of entire functions is described by a class of weights
which can be all brought back to an individual fundamental weight. This is a situation of particular interest to us, which we want to describe in more detail.

\begin{definition}
We say that a plurisubharmonic function $p:\cc^n \to \rr^+$ is an admissible weight\index{weight!admissible} if
\begin{itemize}
\item $\log(1+|z|^2)=O(p(z))$
\item there are four positive constants $A_1,\ldots,A_4$ such that  if
$$|z_1-z_2| \leq A_1|z_1|+A_2$$
then
$$p(z_1) \leq A_3p(z_2)+A_4.$$
\end{itemize}
\end{definition}

There are some important examples of such weights, which need to be immediately highlighted:

\begin{example} {\rm
The function $p(z)=|z|$ is an admissible weight. And, for any $\rho >1 $, the function $p(z)=|z|^{\rho}$ is a weight as well. Weights of this kind are called {\it radial weights}, as they only depend on the modulus of $z$.}
\end{example}

\begin{example} {\rm
The function $p(z)=|{\rm Im} (z)| +\log(1+|z|)$ is a weight as well. It is clearly not a radial weight.}
\end{example}

If $p(z)$ is an admissible weight, we can define two special kinds of space of entire functions with growth conditions:
the space $\mathcal{A}_p(\cc^n)$ defined by
$$
\mathcal{A}_p(\cc^n):=\{f \in \mathcal{H}(\cc^n) \ : \  \exists \ A>0, \ \ B>0 \ : \ \ |F(z)| \leq A \exp(Bp(z)) \}
$$
and the space
$$
\mathcal{A}_{p,0}(\cc^n):=\{f \in \mathcal{H}(\cc^n) \ : \ \forall \varepsilon >0, \ \exists \ A_{\varepsilon}>0 \ : \
|f(z)| \leq A_{\varepsilon} \exp (\varepsilon p(z))\}$$
In the specific case in which the weight is the function $p(z)=|z|^\rho$ we introduce a few additional definitions:
\begin{definition}
The space $\mathcal{A}_\rho(\cc^n)$ is defined by
$$
\mathcal{A}_\rho(\cc^n):=\{f \in \mathcal{H}(\cc^n) \ : \ \  \exists \ A>0,\ \ B>0 \ : \ |F(z)| \leq A \exp(B|z|^\rho) \}
$$
and it is called the space of entire functions of order less or equal to  $\rho$ and of finite type.
\end{definition}
The space $\mathcal{A}_\rho(\cc^n)$ can be endowed with  a natural topology which can be described as follows.
Let
$$
q_j(f)=\sup_{z\in\mathbb C^n}|f(z)| e^{-\sigma_j(z)}, \qquad j=1,2,\ldots,
$$
where $\sigma_j(z)$ are suitable weights.
Then $q_j$ is a seminorm and the family $\{q_j\}$ determines a locally convex topology on $\mathcal{A}_\rho(\cc^n)$.
In terms of AU-structures, one can introduce
the sequence $\Sigma=\{\sigma_j\}_{j=1}^\infty$ where the weights are $\sigma_j(z)=j |z|^p$ and set
$$
\mathcal K=\{k\in\mathcal C^0(\mathbb C^n)\ : \ k(z)=\sup_j \delta_je^{\sigma_j(z)}<\infty , \, \forall z\in\mathbb C^n\}
$$
where $\{\delta_j\}$ is a sequence of positive real numbers. Then $\mathcal K$ gives the so-called AU-structure on
$\mathcal{A}_\rho(\cc^n)$ which can therefore be interpreted as the locally convex space
\begin{equation}\label{AK}
\mathcal A_{\mathcal K}=\left\{ f\in\mathcal H (\mathbb C^n)\ :\
\frac{|f(z)|}{k(z)}\to 0,\ {\rm as}\ z\to\infty, \ \forall k\in\mathcal K \right\}.
\end{equation}

\begin{remark}\label{rmkstar}{\rm
Note that if $p(z)=|z|$, i.e. $\rho=1$, then the space $\mathcal{A}_1(\cc^n)$ is isomorphic, via Fourier-Borel transform, to the space $\mathcal{H}'$ of analytic functionals as described in Definition \ref{entire}. Similarly note that if $p(z)=|{\rm Im}( z)|+\log(1+|z|)$, then the space $\mathcal{A}_p(\cc^n)$ is isomorphic, via Fourier transform, to the space $\mathcal{E}'$ of compactly supported distributions as described in Definition \ref{differentiable}.}
\end{remark}
\begin{definition}
The space $\mathcal{A}_{\rho,0}(\cc^n)$ is defined by
$$
\mathcal{A}_{\rho,0}(\cc^n):=\{f \in \mathcal{H}(\cc^n) \ : \ \forall \varepsilon >0, \exists A_{\varepsilon}>0 \ : \
|f(z)| \leq A_{\varepsilon} \exp (\varepsilon |z|^\rho)\},$$
and it is called the space of entire functions of order less or equal $\rho$ and of minimal type.
\end{definition}
The space $\mathcal{A}_{\rho,0}(\cc^n)$ is endowed with its natural topology.
In terms of AU-structures, it is the locally convex space $\mathcal A_{\mathcal K^*}$, see (\ref{AK}), where
$$
\mathcal K^*=\{k^*\in\mathcal C^0(\mathbb C^n)\ : \ k^*(z)\geq 0, \ k(z)k^*(w)e^{z\cdot w}\, {\rm is\, bounded}\ \forall k\in\mathcal K\}.
$$

\begin{remark}\label{rmkcirc}{\rm
If we take $p(z)=|z|$, we see that $\mathcal{A}_{p,0}=\mathcal{A}_{1,0}$ is the space of the so-called functions of {\it infraexponential type},
\index{infraexponential type!function}or of order one, and type zero, and it is isomorphic, via Fourier-Borel transform, to the space of analytic functionals carried by the origin (which in turn coincides with the space of hyperfunctions supported at the origin). This identification is the key element in the theory of infinite order differential operators.}
\end{remark}

Once these spaces have been defined, a simplified notion of AU-space requires that the dual $X'$ of $X$ is topologically isomorphic to either a space $\mathcal{A}_p$ or a space $\mathcal{A}_{p,0}$, for a suitable admissible weight $p$. That this is the case was established in \cite{struppamemoirs}.

We will now restrict our attention to the case of one complex variable, and we will consider $p(z)=|z|^\rho$ and $q(z)=|z|^{\rho'}$ where $\rho>1,\rho'>1$ are real numbers such that $\frac{1}{\rho}+\frac{1}{\rho'}=1$.
Let us denote by $\mathcal{A}_\rho(\cc)'$ the strong dual of $\mathcal{A}_\rho(\cc)$, namely the space of continuous linear functional on $\mathcal{A}_\rho\mathbb(\cc)$ endowed with the strong topology. Define the Fourier-Borel transform of $\mu\in \mathcal{A}_\rho(\cc)'$ as the entire function
$$
\hat{\mu}(w)=\mu(\exp(-z\cdot w)), \qquad z\in\mathbb C .
$$
Then we have the following duality results which we will be useful in the sequel:
\begin{theorem}\label{iso} Let $\rho,\rho'\in\mathbb R$, $\rho >1$, $\rho'>1$ be such that
$$
\frac{1}{\rho} + \frac{1}{\rho'}=1.
$$
The following isomorphisms
$$
\widehat{\mathcal{A}_\rho(\cc)'}\cong \mathcal{A}_{\rho',0}(\cc)
$$
and
$$
\widehat{\mathcal{A}_{\rho,0}(\cc)'}\cong \mathcal{A}_{\rho'}(\cc),
$$
are algebraic and topological as well.
\end{theorem}
Ehrenpreis demonstrated in his monograph \cite{ehrenpreis} that the AU-spaces are the ideal setting for a very general theory of linear constant coefficients partial differential equations. He demonstrated also, though this part of his work was left unfinished until \cite{bt}, \cite{struppamemoirs}, that a general theory of systems of convolution equations was possible in such spaces. To see how that would be done, we will introduce convolutors on AU-spaces in an abstract way, and we will then show how they can be realized concretely in the spaces described before.

\bigskip
The importance of the spaces $\mathcal A_{\rho}$, $\mathcal A_{\rho,0}$ does not, clearly, exhaust the class of interesting spaces. We therefore close this section with two more results from \cite{taylor} that will be useful in the sequel. To this purpose we need the following notation:
  $$\Delta_R=\{z\in\mathbb{C}\ \ :\ \ |z|< R\}.$$
\begin{proposition}\label{esempiodelta1}
The space
\[
 X={\rm Exp}_1(\mathbb{C}):=\{ f\in\mathcal{H}(\mathbb{C}) \ : \ |f(z)|\leq A |z|^n e^{|z|}\ {\it for\ some}\ n\}
\]
is an AU-space and
\[
 \mathcal{F}X'\cong\{f\in\mathcal{H}(\Delta_1)\ {\it and}\, f^{(n)}\ {\it are\ uniformly\ continuous\ for\ all}\, n\in\mathbb{N}\}.
\]
\end{proposition}

\begin{proposition}\label{esempiodeltaR}
 Let $R>0$ and let $X$ be the space of entire functions of exponential type less that $R$,
i.e.
\[
 X=\{f\in\mathcal{H}(\mathbb{C}) \ : \ \exists \ \varepsilon >0 \ :  \ |f(z)|\leq A e^{(R-\varepsilon)|z|} \ {\it for \ some\,} A >0\}.
\]
Then $X$ is an AU-space and $\mathcal{F}X'$ is isomorphic to the space of functions holomorphic in the disc
$\Delta_R$.
\end{proposition}

\section{Convolutors on Analytically Uniform spaces}

Let $X$ be an AU-space. Convolutors on $X$ can be defined in a standard way as follows: let $X'$ be the strong dual of $X$ and let $\mathcal{F} X'$ be the space of entire functions satisfying suitable growth conditions, which is topologically isomorphic to $X$.
\begin{definition}\label{convolutor}
Let $F$ be an entire function which, by multiplication, defines a continuous map from $\mathcal{F} X'$ to itself. Then a convolutor on $X$ is the continuous operator on $X$ defined as the adjoint of the map that associates to $f\in X'$
the element
$$
\mathcal{F}^{-1}(F(\mathcal{F}(f))).
$$
\end{definition}
Let us consider a few examples to illustrate the concrete meaning of this definition.
\begin{example}{\rm
Let $\mathcal{E}$ be the space of infinitely differentiable functions on $\rr^n$, whose dual $\mathcal{E'}$ is the space of compactly supported distributions. This space, as we have seen before, is topologically isomorphic to the space $\mathcal{A}_p(\cc^n)$, where $p(z)=|{\rm Im} z|+\log(1+|z|)$ (see Remark \ref{rmkstar}). Then every holomorphic function $F$ in $\mathcal{A}_p$ is the Fourier transform of a compactly supported distribution $\mu \in \mathcal{E'}$ and the convolution described abstractly in Definition \ref{convolutor} is actually the convolution $\mu*f$ between a compactly supported distribution $\mu$, and a differentiable function $f$, and which is well known from classical analysis. Specifically,
$$
\mu*f(x)=<\mu,t \to f(x+t)>.
$$
}
\end{example}
A similar argument can be made for the holomorphic case.
\begin{example}{\rm
If $\mathcal{H}$ is the space of entire functions on $\mathbb{C}$ then its dual is the space of analytic functionals which, by Fourier-Borel transform, is topologically isomorphic to the space $\mathcal{A}_p(\cc^n)$ with $p(z)=|z|.$ Then, every entire function of exponential type defines a multiplicator on $\mathcal{A}_p(\cc^n)$ and it is the Fourier-Borel transform of an analytic functional $\mu \in \mathcal{H'}$. Thus, just as in the previous example, the abstract convolution in Definition \ref{convolutor} is nothing but the convolution $\mu*f$ between an analytic functional and an entire function, already classically known.
}
\end{example}

A very special class of convolutors on the space of entire functions is the class of infinite order differential operators, which arise when we consider an analytic functional $\mu$ carried by the origin. The reason for the nomenclature of {\it infinite order differential operator} stems from the fact that the Fourier-Borel transform of an analytic functional carried by the origin is an entire function of infraexponential type (i.e. it belongs to the space $\mathcal{A}_{p,0}$, with $p(z)=|z|$, see Remark \ref{rmkcirc}), and therefore its Taylor expansion converges everywhere on $\cc$ and its action can truly be considered as the action of a differential operator of infinite order.
Specifically, we can give the following definition, see \cite{kaneko}.
\begin{definition} \label{infiniteorderdef} An operator of the form
$$
\sum_{m=0}^{\infty} b_m(z) \displaystyle\frac{d^{m}}{dz^{m}}
$$
is an infinite-order differential operator,  which acts continuously on holomorphic functions in $\mathbb{C}$ if and only if, for every compact set ${K}\subset \mathbb{C}$,
\begin{equation}\label{infiniteorder}
\lim_{k\to \infty}\sqrt[k]{\sup_{z\in {K}} |b_k(z) |\, k!  }=0.
\end{equation}
\end{definition}
\vskip 0.5cm
While infinite order differential operators have great importance both per se, and as part of the theory of hyperfunctions and microfunctions (see e.g. \cite{kaneko}, \cite{katostruppa}), there are many instances when we want to consider operators that may be defined in ways that remind us of infinite order differential operators, while they do not satisfy  condition (\ref{infiniteorder}). This is, classically, the case for the translation operator, which is usually defined as
\begin{equation}\label{translation}
\tau f(x) = f(x+1)=\exp(d/dx)f(x)=\sum_{m=0}^{\infty} \frac{1}{m!} \frac{d^mf}{dx^m}(x).
\end{equation}
While the operator appears to be expressed as an infinite sum of derivatives, the function on the right hand side of (\ref{translation}) does not converge, in general, and in fact makes no sense, except in an intuitive way. To be precise, the translation is only a convolutor on the space, say, of entire functions. Operators of this kind are quite important (in fact, the case of the translation is paradigmatic of what physicists use continuously, sometimes with some abuse of notations), and therefore one may ask whether they can be treated in a general way. That this is the case is guaranteed by Theorem \ref{iso}, which can now be rephrased as follows:
\begin{theorem}\label{duality pq}
Let $\rho$, $\rho'$ be real numbers such that $\rho\geq 1$, $\rho'\geq 1$ and
 $$
 \frac{1}{\rho}+\frac{1}{\rho'}=1.
 $$
The space of convolutors on $\mathcal{A}_\rho$ is isomorphic (via Fourier-Borel transform) to the space $\mathcal{A}_{\rho',0}$  and conversely, the space of convolutors of $\mathcal{A}_{\rho,0}$ is isomorphic to $\mathcal{A}_{\rho'}$.
\end{theorem}
\vskip 0.5cm

\section{Dirichlet series}
\label{dirichlet}
We conclude this chapter by noticing that it may be worthwhile to consider a more general setting, in which holomorphicity is required only on a subset of $\mathbb{C}.$ Let us give a hint of why such a setting might be of interest. Indeed, up to now, we have considered sequences of band-limited functions, converging to exponentials of the form $e^{iax}$ with $a$ large (at least larger than $1$). But one may want to consider the case in which we begin with a Dirichlet series $\sum_{j=1}^{\infty} c_je^{- i\lambda_j x},$ with real frequencies $\lambda_j, \, \lambda_j \to \infty.$  In this case we will show that it may be interesting to consider sequences of band limited functions which converge to such series. As it is well known, however, Dirichlet series only converge in half-planes, and so it is relevant to generalize the study of Sections 4.1 and 4.2 to the case of functions which are holomorphic in cones (half-planes being a special case of cones), and satisfying there suitable growth conditions. As we will show in this section it is possible to prove that, with natural adjustments, the fundamental results of Sections $4.1$ and $4.2$ extend to this more general setting. Those results are important if one wishes to consider more general superoscillating sequences, and study how they propagate when taken as initial values of suitable differential equations.

\bigskip

Let $\Gamma$ be an open convex cone  with vertex at the
origin, for the sake of simplicity, and contained in the right half plane $\Pi_+=\{z\in\mathbb C\ :\ {\rm Re}(z)>0\}$.
We also assume that $\Gamma$ is of the form
$$
\Gamma=\{z\in\mathbb C\  :\  - \theta <{\rm arg}(z)<\theta\},
$$
for some $\theta\in(0, \pi/2]$.
In the particular case $\theta=\pi/2$ the cone $\Gamma$ coincide with the right half plane, i.e. $\Gamma=\Pi_+$.

\bigskip
Let $\mathcal{A}_{p,0}(\Gamma)$ be the space of function $f\in\mathcal{H}(\Gamma)$ such that,
for all $\varepsilon>0$ and for all cones $\Gamma'$ compactly included in $\Gamma$,
$$
|f(z)|\leq C\exp(\varepsilon p(z)),\quad z\in\Gamma',
$$
for $z\in\Gamma'$, for some constant $C=C(\varepsilon,\Gamma',f)>0$. It is standard to
endow $\mathcal{A}_{p,0}(\Gamma)$ with its natural projective limit topology.
As we did in the previous sections, it is important to the characterize the dual of this space, namely the space of linear
continuous functionals on  $\mathcal{A}_{p,0}(\Gamma)$. Any linear continuous
functional in $\mathcal{A}_{p,0}(\Gamma)'$ can be described as an integral against a
measurable function $u$, supported in some cone
$$
K=\{z\in\Pi_+\  :\  -\alpha \leq {\rm arg}(z)\leq \alpha\}+c, \ c\in\mathbb R^+, \, \alpha\in(0,\theta),
$$
and such that, for some $A>0$ and some $\tilde C >0$,
$$
| u(z) |\leq \tilde C \exp(- Ap(z)) , \ z\in K.
$$
By its definition, $K$ is a cone with vertex at a positive real point and has opening less than the opening of $\Gamma$, thus
$K$ is strictly included in $\Gamma$.
\\
To any $u\in \mathcal{A}_{p,0}(\Gamma)'$, we can associate, though not in a unique way, a value
$c\in\mathbb R^+$ and an angle $\alpha\in (0,\theta)$. If $u$ is supported by $K=K(\alpha, c)$, then $u$ is
supported by any cone $K'=K'(\alpha', c')$ with $\alpha\leq\alpha'$ and $c'\leq c$.
To define the Laplace transform of a functional $u\in \mathcal{A}_{p,0}(\Gamma)'$
we first notice that the functions
$$
e_w(z)=\exp(-z\cdot w)
$$
belong (as functions of $z$) to $\mathcal{A}_{p,0}(\Gamma)$, for every $w\in \mathbb C$. Thus
the Laplace transform on $\mathcal{A}_{p,0}(\Gamma)'$ is defined by the following formula:
$$
\hat u(w)=\langle u,\exp(-z\cdot w)\rangle =\int_{\mathbb C} u(z)\exp(-z\cdot w) d\lambda (z)
$$
where $d\lambda (z)$ is the Lebesgue measure on $\mathbb C$.
Let us now write $w$, the dual variable of $z$, in polar coordinates as $w=|w| \exp(i\varphi)$,
$\varphi\in [0,2\pi)$ and define the function
$$
\beta(\varphi,\alpha)=\left\{\begin{array}{ccc}
& 0 & 0\leq \varphi\leq \pi/2 -\alpha\\
& \cos(\pi-\alpha-\varphi) & \pi/2 -\alpha < \varphi\leq \pi -\alpha\\
& 1 & \pi -\alpha < \varphi\leq \pi +\alpha\\
& \cos(\pi+\alpha-\varphi) & \pi +\alpha < \varphi\leq 3\pi/2 +\alpha\\
& 0 & 3\pi/2 +\alpha < \varphi\leq 2\pi.\\
\end{array}\right.
$$
We have the following result (see \cite{bs} for its proof):
\begin{theorem}
The space $\mathcal{A}_{p,0}(\Gamma)'$
is isomorphic, via the Laplace transform, to the
space, denoted by $\widehat{(\mathcal{A}_{p,0}(\Gamma)')}$, of entire functions $f\in\mathcal H(\mathbb C)$ such that, for some $B>0$ and some $\alpha\in (0, \theta)$
satisfy the inequality
$$
| f(w) | \leq C \exp(B |w|^\sigma \beta(\varphi, \alpha)^\sigma-\frac 1B {\rm Re}(w)).
$$
\end{theorem}

Note that the theory of Cauchy problems on holomorphic functions shows that their formal solutions may sometimes converge only in cones, see e.g. \cite{lms}, and as we pointed out this is certainly the case for Dirichlet series, whose domains of convergence (and of absolute convergence) are always half-planes. As it is well known, Dirichlet series are special cases of more general series of the form

\[
f(z)=\sum_{j=1}^{\infty} c_je^{i \lambda_j z}, \ \ c_j,\lambda_j \in \mathbb{C},
\]
\noindent
which now converges on actual cones, whose amplitudes depend on the frequencies $\lambda_j$, while the convergence itself depends on the growth of the coefficients $c_j$. While Ehrenpreis hinted at the possibility of studying these series in the framework of his theory (see \cite{ehrenpreis}), it was only in \cite{bs}, \cite{bs2}, \cite{bs3} that the proper framework for such an approach was identified. Since every exponential in these series can be seen as the limit of a superoscillating sequence, one may want to study Dirichlet series themselves as limits of more complex superoscillating sequences. This is what we will do, thus providing a new class of interesting superoscillating sequences.

Thus consider a generalized Dirichlet series of the form
$$
f(z)=\sum_{j=0}^{+\infty} c_j e^{i\lambda_j z},
$$
which we assumed to be absolutely convergent in a half plane of the form ${\rm Im}(z)> \gamma$ for some negative constant $\gamma$ (which can possibly be $-\infty$).
Now we can consider the restriction of such Dirichlet series to the real axis, and replace the exponential $e^{i\lambda_j x}$ with the limit of the superoscillating sequence
$$
e^{i\lambda_j x}=\lim_{n\to+\infty}\sum_{k=0}^{n} c(n,\lambda_j) e^{i (1-2k/n)x}.
$$
As we have shown in Chapter 3, see Theorem \ref{carnot}, this limit is uniform on the compact set $[-M,M]$ and
the error, see Remark \ref{estimate}, can be estimated by
$$
|e^{i\lambda_j x} - \sum_{k=0}^{n} c(n,\lambda_j) e^{i (1-2k/n)x}|\leq 2\frac Mn \lambda_j ,
$$
and therefore
$$
|\sum_{j=0}^{+\infty} c_j e^{i\lambda_j x}-\sum_{j=0}^{+\infty} c_j \sum_{k=0}^{n} c(n,\lambda_j) e^{i (1-2k/n)x}|
$$
$$
\leq
\sum_{j=0}^{+\infty} |c_j|\, \left|e^{i\lambda_j x} - \sum_{k=0}^{n} c(n,\lambda_j) e^{i (1-2k/n)x}\right|\leq \sum_{j=0}^{+\infty} 2|c_j| \frac Mn \lambda_j,
$$
which converges to $0$ as long as the series $\sum c_j \lambda_j$ converges absolutely.
We have therefore proved the following theorem:
\begin{theorem}
Let $c_j$, $\lambda_j$ be two sequences of complex numbers such that $\sum_{j=0}^{+\infty} c_j e^{i\lambda_j x}$ is convergent on $\mathbb R$ and $\sum c_j\lambda_j$ is absolutely convergent. Then the sequence
$\{\sum_{j=0}^m c_j\sum_{k=0}^n c(n,\lambda_j) e^{i (1-2k/n)x}\}$, $m\in\mathbb N$, is a superoscillating sequence whose limit is
$$
\sum_{j=0}^{+\infty} c_j e^{i\lambda_j x}.
$$
\end{theorem}

\chapter{Schr\"odinger equation and superoscillations}

A natural question, that arises for physical reasons, is to study the evolution of a superoscillatory sequence when we take such a sequence as initial value of a Cauchy problem of some differential equations of physical interest. In some cases the answer to this question is immediate. This is the case if we investigate the behavior of a superoscillating initial datum in the case of the wave equation
$$
u_{tt}(t,x)-c^2u_{xx}(t,x)=0,
$$
where $c\in \mathbb R^+$,
with the initial position $u(0,x)=Y_n(x)$, where $Y_n$ is given by Definition \ref{ipsilon enne} and  the initial speed $u_t(0,x)=0$.
The solution can be written using the d'Alembert formula as
$$
u_n(t,x)=\frac{1}{2}[Y_n(x-ct)+Y_n(x+ct)].
$$
If we replace the initial condition by $u(0,x)=e^{ig(a)x}$ and we keep $u_t(0,x)=0$ we obtain
$$
u(t,x)=\frac{1}{2}[e^{ig(a)(x-ct)}+e^{ig(a)(x+ct)}].
$$
We now consider the difference $u_n(t,x)-u(t,x)$ and the estimate
$$
|u_n(t,x)-u(t,x)|\leq \frac{1}{2}|Y_n(x-ct)-e^{ig(a)(x-ct)}|+\frac{1}{2}|Y_n(x+ct)-e^{ig(a)(x+ct)}|.
$$
So, for $(t,x)$ on every compact set $[0,T]\times K$ for $K$ compact set in $\mathbb R$, we have uniform convergence.
One can therefore say that superoscillations sequences maintain their superoscillating character when evolved with the wave equation, when $n\to\infty$.
On the other hand, it is important to note that the persistence of superoscillations only occurs when one takes the limit for $n\to +\infty$. If one fixes the value of $n$, persistence is only for a finite time and superoscillations are, for large $n$, exponentially weak, see \cite{b4}.

In this chapter, we will consider a more delicate case, namely  the evolution of a superoscillatory sequence for the Schr\"odinger equation
when it is taken as initial value for the free particle as well as for the quantum harmonic oscillator. In the case of the Schr\"odinger equation, many continuity results are known when the data are in $\mathcal L^2(\mathbb R)$, but since our functions do not belong to $\mathcal L^2(\mathbb R)$, we need to follow a different approach which exploits the analyticity of the initial data. In particular, we will look both at the case of the Schr\"odinger equation for the free particle (in which case we prove the longevity of the superoscillatory phenomenon in two different ways) as well as the case of the quantum harmonic  oscillator.

Note that in this chapter, to simplify the notation, we sometimes write $F_n(x)$ instead of $F_n(x,a)$ and  $Y_n(x)$ instead of $Y_n(x,a)$.

\section{Schr\"odinger equation for the free particle}
\label{sec5.1}

In this section we consider the
Cauchy problem
$$
i\frac{\partial \psi(x,t)}{\partial t}=H\psi(x,t),\ \ \ \ \psi(x,0)=Y_n(x),
$$
where
$$
H\psi(x,t):=-\frac{\partial^2 \psi(x,t)}{\partial x^2}.
$$
First, we prove the following result:
\begin{theorem}\label{eifhgsldi} The time evolution of the spatial superoscillating sequence $Y_n(x)$, is given by
$$
\psi_n(x,t)=\sum_{j=0}^nC_j(n,a) e^{ik_j(n) x } e^{-itk_j^2(n)}.
$$
\end{theorem}
\begin{proof}
To solve the Schr\"odinger equation with $Y_n(x)$ as initial condition, we will work in the space of the tempered distributions $\mathcal{S}'(\mathbb{R})$ and use a standard Fourier transform argument.
Let us denote by $\hat{\psi}(\lambda,t)$ the Fourier transform of $\phi$. Taking the Fourier transform of the Schr\"odinger equation we obtain
$$
i\frac{d \hat{\psi}(\lambda,t)}{d t}=\lambda^2\hat{\psi}(\lambda,t)
$$
and integrating we get
$$
\hat{\psi}(\lambda,t)=C(\lambda)e^{-i\lambda^2t}
$$
where the arbitrary function $C(\lambda)$ can be determined by the initial condition and therefore
\[
\begin{split}
C(\lambda)&=\hat{\psi}(\lambda,0)
=\int_{\mathbb{R}}\Big[\sum_{j=0}^nC_j(n,a)e^{ik_j(n){x}}\Big] e^{-i\lambda x}dx\\
&=\sum_{j=0}^nC_j(n,a)\int_{\mathbb{R}}e^{ik_j(n){x}} e^{-i\lambda x}dx.
\end{split}
\]
Here we use the equality $
\mathcal{F}(e^{imx})=2\pi\delta (x-m),
$
which has to be interpreted in $\mathcal{S}'(\mathbb{R})$.
It follows that
$$
C(\lambda)=\sum_{j=0}^nC_j(n,a) \delta (\lambda -k_j(n))
$$
and
$$
\hat{\psi}(\lambda,t)=\sum_{j=0}^nC_j(n,a) \delta (\lambda -k_j(n))e^{-i \lambda^2t}.
$$
Taking now the inverse Fourier transform we have
\[
\begin{split}
\psi(x,t)&=\int_{\mathbb{R}}
\Big[\sum_{j=0}^nC_j(n,a) \delta (\lambda -k_j(n))e^{-i\lambda^2t}\Big] e^{i\lambda x} d\lambda\\
&=\sum_{j=0}^nC_j(n,a) \int_{\mathbb{R}}
\Big[\delta (\lambda -k_j(n))e^{-i\lambda^2t}\Big] e^{i\lambda x} d\lambda
\\
&=\sum_{j=0}^nC_j(n,a) \int_{\mathbb{R}}
\delta (\lambda -k_j(n))e^{i\lambda x} e^{-i\lambda^2t}d\lambda
\\
&=\sum_{j=0}^nC_j(n,a) e^{ik_j(n){x} } e^{-it{k_j^2(n)}},
\end{split}
\]
and the statement follows.
\end{proof}
Our next goal is to show that the function $\psi_n(x,t)$ exhibits, for all values of $t$, the same superoscillatory behavior shown by $Y_n(x)$.
As a first step, we give an equivalent representation of the time evolution $\psi_n$ in terms of the derivatives of the functions $Y_n$.
\begin{theorem}
The function
\begin{equation}\label{psin}
\psi_n(x,t)=\sum_{j=0}^nC_j(n,a)e^{ixk_j(n)}e^{-itk_j^2(n)}
\end{equation}
can be written as
$$
\psi_n(x,t)=\sum_{m=0}^\infty \displaystyle\frac{(it)^m}{m!} \displaystyle\frac{d^{2m}}{dx^{2m}} Y_n(x)
$$
 for every $x\in \mathbb{R}$ and $t\in \mathbb{R}$.
\end{theorem}

\begin{proof}
Using the expansion
$$
e^{-itk_j^2(n)}=\sum_{m=0}^\infty \displaystyle\frac{[-itk_j^2(n)]^m}{m!}
$$
the function $\psi_n(x,t)$ can be rewritten as
\[
\begin{split}
\psi_n(x,t)&=\sum_{m=0}^\infty \displaystyle\frac{(-it)^m}{m!}  \sum_{j=0}^nC_j(n,a)e^{ixk_j(n)}e^{-itk_j^2(n)}\\
&=\sum_{m=0}^\infty \displaystyle\frac{(it)^m}{m!}  \sum_{j=0}^nC_j(n,a)\displaystyle\frac{d^{2m}}{dx^{2m}}e^{ix k_j(n)}\\
&=\sum_{m=0}^\infty \displaystyle\frac{(it)^m}{m!} \displaystyle\frac{d^{2m}}{dx^{2m}} \sum_{j=0}^nC_j(n,a)e^{ixk_j(n)}\\
&=\sum_{m=0}^\infty \displaystyle\frac{(it)^m}{m!} \displaystyle\frac{d^{2m}}{dx^{2m}} Y_n(x).
\end{split}
\]
\end{proof}

We are now led to study the operator formally defined by
$$
U\left(\frac{d}{dx},t\right):=\sum_{m=0}^\infty \displaystyle\frac{(it)^m}{m!} \displaystyle\frac{d^{2m}}{dx^{2m}}
$$
and the spaces of functions on which it acts continuously.
 To this purpose, we extend $U(\frac{d}{dx},t)$ to an operator which may act on holomorphic functions, i.e. we consider operators of the form
$$
U\left(\frac{d}{dx},t\right)=\sum_{m=0}^\infty \frac{(it)^m}{m!} \displaystyle\frac{d^{m}}{dz^{m}}
$$
and we denote by $h(\zeta,t)$ the function which is its symbol. Then
$$
h(\zeta,t)=\sum_{m=0}^\infty \frac{a_m}{m!} \zeta^m
$$
where $a_m=(it)^m$. It is immediate to show that, unless $t=0$, this symbol does not define an infinite order differential operator (in the sense of Definition \ref{infiniteorderdef}). However,  $h(\zeta,t)$ can be thought
as the symbol of a convolution operator for suitable choices of the coefficients $a_m$.
For instance, if $a_m\equiv 1$, i.e. $t=-i$, then $h(\zeta,-i)=e^\zeta$ and therefore
is the symbol of the translation of the unit operator which is nothing but the convolution with the Dirac delta centered at $z =-1$, see (\ref{translation}). Moreover, the function $e^\zeta$ is clearly a multiplier on ${\rm Exp}(\mathbb{C})$ and the convolutor that it defines is the translation as indicated above. If we now consider  the symbol
$$
h(\zeta,t)=\sum_{m=0}^\infty \frac{(it)^m}{m!}\zeta^m
$$
 it is easy to see that such operator defines, for any value of $t$, a convolutor on ${\mathcal H}(\mathbb{C})$, in fact the translation by $it$. This is easily seen because $h(\zeta,t)$ is actually nothing but $\exp(it\zeta)$.
The operator we are now interested in, however, is
$$
\sum_{m=0}^\infty \displaystyle\frac{(it)^m}{m!} \displaystyle\frac{d^{2m}}{dz^{2m}}
$$
whose symbol is $h(\zeta^2,t)$.
It is obvious to see that $h(\zeta^2,t)$ does not define a multiplication operator on ${\rm Exp}(\mathbb{C})$ because it grows at infinity too fast. The appropriate space for which $h$ would be a multiplier and therefore the appropriate space for which $h$ would induce a convolution operator is given by the following result:
\begin{theorem}
For any value of $t$, the operator $U(\frac{d}{dz},t)$ acts continuously on the space
$$
\mathcal{A}_{2,0}:=\left\{f\in \mathcal{O}(\mathbb{C})\ : \ \forall \varepsilon >0 \  \exists \, A_\varepsilon\ \ : \  |f(z)|\leq A_\varepsilon e^{\varepsilon |z|^2}    \right\}
$$
of entire functions of order less or equal $2$ and of minimal type.
\end{theorem}
\begin{proof}
It is immediate to verify that the function $h(\zeta^2,t)$ is a continuous multiplier on $\mathcal{A}_2$, regardless of the magnitude
of $t.$
Theorem \ref{duality pq} shows that any multiplier on $\mathcal{A}_2$ defines a convolutor on the space
$$
\mathcal{A}_{2,0}:=\left\{f\in \mathcal{O}(\mathbb{C})\ : \ \forall \ \ \varepsilon >0 \  \exists \ \ a_\varepsilon\ \ : \  |f(z)|\leq a_\varepsilon e^{\varepsilon |z|^2}    \right\}.
$$
Therefore for every function $f\in \mathcal{A}_{2,0}$, and every $t \in \mathbb{R}$, the function
$$
U\left(\frac{d}{dz},t\right)f=\sum_{m=0}^\infty \displaystyle\frac{(it)^m}{m!} \displaystyle\frac{d^{2m}f}{dz^{2m}}
$$
is a well defined function in $\mathcal{A}_{2,0}$ and the operator $U(\frac{d}{dz},t)$ acts continuously on $\mathcal{A}_{2,0}$.
\end{proof}

As a consequence of the previous discussion we can show that the superoscillatory phenomenon persists for $n\to\infty$ for all values of the time $t$:
\begin{theorem}\label{main}
For $a>1$, and for every $x,t \in\mathbb{R}$ we have
$$
\lim_{n\to\infty}
\psi_n(x,t)=e^{ig(a)x-ig(a)^2t}.
$$
\end{theorem}
\begin{proof}
 The functions $Y_n$ extend to entire functions of order less than or equal 1 and finite type (i.e. of exponential type), and this space is clearly contained in $\mathcal{A}_{2,0}$.
Therefore, in order to show that
$$
\lim_{n\to\infty}
\psi_n(x,t)=e^{ig(a)x-ig(a)^2t},
$$
it is enough to take the limit and recall that
$$
Y_n(x)\to e^{ig(a)x},
$$
so we obtain
$$
\psi(x,t)=\sum_{m=0}^\infty \displaystyle\frac{(it)^m}{m!} \displaystyle\frac{d^{2m}}{dx^{2m}}e^{ig(a)x}
$$
$$
=\sum_{m=0}^\infty \displaystyle\frac{(it)^m}{m!} (ig(a))^{2m}e^{ig(a)x}
$$
$$
=\sum_{m=0}^\infty \displaystyle\frac{(-i g(a)^2t)^m}{m!} e^{ig(a)x}
$$
$$
=e^{ig(a)x-ig(a)^2t}.
$$
This finishes the proof.
\end{proof}

In the case in which $Y_n$ is taken to be $F_n$ the results above reduce to the main theorem in \cite{ACSST4}:
\begin{theorem}
If
$$
\psi_n(x,t)=\sum_{j=0}^nC_j(n,a)e^{ix(1-2j/n)}e^{-it(1-2j/n)^2}.
$$
then
for $a>1$, and for every $x,t \in\mathbb{R}$ we have
$$
\lim_{n\to\infty}
\psi_n(x,t)=e^{iax-ia^2t}.
$$
Moreover
for $n$ large and small $t$ the  approximated error of
$$
|\psi_n(x,t)-e^{i(ax-a^2t)}|
$$
is
$$
\varepsilon(x,t,a,n)=\varepsilon_1(x,t,a,n)+\varepsilon_2(x,t,a,n)
$$
where
$$
\varepsilon_1(x,t,a,n)=\frac{|x-at|}{n}\sqrt{\frac{3}{2}(a^2-1)},
$$
and
 $$
\varepsilon_2(x,t,a,n)=|t|   (a^3+a).
$$
\end{theorem}
\begin{proof}
The first part of the theorem is an immediate application of Theorem \ref{main}. In order to show the second part,
let us observe that
$$
|F_n(x-at)-e^{iax-ia^2t}|\to 0
$$
uniformly on the compact sets of $\mathbb{R}^2$.
We now write the immediate inequality
$$
|\psi_n(x,t)-e^{iax-ia^2t}|\leq |\psi_n(x,t)-F_n(x-at)|+|F_n(x-at)-e^{iax-ia^2t}|,
$$
and we observe that
 $|F_n(x-at)-e^{iax-ia^2t}|$ has been estimated in Remark \ref{estimate}
and it is
$$
\varepsilon_1(x-at,a,n)=\frac{|x-at|}{n}\sqrt{\frac{3}{2}(a^2-1)}.
$$
Thus we are left with the estimate of $|\psi_n(x,t)-F_n(x-at)|$. Consider
\begin{equation}
\begin{split}
&\psi_n(x,t)-F_n(x-at)\\
&=\sum_{k=0}^nC_k(n,a)e^{ix(1-2k/n)}e^{-it(1-2k/n)^2}-\sum_{k=0}^nC_k(n,a)e^{i(x-at)(1-2k/n)}.
\end{split}
\end{equation}
This expansion can be rewritten as
$$
\psi_n(x,t)-F_n(x-at)=\sum_{k=0}^nC_k(n,a)e^{ix(1-2k/n)}[e^{-it(1-2k/n)^2}-e^{-iat(1-2k/n)}].
$$
When $t\to 0$, we can approximate at the first order
$$
e^{-it(1-2k/n)^2}-e^{-iat(1-2k/n)}=-it(1-2k/n)^2+iat(1-2k/n)+o(t)
$$
and, as a consequence,
$$
\psi_n(x,t)-F_n(x-at)\sim t \ \sum_{k=0}^nC_k(n,a)e^{ix(1-2k/n)}[-i(1-2k/n)^2+ia(1-2k/n)]
$$
$$
=t \  [aF_n''(x)-F_n'(x)].
$$
Thus we have
$$
|\psi_n(x,t)-F_n(x-at)|\sim |t|\,   |aF_n''(x)-F_n'(x)|,
$$
and
$$
\varepsilon_2(x,t,a,n):=|t|\,   |aF_n''(x)-F_n'(x)|.
$$
Summarizing, the error is
$$
\varepsilon(x,t,a,n)=\varepsilon_1(x,t,a,n)+\varepsilon_2(x,t,a,n)
$$
$$
=\frac{|x-at|}{n}\sqrt{\frac{3}{2}(a^2-1)}+|t| \,  |aF_n''(x)-F_n'(x)|,
$$
and  the statement follows using Proposition \ref{sjhdjhgfns}.
\end{proof}
We observe that the difference $$\psi_n(x,t)-e^{i(ax-a^2t)}$$  can be written as the sum of two contributions:
$$
\psi_n(x,t)-e^{i(ax-a^2t)}=Z_n(x,t)+W_n(x,t),
$$
where
$$
Z_n(x,t):=\psi_n(x,t)-F_n(x-at)
$$
$$
W_n(x,t):=F_n(x-at)-e^{i(ax-a^2t)}.
$$
Theorem \ref{carnot} yields
$$
|W_n(x,t)|^2=
1+\Big(\cos^2 \Big(\frac{x-at}{n}\Big)+a^2\sin^2 \Big(\frac{x-at}{n}\Big)\Big)^{n}
$$
$$
-2\Big(\cos^2 \Big(\frac{x-at}{n}\Big)+a^2\sin^2 \Big(\frac{x-at}{n}\Big)\Big)^{n/2}\times
 $$
 $$
 \times
 \cos\Big[n\arctan \Big(a \tan \Big(\frac{x-at}{n}\Big)\Big) -a (x-at)\Big].
$$
So
$$
 |W_n(x,t)|^2\to 0, \ \ \ \ {\rm for \ \ all\ } x\in [-K,K], \ \ {\rm and}\  t\in [0,T] \ \ \ {\rm uniformly}.
 $$
While we are unable to find an appropriate estimate for $Z_n(x,t),$ it is still possible to provide a different representation for this term that might be useful in the future and that is of independent interest.
\begin{proposition}
 Let $F_n$ and $\psi_n$ be the function defined above.
Then $Z_n(x,t)$ can be written as
$$
Z_n(x,t)=\sum_{k=0}^nC_k(n,a)e^{ix(1-2k/n)}e^{i \Theta_{k,n}(t)} \sin  [(t/2) (1-2k/n-a)]
$$
where
\begin{equation}\label{frugsc}
\tan \Theta_{k,n}(t):=- \ \displaystyle\frac{\sin [t(1-2k/n)^2]- \sin [t a(1-2k/n)]}{\cos [t(1-2k/n)^2]- \cos [ta(1-2k/n)]}.
\end{equation}
\end{proposition}
\begin{proof}
The result follows from the equality
\[
\begin{split}
Z_n(x,t)&=\psi_n(x,t)-F_n(x-at)
\\
&
=\sum_{k=0}^nC_k(n,a)e^{ix(1-2k/n)}\Big[ e^{it(1-2k/n)^2} - e^{ita(1-2k/n)}\Big].
\end{split}
\]
Observing that
$$
e^{it(1-2k/n)^2} - e^{ita(1-2k/n)}=\rho_{k,n}(t) e^{i \Theta_{k,n}(t)}
$$
where $\Theta_{k,n}(t)$ is given by (\ref{frugsc}) and
$$
\rho^2_{k,n}(t)=2-2\cos [(t(1-2k/n)^2-at(1-2k/n)],
$$
with some  computations the statement follows.
\end{proof}

\section{Approximation by gaussians and persistence of superoscillations}
\label{sec5.2}

In this section we prove that the evolution of a superoscillating sequence taken as initial datum of the Schr\"odinger equation remains superoscillatory following a different approach which may be considered more direct and which allows a more intuitive physical
 interpretation. To prove the result we first  represent the superoscillating sequence $Y_n(x)$ (we will omit the dependence on $a$ unless necessary) as a suitable double integral of a
 kernel multiplied by a modified gaussian by means of the windowed Fourier transform. Then we study an auxiliary problem, namely the time evolution of the modified gaussians when taken as initial value for the Schr\"odinger equation; we let
$Y_n(x)$ evolve by taking its integral representation and evolving the modified gaussians inside it. Finally, one obtains that the evolution of the modified gaussians for any fixed $t$  and large enough $n$, results again in a gaussian with asymptotically the same width and with its center translated in such a way that its tails do not interfere with the original gaussian.
Intuitively, one can explain the longevity of the superoscillations by noting that the size of the tails grows at a rate proportional to $a^n$ while the exponential tails oscillate at a rate proportional to $1/a$. The exponential tail to the left of the superoscillating region is both too slow and too far away to overcome and destroy the superoscillating region.  The exponential tails to the right of the superoscillating region are more of a threat to the longevity of the superoscillating region because the faster moving superoscillating region could ``catch-up" with the slower-moving exponential tail.
Thus the superoscillating region can survive as long as the time $t<\sqrt n$. We now follow \cite{acsstYa80} to make this argument rigorous.

 We begin by writing each function in the superoscillating sequence as the inverse of its windowed Fourier transform, (see Gabor  \cite{gabor}). Using the Gabor chirp centered at $u$:
\begin{equation}\label{chirp}
g_{u,\xi}(t):= e^{i\xi t}g(t-u)
\end{equation}
where
$$
g(t)=\frac{1}{\sqrt \pi}e^{-t^2},
$$
 one can define the so-called windowed Fourier transform:
\begin{definition}
 Let $f$ be in $\mathcal{L}^2(\mathbb R)$,
then its windowed Fourier transform (also known as the short time Fourier transform) is defined as
$$
\mathcal{G}(f)(u,\xi)=\int_{-\infty}^{\infty} f(t)g(t-u) \exp(-i\xi t)\,  dt .
$$
\end{definition}
We state below the following theorem which is well known (see \cite{mallat}):
\begin{theorem}\label{windowFourier}
If $f$ is a function in $\mathcal{L}^2(\mathbb R)$ then $\mathcal{G}f\in \mathcal{L}^2(\mathbb R)$, $\| \mathcal{G}f\|=\| f\|$ and
\[
f(t)= \frac{1}{2\pi}\int_{-\infty}^{\infty}\int_{-\infty}^{\infty} \mathcal{G} f(u,\xi) g(t-u) \exp(i\xi t) d\xi du.
\]
\end{theorem}
In particular, we have:
\begin{corollary}
 Let $K$ be compact in $\mathbb R$. Then
there are functions $g_n(x_0,k_0)$ such that
\[
Y_n(x,a)= \int_K \int_{-\infty}^{\infty} g_n(x_0, k_0,a) \exp\left(\frac{-(x-x_0)^2}{2\Delta(0)^2}\right) \exp(ik_0x) dx_0\ dk_0.
\]
\end{corollary}
\begin{proof}
After a trivial change of variable, this is an immediate consequence of Theorem \ref{windowFourier}.
\end{proof}

We now show how to solve the Cauchy problem for the Schr\"odinger equation when the
initial datum is the Gabor chirp \eqref{chirp} centered at $x_0$ and with initially a single wave number $k_0$, and we use this result to conclude the longevity of superoscillations. Consider
\begin{equation}\label{Schr}
i\frac{\partial \phi(x,t)}{\partial t}=-\frac{\partial^2 \phi(x,t)}{\partial x^2},\ \ \ \ \phi(x,0)=\exp \left(-\frac{(x-x_0)^2}{2\Delta_0^2}+ik_0x\right),
\end{equation}
where $\Delta_0$ is the initial spread of the gaussian.
\\
Taking the Fourier transform of the Schr\"odinger equation we get
$$
i\frac{d \hat{\phi}(p,t)}{d t}=p^2\hat{\phi}(p,t)
$$
and integrating we obtain
$$
\hat{\phi}(p,t)=C(p)e^{-ip^2t}
$$
where the arbitrary function $C(p)$ will be determined by the initial condition. Recalling
the well known property
$$
\int_{\mathbb{R}} \exp( -\frac{x^2}{2\Delta_0^2}) \ \exp(- ip  x)\ dx=\sqrt{2\pi \Delta_0^2}  \  \exp(-\Delta_0^2 p^2)
$$
and the fact that $\mathcal{F}[g(x-x_0)]=\mathcal{F}[g(x)]e^{-ix_0 p}$,
we obtain
\[
\begin{split}
C(p)&=\hat{\phi}(p,0)\\
&=  \int_{\mathbb{R}}\exp\left( -\frac{(x-x_0)^2}{2\Delta_0^2}+ik_0x\right)\ \exp ({-ipx})dx\\
&=\sqrt{2\pi \Delta_0^2}  \  \exp(-ipx_0-\Delta_0^2(p-k_0)^2).\\
\end{split}
\]
Taking now the inverse Fourier transform  we have
$$
\hat{\phi}(p,t)=\frac{1}{2\pi}\sqrt{2\pi \Delta_0^2} \int_{\mathbb{R}}   \  \exp(-ipx_0-\Delta_0^2(p-k_0)^2)   \exp( ip x) dp.
$$
We finally obtain (see also formula (5.4) in \cite{aharonov_book} in which we have set $\hbar=m=1$):
\[
\begin{split}
\phi(x,t,x_0,k_0)&=\frac{1}{(1+2it)^{1/2}}
\\
&
\times
\exp\left(-i\frac{k_0^2}{2}t\right)\exp(ik_0x)\exp\left(-\frac{(x-x_0-2k_0t)^2}{2(\Delta_0^2+2it)}\right),
\end{split}
\]
which can be written as
$$
\phi(x,t,x_0,k_0)=\frac{1}{(1+2it)^{1/2}}\exp\left[ i\left(\frac{k_0^2}{2}t+k_0x+t\frac{(x-x_0-2k_0t)^2}{\Delta_0^4+4t^2}\right)\right ]
$$
$$
\times
\exp\left(\!\!-\frac{(x-x_0-2k_0t)^2}{2(\Delta_0^2+4\frac{t^2}{\Delta_0^2})}\right).
$$
Thus the evolution according to the Schr\"odinger equation of the functions in the superoscillatory sequence is  given  by
\[
 Y_n(x,t)=\int\int g_n(x_0,k_0) \phi(x,t,x_0,k_0) dx_0\, dk_0.
\]
We are now ready to prove our result, namely that the functions $Y_n(x,t)$ preserve the superoscillatory behavior of $Y_n(x)$, and therefore that superoscillations persist for large values of $t$, when evolved according to the Schr\"odinger equation.

\begin{theorem}
 Let $Y_n(x)$ be a superoscillatory sequence. Then, for any fixed time $t$, its evolution $Y_n(x,t)$ obtained
by solving the Cauchy problem for the Schr\"o\-din\-ger equation with initial datum $Y_n(x)$ is still a superoscillatory sequence on any arbitrary large set in $\mathbb R$.
\end{theorem}
\begin{proof}
For any time $t$ we can choose an $n$ such that $t\approx n^{\frac 12 -\varepsilon}$.
Now we know that the time evolution of $\phi(x)$ is, up to the factor $(1+2it)^{-1}$, the product of an oscillatory function (with an amplitude of $1$)
$$
\exp\left[ i\left(\frac{k_0^2}{2}t+k_0x+t\frac{(x-x_0-2k_0t)^2}{\Delta_0^4+4t^2}\right)\right ]
$$
and of a translated gaussian
$$
\exp\left(-\frac{(x-x_0-2k_0t)^2}{2\left(\Delta_0^2+4\frac{t^2}{\Delta_0^2}\right)}\right).
$$
The spreads of these new gaussians are given by $\Delta^2(t)=\Delta_0^2 + t^2/\Delta_0^2$. The assumption on $t$ and the choice  $\Delta_0\approx n^{\frac{1}{2}}$ show that the spreads are approximatively
the same at any given moment. Now observe that this wave packet has width which is approximately $\sqrt n$, therefore if we consider it centered in the point $x_0+2k_0t$, and so in the interval $[\lambda \sqrt n, (\lambda+1) \sqrt n]$ for some $\lambda$, we see that its contribution outside its spread does not interfere with the original superoscillatory region.
\end{proof}

\begin{remark} {\rm We have demonstrated the permanence of superoscillations in two different ways. On one hand, the use of the theory of Analytically Uniform spaces allowed us to directly show that superoscillations remain when $n$ goes to infinity. On the other hand, by the use of the windowed Fourier transform, we have shown that if we fix $n$, the superoscillatory phenomenon lasts until $t<\sqrt{n}.$ This last result obviously implies the first, but it would be interesting to fully understand the relationship between the two statements. }
\end{remark}

\section{Quantum harmonic oscillator}
\label{sec5.3}

The mathematical strategy that has been used in Section \ref{sec5.1} to study the evolution of superoscillations works well in most cases when
we have explicit solutions for Green's functions and propagators for the physical system under consideration.
We will now show how to use this strategy in the case of the quantum harmonic oscillator described by the Hamiltonian
\[ \begin{split}
  H(t,x)=-\frac{\hbar^2}{2m}\frac{\partial^2}{\partial x^2}+\frac{1}{2}m\omega^2(t) x^2-f(t)x.
  \label{}
\end{split}  \]
From a physical point of view, this Hamiltonian represents a harmonic oscillator of mass $m$  and time-dependent frequency $\omega(t)$ under the influence of the external time-dependent force $f(t)$.
For the sake of simplicity, in the sequel we will rescale the variables in order to solve
 the Cauchy problem for the quantum harmonic oscillator
$$
i\frac{\partial \psi(t,x)}{\partial t}=\frac{1}{2}\Big( -\frac{\partial^2 }{\partial x^2} +x^2 \Big)\psi(t,x),\ \ \ \ \psi(0,x)=F_n(x,a),
$$
and we will show that
its solution is
\[
\begin{split}
\psi_n(t,x)&=(\cos t)^{-1/2}\exp\Big( -(i/2) x^2\tan t\Big)
\\
&
\times \sum_{k=0}^nC_k(n,a) \exp( ix(1-2k/n)/\cos t -(i/2)(1-2k/n)^2\tan t).
\end{split}
\]
Taking the limit for $n\to \infty$ one obtains
$$
\lim_{n\to \infty}\psi_n(t,x)=(\cos t)^{-1/2}\exp( -(i/2) (x^2+a^2)\tan t +iax/\cos t),
$$
and so the superoscillations are amplified by the potential and the analytic solution blows up for $t=\pi/2$.
Moreover, even when $a\in (0,1)$, the harmonic oscillator displays a superoscillatory phenomenon
since the solution contains the term
$ \exp( -(i/2) (x^2+a^2)\tan t +iax/\cos t)$, which increases arbitrarily as $t$ approaches $\pi/2.$
This is a new feature which does not occur for the free particle.

We start by writing the solution of the Cauchy problem for the quantum harmonic oscillator using its Green's function $G(t,x,0,x')$.
The Green's function is such that
 $G(t,x,0,x')=\theta(t)\tilde{G}(t,x,0,x')$, where $\theta(t)$ is the Heaviside function and $\tilde{G}(t,x,0,x')$ is the propagator, see \cite{feynman},  \cite{gy}, \cite{montroll}, \cite{schulman}.
So the solution of the Cauchy problem
\begin{equation}\label{Cauchy}
i\frac{\partial \psi(t,x)}{\partial t}=\frac{1}{2}\Big( -\frac{\partial^2 }{\partial x^2} +x^2 \Big)\psi(t,x),\ \ \ \ \psi(0,x)=\psi_0(x)
\end{equation}
is
\begin{equation}
 \psi(t,x)=\int_{\mathbb{R}} {G}(t,x,0,x')\psi_0(x')dx',
\end{equation}
where
$$
G(t,x,0,x'):=(2\pi i\sin t)^{-1/2}  e^{(2xx'-(x^2+x'^2)\cos t)/(2 i \sin t)}.
$$

We first prove the following useful result.
\begin{proposition}\label{PropPsiA}
Let $a\in \mathbb{R}$. Then the solution of the Cauchy problem
\begin{equation}\label{CauchyHarm}
i\frac{\partial \psi(t,x)}{\partial t}=\frac{1}{2}\Big( -\frac{\partial^2 }{\partial x^2} +x^2 \Big)\psi(t,x),\ \ \ \ \psi(0,x)=e^{iax}
\end{equation}
is
\begin{equation}\label{psiNoscA}
\psi_a(t,x)=(\cos t)^{-1/2} \exp( -(i/2) (x^2+a^2)\tan t +iax/\cos t).
\end{equation}
\end{proposition}
\begin{proof}
Using the Green function we get
\[
\begin{split}
\psi_a(t,x)&=(2\pi i\sin t)^{-1/2} \int_{\mathbb{R}} \exp ((2xx'-(x^2+x'^2)\cos t)/(2 i \sin t))\ e^{iax'} \ dx'
\\
&=(2\pi i\sin t)^{-1/2} \ \exp (-x^2\cos t/(2 i \sin t))
\\
&
\times
\int_{\mathbb{R}} \exp ((2xx'-x'^2\cos t+iax'(2i\sin t))/(2 i \sin t)) \ dx',
\end{split}
\]
and with some computations we
obtain
\[
\begin{split}
\psi_a(t,x)&=(2\pi i\sin t)^{-1/2} \
\\
&
\times \exp (-x^2\cos t/(2 i \sin t))\ \exp (-i(x-a\sin t)^2/(2\sin t\cos t))
\\
&
\times
\int_{\mathbb{R}} \exp (i(x'\cos t-(x-a\sin t))^2/(2  \sin t \cos t)) \ dx'.
\end{split}
\]
We now perform a change of variable and we use the regularized integral
$$
\int_{\mathbb{R}}e^{i\alpha x^2}\, dx=\lim_{\beta\to 0^+}\int_{\mathbb{R}}e^{- x^2(\beta-i\alpha)}\, dx=\Big(\frac{i\pi }{\alpha}\Big)^{1/2},
$$
which yields
$$
\int_{\mathbb{R}} \exp (i(x'\cos t-(x-a\sin t))^2/(2  \sin t \cos t)) \ dx'=\frac{1}{\cos t}(2\pi i \sin t\cos t)^{1/2},
$$
from which one finally obtains
\[
\begin{split}
\psi_a(t,x)&=(2\pi i\sin t)^{-1/2} \ \exp (-x^2\cos t/(2 i \sin t))
\\
&
\times \exp (-i(x-a\sin t)^2/(2\sin t\cos t)) \times \frac{1}{\cos t}(2\pi i \sin t\cos t)^{1/2}.
\end{split}
\]
Now the statement follows from some standard computations.
\end{proof}

Now we can prove the following:
\begin{theorem}\label{TH23} The solution of the  Cauchy problem
\begin{equation}\label{CauchyHarmggg}
i\frac{\partial \psi(t,x)}{\partial t}=\frac{1}{2}\Big( -\frac{\partial^2 }{\partial x^2} +x^2 \Big)\psi(t,x),\ \ \ \ \psi(0,x)=F_n(x,a)
\end{equation}
where $F_n(x,a)=\left(\cos\left(\frac xn\right)+ia\sin\left(\frac xn\right)\right)^n$ is
\begin{equation}\label{psiNoscfff}
\begin{split}
\psi_n(t,x)&=(\cos t)^{-1/2} \exp (-(i/2) x^2\tan t)
\\
&
\times
\sum_{k=0}^nC_k(n,a)\exp ( ix(1-2k/n)/\cos t -(i/2)(1-2k/n)^2\tan t).
\end{split}
\end{equation}
Moreover, if we set $\psi(t,x)=\lim_{n\to \infty}\psi_n(t,x)$, then
\begin{equation}\label{psiNoscinfty}
\psi(t,x)=(\cos t)^{-1/2} \exp (  -(i/2) (x^2+a^2)\tan t +iax/\cos t).
\end{equation}
\end{theorem}
\begin{proof} To prove that the solution of the Cauchy problem (\ref{CauchyHarmggg}) is given by (\ref{psiNoscfff}) we
observe that the initial datum $F_n(x,a)$ is a linear combination of the exponentials
$e^{ix(1-2k/n)}$. Formula (\ref{psiNoscfff}) then follows from Proposition \ref{PropPsiA}. We now compute  $\lim_{n\to \infty}\psi_n(t,x)$ in two steps.
\\
\\
We first observe that Formula (\ref{psiNoscfff})
can be written as
\[
\begin{split}
\psi_n(t,x)&=(\cos t)^{-1/2} \exp ( -(i/2) x^2\tan t)\
\\
&
\times \sum_{k=0}^nC_k(n,a) \exp (ix (1-2k/n)/\cos t  -(i/2) (1-2k/n)^2\tan t),
\end{split}
\]
and  the exponential $e^{-(i/2)(1-2k/n)^2\tan t}$ can be expanded in series as
$$
e^{-(i/2)(1-2k/n)^2\tan t}=\sum_{m=0}^\infty \dfrac{[-(i/2)(1-2k/n)^2\tan t]^m}{m!}.
$$
The identity
$$
(-\cos^2 t)^m\frac{\partial^{2m}}{\partial x^{2m}}e^{ix (1-2k/n)/\cos t}=(1-2k/n)^{2m}e^{ix (1-2k/n)/\cos t}
$$
and some additional computations give
\[
\begin{split}
\psi_n(t,x)=(\cos t)^{-1/2} e^{ -(i/2) x^2\tan t}\ \sum_{m=0}^\infty\frac{1}{m!}\Big(\frac{i}{2}\sin t\cos t\Big)^m \frac{\partial^{2m}}{\partial x^{2m}}F_n(x/\cos t).
\end{split}
\]
Now we consider the operator
\begin{equation}\label{OPRU}
U\left(\frac{d}{dx},t\right):=\sum_{m=0}^\infty\frac{1}{m!}\Big(\frac{i}{2}\sin t\cos t\Big)^m \frac{d^{2m}}{d x^{2m}}.
\end{equation}
  Note that $U$ is continuous on the space $A_{2,0}$ and the functions $F_n$ extend to entire functions of order less than or equal 1 and finite type (i.e. of exponential type), and this space is clearly contained in $A_{2,0}$.
Thus we can compute the limit below as follows
\[
\begin{split}
\psi(t,x):=\lim_{n\to \infty} \psi_n(t,x)&=(\cos t)^{-1/2} e^{ -(i/2) x^2\tan t}\ U(t) \lim_{n\to \infty} F_n(x/\cos t)
\\
&=(\cos t)^{-1/2} e^{ -(i/2) x^2\tan t}\ U(t)  F(x/\cos t)
\end{split}
\]
where $F(x):=e^{iax}$,
since, by Theorem \ref{carnot}, $ F_n(x/\cos t)$ converges uniformly to $F(x/\cos t)$ on the compact set $|x|\leq M$, where $M>0$, for every fixed $t$ in $[0,\pi/2)$.
\\
The continuity of the operator $\left(\frac{d}{dx},t\right)$ yields
\[
\begin{split}
\psi(t,x)&=(\cos t)^{-1/2} e^{ -(i/2) x^2\tan t}\ \left(\frac{d}{dx},t\right)  F(x/\cos t)
\\
&=
(\cos t)^{-1/2} e^{ -(i/2) x^2\tan t}\sum_{m=0}^\infty\frac{1}{m!}\Big(\frac{i}{2}\sin t\cos t\Big)^m \frac{d^{2m}}{d x^{2m}}\ e^{iax/\cos t }
\\
&=
(\cos t)^{-1/2} e^{ -(i/2) x^2\tan t}\sum_{m=0}^\infty\frac{1}{m!}\Big(\frac{i}{2}\sin t\cos t\Big)^m \Big(ia/\cos t\Big)^{2m}\ e^{iax/\cos t }
\\
&=
(\cos t)^{-1/2} e^{ -(i/2) x^2\tan t}\sum_{m=0}^\infty\frac{1}{m!}\Big(-\frac{i}{2}a^2\tan t\Big)^m \ e^{iax/\cos t },
\end{split}
\]
from which we obtain (\ref{psiNoscinfty}).
\end{proof}

\chapter{Superoscillating functions and convolution equations}

\section[Convolution operators for generalized Schr\"odinger equations]{Convolution operators for generalized Schr\"odinger equations}

We now consider the Cauchy problem associated with a modified version of the Schr\"{o}dinger equation.
In order to state the results, it is convenient to distinguish two cases, because different differential equations
are involved. As customary, $F_n(x,a)$ denotes the function defined in \ref{ef}. We begin with the following case:
\begin{theorem}\label{theorempolEVEN}
Consider, for $p$ even, the Cauchy problem for the modified Schr\"{o}dinger equation
\begin{equation}\label{modfSpolCAUCHYEVEN}
i\frac{\partial \psi(x,t)}{\partial t}=-\frac{\partial^p \psi(x,t)}{\partial x^p},\ \ \ \ \psi(x,0)=F_n(x,a).
\end{equation}
Then the solution $\psi_n(x,t;p)$, is given by
$$
\psi_n(x,t;p)=\sum_{k=0}^nC_k(n,a) e^{ix(1-2k/n)  } e^{it{(-i(1-2k/n))^p}}.
$$
Moreover, for all   $t\in [-T, T]$, where $T$ is any real positive number, we have
$$
\lim_{n\to\infty}\psi_n(x,t;p)=e^{it(-ia)^p} e^{iax},
$$
for $x\in K$, where $K$ is any compact set in $\mathbb{R}$.
\end{theorem}
\begin{proof}
We will prove the result by working in the space of the tempered distributions $\mathcal{S}'(\mathbb{R})$ and use a standard Fourier transform argument to solve  the Cauchy problem (\ref{modfSpolCAUCHYEVEN}).
We start with the equation
$$
i\frac{d \hat{\psi}(\lambda,t)}{d t}=-(-i\lambda)^p \ \hat{\psi}(\lambda,t)
$$
and, integrating, we obtain
$$
\hat{\psi}(\lambda,t)=C(\lambda)e^{i(-i\lambda)^p t}
$$
where the arbitrary function $C(\lambda)$ is determined by the initial condition
\[
\begin{split}
C(\lambda)=\hat{\psi}(\lambda,0)
&
=\sum_{k=0}^nC_k(n,a)\int_{\mathbb{R}}e^{i(1-2k/n){x}} e^{-i\lambda x}dx
\\
&=
2\pi \sum_{k=0}^nC_k(n,a) \delta (\lambda -(1-2k/n)),
\end{split}
\]
and where we have used the fact that in $\mathcal{S}'(\mathbb{R})$ it is
$
\mathcal{F}(e^{imx})=2\pi\delta (x-m).
$
Thus we have
$$
\hat{\psi}_n(\lambda,t;p)=2\pi\sum_{k=0}^nC_k(n,a) \delta (\lambda -(1-2k/n))e^{i(-i\lambda)^p t},
$$
and  taking the inverse Fourier transform we obtain
\[
\begin{split}
\psi_n(x,t;p)=&
\int_{\mathbb{R}}\Big[\sum_{k=0}^nC_k(n,a) \delta (\lambda -(1-2k/n))e^{it(-i\lambda)^p }\Big] e^{i\lambda x} d\lambda
\\
=&\sum_{k=0}^nC_k(n,a) e^{i(1-2k/n){x} } e^{it{(-i(1-2k/n))^p}}.
\end{split}
\]
Since we can write:
$$
e^{it (-i(1-2k/n))^p}=\sum_{m=0}^\infty \dfrac{[it (-i(1-2k/n))^p]^m}{m!},
$$
we get
\[
\begin{split}
\psi_n(x,t;p)&=\sum_{m=0}^\infty \dfrac{(it)^m}{m!}  \sum_{k=0}^nC_k(n,a)(-i(1-2k/n))^{mp} e^{ix(1-2k/n)}
\\
&=\sum_{m=0}^\infty \dfrac{((-1)^p  i t)^m}{m!}  \sum_{k=0}^nC_k(n,a)(i(1-2k/n))^{mp} e^{ix(1-2k/n)}
\\
&=\sum_{m=0}^\infty \dfrac{(  i t)^m}{m!}  \sum_{k=0}^nC_k(n,a)\frac{d^{mp}}{dx^{mp}} e^{ix(1-2k/n)}
\\
&=\sum_{m=0}^\infty \dfrac{(  i t)^m}{m!}  \frac{d^{mp}}{dx^{mp}}F_n(x,a).
\end{split}
\]
Since the operator
$$
\sum_{m=0}^\infty \dfrac{(  i t)^m}{m!}  \frac{d^{mp}}{dx^{mp}},
$$
is continuous, when we replace $x$ by $z$, in the space $\mathcal{A}_p$ we can pass to the limit and thanks to Theorem \ref{duality pq} we have:
$$
\psi(x,t;p)=\lim_{n\to \infty}\psi_n(x,t;p)=\sum_{m=0}^\infty \dfrac{(  i t)^m}{m!}  \frac{d^{mp}}{dx^{mp}}e^{iax}.
$$
Observing that
$$
\psi(x,t)=\sum_{m=0}^\infty \dfrac{(  i t)^m}{m!}  (ia)^{mp} e^{iax}
$$
we finally obtain
$$
\psi(x,t;p)=\sum_{m=0}^\infty \dfrac{((ia)^p  i t)^m}{m!}  e^{iax}=e^{it(ia)^p} e^{iax}.
$$
\end{proof}
At this point one may wonder
 if it is possible to compute the limit
$$
\lim_{n\to \infty} \sum_{k=0}^nC_k(n,a)e^{i(1-2k/n){x}}
$$
 when we replace $(1-2k/n)$ by  $(1-2k/n)^p$ where $p $ is an arbitrary
natural number and to obtain new superoscillating functions. In more precise term, we consider the following problem.
\begin{problem}\label{problem3.1}
Let  $a\in \mathbb{R}$, $t\in [-T,T]$ where $T$ is any real positive number.
Show that the sequence
$$
Y_n(t,a,p)=\sum_{k=0}^nC_k(n,a)e^{it(1-2k/n)^p }
$$
is $f$-superoscillating for $p\in \mathbb{N}$ and for a suitable function $f$.
\end{problem}

Theorem \ref{theorempolEVEN} gives the answer that solves Problem \ref{problem3.1}, in fact we have the following corollary:
\begin{corollary}\label{cororempol}
Let $a>1$, $p$ even, and let  $T$ be a real positive number. Then, for all $t\in[-T,T]$, the sequence
$$
u_n(t)=\sum_{k=0}^nC_k(n,a) e^{it{(-i(1-2k/n))^p}}
$$
is $e^{it(-ia)^p}$-superoscillating,  i. e. we have
$$
\lim_{n\to\infty}u_n(t)=e^{it(-ia)^p}.
$$
\end{corollary}

 When $p$ is an odd number we have to modify the differential equation in Theorem \ref{theorempolEVEN}, removing the imaginary unit in front of the time derivative  and changing the sign,  in order to  get the
analogue of  Theorem  \ref{theorempolEVEN} and of Corollary  \ref{cororempol}.
\begin{theorem}\label{theorempolODD}
Consider, for $p$ odd, the Cauchy problem for the mo\-di\-fi\-ed Schr\-\"{o}\-din\-ger equation
\begin{equation}\label{modfSpolCAUCHYODD}
\frac{\partial \psi(x,t)}{\partial t}=\frac{\partial^p \psi(x,t)}{\partial x^p},\ \ \ \ \psi(x,0)=F_n(x,a).
\end{equation}
Then the solution $\psi_n(x,t;p)$, is given by
$$
\psi_n(x,t;p)=\sum_{k=0}^nC_k(n,a) e^{i(1-2k/n) x } e^{t(-i(1-2k/n))^p}.
$$
Moreover, for all  $t\in [-T, T]$, where $T$ is any real positive number, and all $x\in\mathbb R$, we have
$$
\lim_{n\to\infty}\psi_n(x,t,p)=e^{t(-ia)^p} e^{iax}.
$$
The limit is uniform for $x$ in the compact sets of $\mathbb{R}$.
\end{theorem}
\begin{proof}
The proof follows the lines of the proof of Theorem \ref{theorempolEVEN}. In fact, we observe that
\[
\begin{split}
\psi_n(x,t,p)=&\int_{\mathbb{R}}
\Big[\sum_{k=0}^nC_k(n,a) \delta (\lambda -(1-2k/n))e^{(-i\lambda)^p t}\Big] e^{i\lambda x} d\lambda
\\
=&\sum_{k=0}^nC_k(n,a) e^{i(1-2k/n){x} }e^{t(-i(1-2k/n))^p}.
\end{split}
\]
To pass to the limit, we use
the strategy used in the proof of Theorem \ref{theorempolEVEN} and we get the statement.
\end{proof}
One may ask what happens if we use a different modification of the superoscillating datum,  specifically, if we
use a datum of the form
$$
\sum_{k=0}^nC_k(n,a)e^{-ix(1-2k/n)^\ell}
$$
for the modified Schr\"{o}dinger equation. The next result shows that, in this case,
we do not further enlarge the class of  superoscillating sequences.
\begin{theorem}\label{theorempoldue}
Let $a>1$, let $p$ be even, $\ell\in \mathbb{N}$ and  $t\in [-T, T]$, where $T$ is any real positive number.
Consider the Cauchy problem for the modified Schr\"{o}dinger equation
\begin{equation}\label{modfSpolCAUCHY}
i\frac{\partial \psi(x,t)}{\partial t}=-\frac{\partial^p \psi(x,t)}{\partial x^p},\ \ \ \
\psi(x,0)=\sum_{k=0}^nC_k(n,a)  e^{-ix(1-2k/n)^\ell}.
\end{equation}
Then the solution $\psi_n(x,t;p)$, is given by
$$
\psi_n(x,t;p)=\sum_{k=0}^nC_k(n,a)e^{-ix(1-2k/n)^\ell}e^{it (-i(1-2k/n)^\ell)^p}
$$
Moreover, for all  $t\in [-T, T]$, where $T$ is any real positive number and any $x\in\mathbb R$, we have
$$
\lim_{n\to\infty}\psi_n(x,t)=e^{it(-i)^{p}a^{p\, \ell}} e^{iax}.
$$
The limit is uniform for $x$ in the compact sets of $\mathbb{R}$.
\end{theorem}
\begin{remark}
{\rm   Using the nomenclature of Definition \ref{superoscill}, we see that
the previous result shows that  $g(a)=a^{p\ell}$ and thus it does not provide any generalization with respect to Theorem \ref{theorempolEVEN} and \ref{theorempolODD}.
If we want to find a more general limit function $e^{ig(a)x}$, then we have to leave the realm of the differential equations as we show in the next result.
}
\end{remark}

We will now consider a much more general situation  in which the right hand side of the differential equation to solve is an infinite series of derivatives. This will lead us to consider convolution equations and to use the general theory of AU-spaces.

Let $\{a_p\}$ be a sequence of complex numbers and  consider the convolution equation formally defined by
\begin{equation}\label{nome}
 i\frac{\partial \psi(x,t)}{\partial t}=-\sum_{p=0}^{\infty} a_p\frac{\partial^p \psi(x,t)}{\partial x^p}.
\end{equation}
 In order to understand if the superoscillatory behavior persists when we take a superoscillating initial datum, we first need to understand on which space the infinite series of derivatives actually operates. As we saw in Section \ref{sec5.1}, to study the Schr\"odinger equation we are naturally led to the study of the convolution operator
\[
 U_2\left(\frac{d}{dz},t\right):=\sum_{m=0}^{\infty} \frac{(it)^m}{m!}\frac{d^{2m} }{dz^{2m}}.
\]
Similarly, as we have shown, if one wants to study the modified equation
\[
 i\frac{\partial \psi(x,t)}{\partial t}=-\frac{\partial^p \psi(x,t)}{\partial x^p},
\]
one naturally needs to consider the operator:
\[
U_p\left(\frac{d}{dz},t\right):=\sum_{m=0}^{\infty} \frac{(it)^m}{m!}\frac{d^{pm} f}{dz^{pm}},
\]
where, depending on the parity of $p\in \mathbb{N}$, one might have to replace $it$ by $-it$.
Then if we want to solve equation (\ref{nome}), we will see that we need to study
the operator that can formally be written as the infinite product of the operators we have just considered, i.e. the operator
$$
U_{\infty}\left(\frac{d}{dz},t\right) = \prod_{p=0}^{\infty} \left(\sum_{m=0}^{\infty} \frac{(ita_p)^m}{m!}\frac{d^{p m}}{dz^{pm}}\right)=\prod_{p=0}^{ \infty} U_p \left(\frac{d}{dz}, a_p t\right).
$$

This operator can actually be regarded as the operator associated with the multiplier given by the function
\[
 \hat{U}_\infty (\zeta,t):= \prod_{p=0}^{\infty}\left(\sum_{m=0}^{\infty} \frac{(ita_p)^m}{m!}\zeta^{p m}\right).
\]
Thus this multiplier can be written in the form
\[
\begin{split}
 \hat{U}_\infty (\zeta,t) &= \prod_{p=0}^{\infty}\left(\sum_{m=0}^{\infty}
\frac{1}{m!} \left(ita_p\zeta^p\right)^m\right)
\\
&=\prod_{p=0}^{\infty}\exp(ita_p\zeta^p)
\\
&=\exp \left(it\sum_{p=0}^{\infty} a_p \zeta^p\right).
\end{split}
\]
Under suitable conditions on the sequence $\{a_p\}$, the function
$\hat U_\infty (\zeta,t)$ is holomorphic,  as a function of $\zeta$, in the open disc $|\zeta|<1$. This is true, for example, in the case $a_p=1$, for all $p$. In fact we obtain $\hat U_\infty (\zeta,t)=\exp\left(\frac{it}{1-\zeta}\right)$.
As a consequence, and in view of Proposition \ref{esempiodelta1}, the operator $U_\infty \left(\frac{d}{dz},t\right)$ acts continuously on the space
$
 {\rm Exp}_1(\mathbb{C}).
$
\begin{remark}
 {\rm Under stronger restrictions on the sequence $a_p$, one has that the function
$\hat U_\infty (\zeta,t)$ can be made holomorphic on a disc of arbitrary radius $R$. For example, in the case the coefficients $a_p$ are such that they define an infinite order differential operator, the function
$\sum_{n=0}^{\infty} a_p \zeta^p$ is entire and so is $\hat U_\infty(\zeta,t)$. }
\end{remark}
Noting that every function of the form
\[
 Y_n(z,a):= \sum_{j=0}^n C_j(n,a)e^{ik_j(n)z}
\]
with $|k_j(n)|\leq 1$ belong to ${\rm Exp}_1(\mathbb{C})$,
we can summarize this discussion in the following result.
\begin{theorem}\label{3punto2}
 Let $\{a_p\}$
be a sequence of complex numbers such that the function $\sum_{n=0}^{\infty} a_p \zeta^p$ is analytic in the disc $|\zeta| <R$ for some $R>0$. Then the function
\[
 \hat U_\infty(\zeta,t)=\exp\left(it\sum_{n=0}^{\infty} a_p \zeta^p\right)
\]
is a continuous multiplier on the space of functions analytic in the disc $|\zeta|<R$ and the associated operator
\[
U_{\infty}\left(\frac{d}{dz},t\right)=\prod_{p=0}^{\infty} \left(\sum_{m=0}^{\infty} \frac{(ita_p)^m}{m!}\frac{d^{p m}}{dz^{pm}}\right)=\prod_{p=0}^{\infty} U_p\left(\frac{d}{dz},a_p t\right)
\]
acts continuously on the space of entire functions of exponential type less than $R$.
\end{theorem}
\begin{proof}
This is an immediate consequence of Proposition \ref{esempiodeltaR} and the previous discussion.
\end{proof}
\begin{remark}{\rm
It is immediate that by the root test, the function $\sum_{p=0}^{\infty}a_p\zeta^p$ is analytic in $\Delta_R$ if and only if $\lim_{p\to\infty}R \, \sqrt[p]{|a_p|}<1$.}
\end{remark}

\begin{theorem}\label{theorempoltre}
Let $a\in\mathbb{R}$, $a>1$. Consider a sequence of complex numbers $\{a_p\}$ such that the function $\sum_{p=0}^{\infty} a_p\zeta^p$ is holomorphic in $\Delta_{a'}$ for $a'>a$ and assume that $G(ia)$ is real and $|G(ia)|\geq 1$. Consider, in the space of entire functions of exponential type less than $a'$,
 the Cauchy problem for the generalized Schr\"{o}dinger equation
\begin{equation}\label{modfSpolCAUCHYgen}
i\frac{\partial \psi(z,t)}{\partial t}=- G\left(\frac{d}{dz}\right)\psi(z,t),\ \ \ \
\psi(z,0)=F_n(z,a),
\end{equation}
where
$$
G\left(\frac{d}{dz}\right)=\sum_{p=0}^\infty a_p\frac{d^p}{dz^p}.
$$
Then the solution $\psi_n(z,t)$, is given by
$$
\psi_n(z,t)=\sum_{k=0}^nC_k(n,a)e^{-iz(1-2k/n)}e^{it G(-i(1-2k/n))}.
$$
Moreover, for all fixed $t$ we have
$$
\lim_{n\to\infty}\psi_n(z,t)=e^{it G(ia)} e^{iaz},
$$
and the convergence is uniform on all compact sets of $\mathbb{C}$.
\end{theorem}
\begin{proof}
Using the previous method we have
$$
i\frac{d \hat{\psi}(\lambda,t)}{d t}=-G(-i\lambda) \ \hat{\psi}(\lambda,t)
$$
and integrating we obtain
$$
\hat{\psi}(\lambda,t)=C(\lambda)e^{it G(-i\lambda)},
$$
where the arbitrary function $C(\lambda)$ can be determined by the initial condition
$$
C(\lambda)=\hat{\psi}(\lambda,0)=
\sum_{k=0}^nC_k(n,a) \delta (\lambda -(1-2k/n)).
$$
We then have
$$
\hat{\psi}(\lambda,t)=\sum_{k=0}^nC_k(n,a) \delta (\lambda -(1-2k/n))e^{i t G(-i\lambda)},
$$
so
$$
\psi_n(z,t)=\sum_{k=0}^nC_k(n,a)e^{-iz(1-2k/n)}e^{it G(-i(1-2k/n))}.
$$
Now we observe that
$$
\psi_n(z,t)=\sum_{k=0}^nC_k(n,a)e^{-iz(1-2k/n)}e^{it \sum_{p=0}^\infty a_p(-i(1-2k/n))^p}
$$
so
\[
\begin{split}
\psi_n(z,t)=&\sum_{k=0}^nC_k(n,a)e^{-iz(1-2k/n)}\prod_{p=0}^\infty e^{it a_p(-i(1-2k/n))^p}
\\
=&\sum_{k=0}^nC_k(n,a)e^{-iz(1-2k/n)}\prod_{p=0}^\infty \sum_{m=0}^\infty \frac{(it a_p)^m}{m!} (-i(1-2k/n))^{mp}
\\
=&\prod_{p=0}^\infty \sum_{m=0}^\infty \frac{(it a_p)^m}{m!}\sum_{k=0}^nC_k(n,a) (i(1-2k/n))^{mp}e^{-iz(1-2k/n)}
\end{split}
\]
and we get
\[
\begin{split}
\psi_n(z,t)=&\prod_{p=0}^\infty \sum_{m=0}^\infty \frac{(it a_p)^m}{m!}\sum_{k=0}^nC_k(n,a) \frac{d^{mp}}{d^{mp}} e^{-iz(1-2k/n)}
\\
=&\prod_{p=0}^\infty \sum_{m=0}^\infty \frac{(it a_p)^m}{m!} \frac{d^{mp}}{d^{mp}} \sum_{k=0}^nC_k(n,a)e^{-iz(1-2k/n)}
\\
=&\prod_{p=0}^\infty \sum_{m=0}^\infty \frac{(it a_p)^m}{m!} \frac{d^{mp}}{d^{mp}} F_n(z,a).
\end{split}
\]
Thanks to Theorem \ref{3punto2}, we can pass to the limit for $n\to \infty$
\[
\begin{split}
\psi(z,t)&=\prod_{p=0}^\infty \sum_{m=0}^\infty \frac{(it a_p)^m}{m!} \frac{d^{mp}}{d^{mp}} e^{iaz}
\\
&=\prod_{p=0}^\infty \sum_{m=0}^\infty \frac{(it a_p)^m}{m!} (ia)^{mp}e^{iaz}
\end{split}
\]
so we have
\[
\begin{split}
\psi(z,t)&=\prod_{p=0}^\infty \sum_{m=0}^\infty \frac{(it a_p(ia)^{p})^m}{m!} e^{iaz}
\\
&=\prod_{p=0}^\infty e^{(it a_p(ia)^{p})} e^{iaz}
\\
&= e^{it \sum_{p=0}^\infty(a_p(ia)^{p})} e^{iaz},
\end{split}
\]
and we obtain
$$
\lim_{n\to\infty}\psi_n(z,t)=e^{it G(ia)} e^{iaz}.
$$
\end{proof}
\begin{remark}{\rm By setting $g(a)=G(ia)$ and by suitably choosing the coefficients $a_p$ of the series expressing $G$,  we can obtain
a very large class of superoscillating functions.}
\end{remark}

\section{Formal solutions to Cauchy problems for linear constant coefficients differential equations}
\label{sec6.1}

Formal solutions of Cauchy problems for linear constant coefficients partial differential equations in the complex domain are known since the early work of S. Kowalevski who considered, in \cite{sonya}, the characteristic Cauchy problem for the complex heat equation:
$$
\frac{\partial}{\partial t} u(z,t)= \frac{\partial^2}{\partial z^2}  u(z,t),\ \ \ \ \ u(z,0)=\varphi(z),
$$
where $t$ and $z$ are complex  variables and $\varphi$ is a function holomorphic in a neighborhood of the origin.
In \cite{sonya} it is proved that the unique formal solution
$$
u(z,t)= \sum_{m=0}^\infty \frac{t^m}{m!}\frac{\partial^{2m}}{\partial z^{2m}} \varphi(z)
$$
of the Cauchy problem
converges if and only if $\varphi(z)$ is an entire function of exponential order at most $2$.
Since then, the question of how to construct formal solutions to generalizations of the heat equation, and how to ensure their convergence, has been considered by several mathematicians, see e.g. \cite{balser 2004}, \cite{ichinobe}, \cite{lms},  and the references therein.
\\
We will consider the case, see  \cite{ichinobe}, in which one studies the Cauchy problem associated with the differential equation
\begin{equation}\label{equadiffgeneral}
\frac{\partial^{r\nu} }{\partial t^{r\nu}}u(z,t)-\sum_{j=1}^{\nu} a_j\frac{\partial^{r(\nu-j)}}{\partial t^{r(\nu-j)}}\frac{\partial^{jp} }{\partial z^{jp}} u(z,t)=0
\end{equation}
where $r,p, \nu \in\mathbb N$, $1\leq r<p$, $\nu\geq 1$, $a_j\in\mathbb C$ and with initial conditions
\begin{equation}\label{IC}
\frac{\partial^{\ell} }{\partial t^{\ell}}u(z,0)=0, \ \ \ell=0,\ldots,r\nu -2, \ \
\frac{\partial^{r\nu-1}}{\partial  t^{r\nu-1}}u(z,0)=\varphi(z).
\end{equation}
The differential equation (\ref{equadiffgeneral}) can be rewritten as
$$
P\left(\frac{\partial}{\partial z}, \frac{\partial}{\partial t}\right)u(z,t)=\prod_{j=1}^\mu P_j^{\ell_j}u(z,t)=0,
 $$
 where
$$
P_j=P_j \left(\frac{\partial}{\partial z}, \frac{\partial}{\partial t}\right) =\frac{\partial^r}{\partial t^r}-\alpha_j \frac{\partial^p}{\partial z^p},
$$
$\ell_j$, $\mu$ are suitable natural numbers, and $\alpha_j \in \mathbb{C}$ are the distinct roots of the characteristic equation
\begin{equation}\label{chareq}
\lambda^\nu-\sum_{j=1}^\nu a_j\lambda^{\nu -j}=0.
\end{equation}

\begin{theorem}\label{FORMALCAU}
The formal solution to the differential equation (\ref{equadiffgeneral}) with initial conditions (\ref{IC}) is given by
$$
\tilde u(z,t)=\sum_{m\geq r\nu-1} u_m(tz) \frac {t^m}{m!}=\sum_{m\geq 0} u_{rm+r\nu-1}(tz) \frac {t^{rm+r\nu-1}}{(rm+r\nu-1)!},
$$
where
$$
u_{rm+r\nu-1}(z)=A(m) \varphi^{(pm)} (z),\qquad m\geq 0
$$
and $A(m)$ are coefficients that can be explicitly computed by solving a suitable difference equation.
\end{theorem}
\begin{remark}\label{RKAm}
{\rm
 It turns out, see \cite{ichinobe} for the details, that the coefficients $A(m)$  are of the form
$$
A(m)=\sum_{j=1}^\mu \alpha_j^m\sum_{k=1}^{\ell_j} c_{jk} m^{k-1},\qquad m\geq 0,
$$
where $\ell_j$, $\mu$, $\alpha_j$ are as above and $c_{jk}$ are suitable complex numbers.
}
\end{remark}

 Theorem \ref{FORMALCAU} and Remark \ref{RKAm} yield the following result:
\begin{corollary}\label{corFORMALCAU} The formal solution to the Cauchy problem given by the differential equation (\ref{equadiffgeneral}) with initial conditions (\ref{IC}), and $x$ replaced by $z$, is given by
\begin{equation}\label{formalsol}
\tilde u(z,t)=\sum_{m\geq 0} t^{rm+r\nu -1}\frac{A(m)}{(rm+r\nu -1)!} \frac{d^{pm}}{dz^{pm}} \varphi (z):=\tilde U_p\left(\frac{d}{dz},t\right) \varphi (z).
\end{equation}
\end{corollary}
We now want to study the case when we specifically take $\varphi(z)$ to be a function of the form $F_n(z,a)$. To do so, we need to take the limit of
 $\tilde U_p(\frac{d}{dz},t)F_n(z,a)$ when $n$ becomes arbitrarily large, where the operator $\tilde U_p(\frac{d}{dz},t)$ is defined in (\ref{formalsol}). We will fully answer this question in Section \ref{sec6.3}.

\section{Differential equations of non-Kowalevski type}
\label{sec6.3}

Let us now go back to the differential equation (\ref{equadiffgeneral}) whose formal solution is given in Theorem \ref{FORMALCAU}, and let us study some special cases in which we can guarantee that, when we assign a superoscillatory
initial condition in the Cauchy problem (\ref{IC}) we have that the solution to (\ref{equadiffgeneral}) is superoscillating.
We begin with the case $\nu=r=1$, which is a very simple modification of what we have done in Section \ref{sec6.1}.
\begin{theorem}
The solution $\psi_n(x,t)$ to the Cauchy problem
\begin{equation}
\left\{
\begin{array}{c}
\displaystyle\frac{\partial }{\partial t}\psi (x,t)= a_1\displaystyle\frac{\partial^{p} }{\partial z^{p}} \psi(x,t)\\
\ \
\\
\psi (x,0)= F_n (x,a)\\
\end{array}
\right.
\end{equation}
is such that
\[
\psi (x,t)=\lim_{n\to \infty}\psi_n(x,t)=e^{t a_1 (ia)^p } e^{iax}.
\]
and $\psi (x,t)$
is superoscillating in time when:
\begin{enumerate}
\item[i)] $p$ is even, and  $a_1$ is purely imaginary or $a_1=\alpha+i\beta$ and $(ia)^p\alpha >0$, in which case the superoscillation is amplified, or $a_1=\alpha+i\beta$ and $(ia)^p\alpha <0$, in which case the superoscillation is damped.
\item[ii)] $p$ is odd, and $a_1$ is real or $a_1=\alpha+i\beta$ and $i(ia)^p\beta >0$, in which case the superoscillation is amplified, or $a_1=\alpha+i\beta$ and $i(ia)^p\beta <0$, in which case the superoscillation is damped.
\end{enumerate}
\end{theorem}
\begin{proof}
To prove the result we consider the previous problem in the complex domain by substituting $x$ with a complex variable $z$.
 Theorem \ref{FORMALCAU} yields that the formal solution is
$$
\psi_n(z,t)=\sum_{m\geq 0} t^{m}\frac{A(m)}{m!} \frac{d^{pm}}{dz^{pm}} F_n (z,a)
$$
where $A(m)=a_1^m$, since the characteristic equation is $\lambda -a_1=0$. By taking the limit, see Theorem \ref{3punto2}, we have
\[
\begin{split}
\lim_{n\to \infty}\psi_n(z,t)&=\sum_{m\geq 0} t^{m}\frac{a_1^m}{m!} \frac{d^{pm}}{dz^{pm}} F_n (z,a)\\
&=\sum_{m\geq 0} t^{m}\frac{a_1^m}{m!} \frac{d^{pm}}{dz^{pm}} e^{iaz}\\
&=\sum_{m\geq 0} (a_1 t)^{m}{m!} (ia)^{pm} e^{iaz}\\
&=e^{t a_1 (ia)^p } e^{iaz}.\\
\end{split}
\]
Assume first that $p$ is even. Then
$e^{t a_1 (ia)^p }=e^{(-1)^{p/2} t a_1 a^{p} }$, thus if $a_1$ is purely imaginary the superoscillation persists in time.
If $a_1=\alpha+i\beta$ then
$$
e^{t a_1 (ia)^p }=e^{t(\alpha+i\beta) (ia)^p }= e^{t\alpha(ia)^p }e^{it\beta (ia)^p }.
$$
The factor
$e^{t\alpha(ia)^p }$ is amplifying or damping the superoscillation according to the sign of $\alpha(ia)^p$.
\\
If $p$ is odd then $e^{t a_1 (ia)^p }=e^{i (-1)^{(p-1)/2} t a_1 a^{p} }$, thus if $a_1$ is real the superoscillation persists in time.
If $a_1=\alpha+i\beta$ then
$$
e^{t a_1 (ia)^p }=e^{t(\alpha+i\beta) (ia)^p }= e^{t \alpha(ia)^p }e^{it\beta (ia)^p }.
$$
The factor $e^{it\beta (ia)^p }$ is amplifying or damping the superoscillation according to the sign of $i\beta(ia)^p$. By restricting to the real axis we obtain the statement.
\end{proof}
We can extend this result to a more general differential equation: consider the symbol $g(\zeta ,t)$ of the operator $U_p(\frac{d}{dz},t)$:
$$
g(\zeta ,t)= \sum_{m\geq 0} t^{rm+r\nu -1}\frac{A(m)}{(rm+r\nu -1)!} \zeta^{pm}.
$$
We now prove the following result:
\begin{proposition}\label{order}
The function
$$
g(\zeta ,t)= \sum_{m\geq 0} t^{rm+r\nu -1}\frac{A(m)}{(rm+r\nu -1)!} \zeta^{pm}
$$
is an entire function in $\zeta$ of order $p/r$.
\end{proposition}
\begin{proof}
To prove that a function $f(z)=\sum_{m\geq 0} c_mz^m$ is entire, it is enough to show that the number
$$
\rho=\overline{\lim}_{m\to\infty} \frac{m\log m}{\log\left|\displaystyle\frac{1}{c_m}\right|}
$$
is finite. This, by formula (1.05) in \cite{levin} will also automatically prove that the order of the entire function is $\rho.$
In the case of the function $g(\zeta)$ the formula gives, for $t\in[-T,T]$:
\[
\begin{split}
\rho=& \overline{\lim}_{m\to\infty} \frac{pm\log(pm)}{\log\left(\displaystyle\frac{(rm+r\nu-1)!}{|A(m) t^{rm+r\nu-1}|}\right)}.\\
\end{split}
\]
By applying the Stirling formula to the factorial, we can rewrite as follows:
\[
\begin{split}
\rho=&\overline{\lim}_{m\to\infty} \frac{pm(\log(m)+\log(p))}{(rm+r\nu-1)\log (rm+r\nu-1)-\log( |A(m)| |t|^{rm+r\nu-1})}\\
=&\overline{\lim}_{m\to\infty} \frac{pm(\log(m)+\log(p))}{rm(\log (m)+\log (r))-\log( |\alpha|^m m^M)-(rm+r\nu-1)\log( |t|)}\\
=&\overline{\lim}_{m\to\infty} \frac{pm\log(m)}{rm\log (m)-m\log( |\alpha|)-M\log(m)-rm\log( |t|)}=\frac pr.\\
\end{split}
\]
\end{proof}
We now prove a result in which the derivative with respect to time can be of order higher that $1$.
\begin{theorem}
The Cauchy problem associated with the differential equation
\begin{equation}
\frac{\partial^{r\nu} }{\partial t^{r\nu}}\psi(x,t)=\sum_{j=1}^{\nu} a_j\frac{\partial^{r(\nu-j)}}{\partial t^{r(\nu-j)}}\frac{\partial^{jp} }{\partial x^{jp}} \psi (x,t)
\end{equation}
where $r,p, \nu \in\mathbb N$, $1\leq r<p$, $\nu\geq 1$, $a_j\in\mathbb C$ and with initial conditions
\begin{equation}
\frac{\partial^{\ell} }{\partial t^{\ell}}\psi (x,0)=0, \ \ \ell=0,\ldots,r\nu -2, \ \
\frac{\partial^{r\nu-1}}{\partial  t^{r\nu-1}}\psi (x,0)=F_n(x,a).
\end{equation}
has formal solution
\begin{equation}\label{formalsolpsi}
\psi_n(x,t)=\sum_{m\geq r\nu-1} u_{m,n}(tx) \frac {t^m}{m!}=\sum_{m\geq 0} u_{rm+r\nu-1,n}(tx) \frac {t^{rm+r\nu-1}}{(rm+r\nu-1)!},
\end{equation}
where
$$
u_{rm+r\nu-1,n}(x)=A(m) F_n^{(pm)} (x,a),\qquad m\geq 0.
$$
Moreover, this formal solution actually converges to a function $\psi$ and it is possible to compute its limit
when $n\to\infty$ to be the function
$$
\psi(x,t)=\sum_{m\geq r\nu-1} u_m(tx) \frac {t^m}{m!}=\sum_{m\geq 0} u_{rm+r\nu-1}(tx) \frac {t^{rm+r\nu-1}}{(rm+r\nu-1)!},
$$
where
$$
u_{rm+r\nu-1}(x)=A(m) F^{(pm)} (x),\qquad m\geq 0.
$$
\end{theorem}
\begin{proof}
  We begin by complexifying the Cauchy problem and by Corollary \ref{corFORMALCAU} we know that the formal solution
  $\psi_n(z,t)$
  is obtained by applying the convolutor
  $$
  \sum_{m=0}^\infty \displaystyle\frac {t^{rm+r\nu-1}}{(rm+r\nu-1)!} A(m) \displaystyle\frac{d^{pm}}{dz^{pm}}
  $$
   to $F_n(z,a).$ Since the symbol of the convolutor has order $p$, see Proposition \ref{order}, and thus it belongs to $\mathcal{A}_p$, we see that the limit can be taken inside the series as long as the initial condition is a function in $\mathcal{A}_{p',0}$  and the statement follows by observing that $p\geq 1$ and therefore $p'\geq 1$, that $F_n\in \mathcal{A}_{p',0}$ and, finally, by restricting to the real axis.
\end{proof}

\bigskip

\noindent{\it Heat equation}

\bigskip

We now consider, following the lines of the previous discussion, a Cauchy problem for the heat equation:
$$
\frac{\partial}{\partial t} \psi (x,t)= \frac{\partial^2}{\partial z^2}  \psi (x,t),\ \ \ \ \ \psi(x,0)=F_n(x,a).
$$
With techniques similar to those used to treat the Schr\"odinger equation, we deduce that the formal solution to the complex version of this problem is, see (\ref{formalsolpsi}):
$$
\psi_n(z,t)=\sum_{m\geq 0} \frac{t^m}{m!} \frac{d^{2m}}{dz^{2m}} F_n(z,a).
$$
Computations similar to those done in the proof of Theorem \ref{theorempolEVEN}, and noticing again that the symbol of the operator belongs to $\mathcal{A}_2$, show that when we restrict to the real axis we have that
for all $t\in[-T,T]$, $T>0$
$$
\lim_{n\to\infty}\psi_n(x,t)=e^{-a^2t} e^{iax},
$$
thus the superoscillation is damped in time.

\section{An application to the harmonic oscillator}

We conclude by considering a more general type of superoscillating initial datum for the harmonic oscillator.

\begin{theorem}\label{TH24} Let $p$ even. Consider the superoscillating  function
$$
Y_n(x)=\sum_{k=0}^n C_k(n,a)e^{ix(-i(1-2k/n))^p}.
$$
 Then the solution of the  Cauchy problem
\begin{equation}\label{CauchyHarmP}
i\frac{\partial \psi(t,x)}{\partial t}=\frac{1}{2}\Big( -\frac{\partial^2 }{\partial x^2} +x^2 \Big)\psi(t,x),\ \ \ \ \psi(0,x)=Y_n(x)
\end{equation}
is given by
\begin{equation}
\begin{split}
&\psi_n(t,x)=(\cos t)^{-1/2} \exp ( -(i/2) x^2\tan t)\
\\
&
\times \sum_{k=0}^nC_k(n,a)\exp (ix (-i(1-2k/n))^p/\cos t-(i/2) (-i(1-2k/n))^{2p}\tan t).
\end{split}
\end{equation}
Moreover, if we set $\psi(t,x)=\lim_{n\to \infty}\psi_n(t,x)$, then
\begin{equation}\label{psiNoscPP}
\psi(t,x)=(\cos t)^{-1/2} \exp ( -(i/2) (x^2+a^{2p})\tan t +i(-ia)^px/\cos t).
\end{equation}
\end{theorem}
\begin{proof}
The solution (\ref{CauchyHarmP}) is obtained using Proposition \ref{PropPsiA} and as initial datum the sequence
$$
Y_n(x)=\sum_{k=0}^nC_k(n,a) e^{ix{(-i(1-2k/n))^p}}.
$$
To prove the second part of the  theorem we write (\ref{psiNoscPP}) as
\[
\begin{split}
&\psi_n(t,x)=(\cos t)^{-1/2} \exp ( -(i/2) x^2\tan t)\
\\
&
\times\sum_{k=0}^nC_k(n,a) \exp (ix (-i(1-2k/n))^p/\cos t-(i/2) (-i(1-2k/n))^{2p}\tan t),
\end{split}
\]
and we replace  the term $\exp({-(i/2)(1-2k/n)^2\tan t})$ by the series expansion
$$
\exp({-(i/2) (-i(1-2k/n))^{2p}\tan t})=\sum_{m=0}^\infty \dfrac{[-(i/2) (-i(1-2k/n))^{2p}\tan t]^m}{m!},
$$
so we get
\[
\begin{split}
\psi_n&(t,x)=(\cos t)^{-1/2} \exp({ -(i/2) x^2\tan t})\
\\
&
\times\sum_{k=0}^nC_k(n,a)\exp({ix (-i(1-2k/n))^p/\cos t})
\sum_{m=0}^\infty \dfrac{[-(i/2) (-i(1-2k/n))^{2p}\tan t]^m}{m!}
\\
&
=
(\cos t)^{-1/2} \exp({ -(i/2) x^2\tan t})\
\\
&
\times
\sum_{m=0}^\infty \frac{1}{m!} \sum_{k=0}^nC_k(n,a)  [-(i/2) (-i(1-2k/n))^{2p}\tan t]^m
\\
&
\times \exp [ix (-i(1-2k/n))^p/\cos t].
\end{split}
\]
Now observe that
\[
\begin{split}
\frac{\partial^m}{\partial x^m}& \exp({ix (-i(1-2k/n))^p/\cos t})\\
&=\Big(\frac{i}{\cos t}\Big)^m (-i(1-2k/n))^p \exp({ix (-i(1-2k/n))^p/\cos t}),
\end{split}
\]
so that
$$
\psi_n(t,x)=(\cos t)^{-1/2} \exp({ -(i/2) x^2\tan t})\ \sum_{m=0}^\infty\frac{1}{m!}\Big(\frac{i}{2}\sin t\cos t\Big)^m \frac{\partial^{2m}}{\partial x^{2m}}Y_n(x/\cos t).
$$
Since it is
$$
\lim_{n\to\infty}Y_n(x)=e^{ix(-ia)^p}:=Y(x)
$$
uniformly on the compact sets, thanks to Step 2 in the proof of Theorem \ref{TH23}, we get:
$$
\psi(t,x)=(\cos t)^{-1/2} e^{ -(i/2) x^2\tan t}\ \sum_{m=0}^\infty\frac{1}{m!}\Big(\frac{i}{2}\sin t\cos t\Big)^m \frac{\partial^{2m}}{\partial x^{2m}}Y(x/\cos t)
$$
so
\[
\begin{split}
\psi(t,x)&=(\cos t)^{-1/2} \exp({ -(i/2) x^2\tan t})\
\\
&
\times\sum_{m=0}^\infty\frac{1}{m!}\Big(\frac{i}{2}\sin t\cos t\Big)^m \Big(\frac{-a^p}{\cos^2t}\Big)^m\exp({ix(-ia)^p/\cos t}),
\end{split}
\]
i.e.
\[
\begin{split}
\psi(t,x)&=(\cos t)^{-1/2} e^{ -(i/2) x^2\tan t}\
\\
&
\times
\sum_{m=0}^\infty\frac{1}{m!}\Big(-\frac{i}{2} a^{2p}\sin t\cos t\Big)^m e^{ix(-ia)^p/\cos t}.
\end{split}
\]
The statement follows.
\end{proof}

\begin{remark}{\rm
A similar result holds also in the case $p$ is an odd number and the sequence is $Z_n(x)=\sum_{k=0}^n C_k(n,a)e^{x(-i(1-2k/n))^p}$. }
\end{remark}

\bigskip

\noindent{\em Driven harmonic oscillator}

\bigskip

Consider the time-dependent Schr\"{o}dinger equation
\begin{equation} \begin{split}
  \left[i\hbar\frac{\partial}{\partial t}-H(t,x)\right]\psi(t,x)=0
  \label{schrodinger}
\end{split}  \end{equation}
for a Hamiltonian $H(t,x)$.
The Green's function $G(t,x,t',x')$ for equation \eqref{schrodinger} satisfies
\begin{equation} \begin{split}
  \left[i\hbar\frac{\partial}{\partial t}-H(t,x)\right]G(t,x,t',x')=i\hbar\delta(t-t')\delta(x-x').\label{green_function}
\end{split}  \end{equation}
It follows that
\begin{equation} \begin{split}
  G(t,x,t',x')=\theta(t-t')\tilde{G}(t,x,t',x'),\\
  \left[i\hbar\frac{\partial}{\partial t}-H(t,x)\right]\tilde{G}(t,x,t',x')=0,\label{propagator}\\
  \tilde{G}(t,x,t,x')=\delta(x-x'),
\end{split}  \end{equation}
where $\tilde{G}(t,x,t',x')$ is the propagator for equation \eqref{schrodinger}, and
\[ \begin{split}
  \psi(t,x)=\int_{\mathbb{R}}\tilde{G}(t,x,t',x')\psi(t',x')dx'.
  \label{}
\end{split}  \]

Explicit solutions for Green's functions and propagators are known only for a few physical systems.
One of them is described by the Hamiltonian
\[ \begin{split}
  H(t,x)=-\frac{\hbar^2}{2m}\frac{\partial^2}{\partial x^2}+\frac{1}{2}m\omega^2(t) x^2-f(t)x.
  \label{}
\end{split}  \]
Physically this Hamiltonian represents a harmonic oscillator of mass $m$  and time-dependent frequency $\omega(t)$ under the influence of the external time-dependent force $f(t)$.
The corresponding propagator is
\[ \begin{split}
  \tilde{G}(t,x,t',x')=&\left[\frac{m}{2\pi i\hbar g(t,t')}\right]^{1/2}\exp{\left[\frac{i}{\hbar}S(t,x,t',x')\right]},
  \label{}
\end{split}  \]
where
\[ \begin{split}
  S(t,x,t',x')=\int_{t'}^t\left[\frac{1}{2}m\left(\frac{dy(s)}{ds}\right)^2-\frac{1}{2}m\omega^2(s) y^2(s)+f(s)y(s)\right]ds,
  \label{}
\end{split}  \]
$y(s)$ is the solution of the boundary value problem
\[ \begin{split}
  &\frac{d^2 y(s)}{ds^2}+\omega^2(s) y(s)=f(s),\\
  &y(t)=x,\\
  &y(t')=x',
  \label{}
\end{split}  \]
and $g(t,t')$ is the solution of the initial value problem
\[ \begin{split}
  &\frac{\partial^2 g(t,t')}{\partial t^2}+\omega^2(t) g(t,t')=0,\\
  &g(t',t')=0,\\
  &\left.\frac{\partial g(t,t')}{\partial t}\right\rvert_{t=t'}=1.
  \label{}
\end{split}  \]

The following four limiting cases of the above Hamiltonian are of particular interest.
\begin{itemize}
\item[(1)]
A free particle
\[
\begin{split}
  &H(t,x)=-\frac{\hbar^2}{2m}\frac{\partial^2}{\partial x^2},\\
  &g(t,t')=t-t',\\
  &S(t,x,t',x')=\frac{m(x-x')^2}{2(t-t')}.
  \label{}
\end{split}
\]
\item[(2)]
A particle in a uniform field
\[ \begin{split}
  &H(t,x)=-\frac{\hbar^2}{2m}\frac{\partial^2}{\partial x^2}-fx, \quad f=\textrm{const},\\
  &g(t,t')=t-t',\\
  &S(t,x,t',x')=\frac{m(x-x')^2}{2(t-t')}+\frac{1}{2}f(t-t')(x+x')-\frac{f^2(t-t')^3}{24m}.
  \label{}
\end{split}  \]
\item[(3)]
A harmonic oscillator
\[ \begin{split}
  &H(t,x)=-\frac{\hbar^2}{2m}\frac{\partial^2}{\partial x^2}+\frac{1}{2}m\omega^2 x^2, \quad \omega=\textrm{const},\\
  &g(t,t')=\frac{\sin{\omega(t-t')}}{\omega},\\
  &S(t,x,t',x')=\frac{m\omega}{2\sin{\omega(t-t')}}\left[(x^2+x^{\prime 2})\cos{\omega(t-t')}-2xx'\right].
  \label{}
\end{split}  \]
\item[(4)]
A driven harmonic oscillator
\[ \begin{split}
  &H(t,x)=-\frac{\hbar^2}{2m}\frac{\partial^2}{\partial x^2}+\frac{1}{2}m\omega^2 x^2+f(t)x, \quad \omega=\textrm{const},\\
  &g(t,t')=\frac{\sin{\omega(t-t')}}{\omega},\\
  &S(t,x,t',x')=\frac{m\omega}{2\sin{\omega(t-t')}}\biggl[(x^2+x^{\prime 2})\cos{\omega(t-t')}-2xx'\nn\\
    &+2x I(t,t')+2x' I(t',t)-2J(t,t')\biggr],
    \label{}
\end{split}
\]
where
$$
  I(t,t')=\frac{1}{m\omega}\int_{t'}^t f(s)\sin{\omega(s-t')}ds,
  $$
  and
  $$
  J(t,t')=\frac{1}{m^2\omega^2}\int_{t'}^t\int_{t'}^s f(s)f(s')\sin{\omega(t-s)}\sin{\omega(s'-t')}ds'ds.
  $$
\end{itemize}

We consider the case of a driven harmonic oscillator since the first three cases can be obtained as its appropriate limits and of course we have discussed in detail the first and the third cases in Sections \ref{sec5.1} and \ref{sec5.3}, respectively.
The computations are justified as in the previous section and for this reason we simply put the main points.
For the initial values
\[ \begin{split}
  &\psi(0,x)=e^{iapx/\hbar},\\
  &\psi_n(0,x)=\sum_{k=0}^n C_k(n,a)\exp{\left[\frac{ipx}{\hbar}\left(1-\frac{2k}{n}\right)\right]},
  \label{}
\end{split}
\]
we find
\[
\begin{split}
  &\psi(t,x)=(\cos{\omega t})^{-1/2}\exp{\biggl\{\frac{im\omega}{2\hbar\sin{\omega t}}\biggl[-\frac{1}{\cos{\omega t}}\biggl(x-\frac{ap\sin{\omega t}}{m\omega}-I(0,t)\biggr)^2}\nn
  \\
  &{+x^2\cos{\omega t}+2x I(t,0)-2J(t,0)\biggr]\biggr\}}, \label{}\nn
  \end{split}
\]
and
\[
\begin{split}
  \psi_n(t,x)&=(\cos{\omega t})^{-1/2}\exp{\biggl\{\frac{im\omega}{2\hbar\sin{\omega t}\cos{\omega t}}\Bigl[-x^2\sin^2{\omega t}+2xI(t,0)\cos{\omega t}}\nn
  \\
   &{-2J(t,0)\cos{\omega t}+2xI(0,t)-I^2(0,t)\Bigr]\biggr\}}
  \\
  &
  \times \sum_{m=0}^\infty\frac{1}{m!}\left(\frac{i\hbar}{2m\omega}\sin{\omega t}\cos{\omega t}\right)^m\nn\frac{\partial^{2m}}{\partial x^{2m}}\psi_n\left(0,\frac{x-I(0,t)}{\cos{\omega t}}\right).
\end{split}
\]
Using
\[ \begin{split}
  \lim_{n\to\infty}\psi_n\left(0,\frac{x-I(0,t)}{\cos{\omega t}}\right)=\exp\left[\frac{iap(x-I(0,t))}{\hbar\cos{\omega t}}\right],
  \label{}
\end{split}
\]
we obtain
\[ \begin{split}
  &\lim_{n\to\infty}\psi_n(t,x)=(\cos{\omega t})^{-1/2}\exp{\biggl\{\frac{im\omega}{2\hbar\sin{\omega t}}\biggl[-\frac{1}{\cos{\omega t}}\biggl(x-\frac{ap\sin{\omega t}}{m\omega}-I(0,t)\biggr)^2}\nn\\ &{+x^2\cos{\omega t}+2x I(t,0)-2J(t,0)\biggr]\biggr\}}=\psi(t,x).\label{}
\end{split}
\]
The amplitude and frequency of $\psi(t,x)$ simultaneously diverge for $\omega t=(2k+1)(\pi/2)$, $k\in\Z$, while
the frequency of oscillations of $\psi(t,x)$ in $x$ increases with $|x|$ without bound for any $a$.
This is rather a consequence of the Hamiltonian of the harmonic oscillator generating the time evolution of a wave function with the infinite norm.
\bigskip

\noindent{\em Persistence of superoscillations}

\bigskip

We now introduce another large class of superoscillating functions and we show that for this class the superoscillatory behavior persists in time, and in fact cannot arise if the initial datum is not superoscillating.
\begin{theorem}
   Consider a sequence of functions $\{\psi_n(t,x)\}_{n=1}^\infty$ of the form
\begin{equation} \begin{split}
  \quad \psi_n(t,x)=\sum_{k=-n}^n c_{n,k}(t)\exp{\left(\frac{ikpx}{n\hbar}\right)},
  \end{split}   \label{psi_n}
\end{equation}
with suitable coefficients $c_{n,k}$. Suppose that each $\psi_n(t,x)$ satisfies the time-dependent Schr\"{o}dinger equation
\begin{equation} \begin{split}
  \left[i\hbar\frac{\partial}{\partial t}+\frac{\hbar^2}{2m}\frac{\partial^2}{\partial x^2}-V_n(t,x)\right]\psi_n(t,x)=0
  \label{schrodinger1}
\end{split}
\end{equation}
for a certain potential energy $V_n(t,x)$. Then the sequence $\{\psi_n(t,x)\}_{n=1}^\infty$ is superoscillatory at any given time if and only if it is superoscillatory at any other time.
  \label{}
\end{theorem}
\begin{proof}
Substituting \eqref{psi_n} into \eqref{schrodinger1}, we find
\begin{equation} \begin{split}
  \sum_{k=-n}^n\left[i\hbar\frac{\partial c_{n,k}(t)}{\partial t}-\frac{k^2p^2}{2mn^2}c_{n,k}(t)-V_n(t,x)c_{n,k}(t)\right]\exp{\left(\frac{ikpx}{n\hbar}\right)}=0.
  \label{psi_n_expansion}
\end{split}
\end{equation}
Although the set of functions $$\{\exp{[ikpx/(n\hbar)]}\}_{k=-n}^n$$ is not complete on $[-n\pi\hbar/p,n\pi\hbar/p]$ for any finite $n$, we have that equation \eqref{psi_n_expansion} holds for all $x$ if and only if the expression in the square brackets is identically zero.
Furthermore, since $t$ is arbitrary in this expression, this implies that $V_n(t,x)$ does not depend on $x$.
After setting $V_n(t,x)=V_n(t)$, the resulting differential equation has the solution
\[ \begin{split}
  c_{n,k}(t)=c_{n,k}(t')\exp{\left[-\frac{ik^2p^2(t-t')}{2m\hbar n^2}+\frac{i}{\hbar}\int_{t'}^t V_n(s)ds\right]},
  \label{}
\end{split}  \]
which leads to
\begin{equation} \begin{split}
  \psi_n(t,x)=\exp{\left[\frac{i}{\hbar}\int_{t'}^t V_n(s)ds\right]}\sum_{k=-n}^n c_{n,k}(t')\exp{\left[-\frac{ik^2p^2(t-t')}{2m\hbar n^2}+\frac{ikpx}{n\hbar}\right]}.
  \label{psi_n_sol}
\end{split}  \end{equation}
Setting $t=t'$ in equation \eqref{psi_n_sol}, multiplying the result by $\exp{[-ilpx/(n\hbar)]}$, where $-n\le l\le n$, and integrating over $x\in[-n\pi\hbar/p,n\pi\hbar/p]$, we find
\[ \begin{split}
  c_{n,l}(t')=\frac{p}{2\pi n\hbar}\int_{-n\pi\hbar/p}^{n\pi\hbar/p}\psi_n(t',x)\exp{\left(-\frac{ilpx}{n\hbar}\right)}dx,
  \label{}
\end{split}  \]
which leads to
\begin{equation} \begin{split}
  \psi_n(t,x)=&\frac{p}{2\pi n\hbar}\exp{\left[\frac{i}{\hbar}\int_{t'}^t V_n(s)ds\right]}\times \nn\\ &\int_{-n\pi\hbar/p}^{n\pi\hbar/p}\psi_n(t',x')\sum_{k=-n}^n \exp{\left[-\frac{ik^2p^2(t-t')}{2m\hbar n^2}+\frac{ikp(x-x')}{n\hbar}\right]}dx'.
  \label{psi_n_solution}
\end{split}  \end{equation}
Setting
\begin{equation} \begin{split}
  &\psi(t,x)=\lim_{n\to\infty}\psi_n(t,x),\label{psi}\\
  &V(t)=\lim_{n\to\infty}V_n(t)
\end{split}  \end{equation}
and taking the limit $n\to\infty$ in \eqref{psi_n_solution}, we find
\[ \begin{split}
  \psi(t,x)=&\frac{p}{2\pi\hbar}\exp{\left[\frac{i}{\hbar}\int_{t'}^t V(s)ds\right]}\times\nn\\ &\int_{-\infty}^\infty\psi(t',x')\lim_{n\to\infty}\frac{1}{n}\int_{-\infty}^\infty \exp{\left[-\frac{ik^2p^2(t-t')}{2m\hbar n^2}+\frac{ikp(x-x')}{n\hbar}\right]}dkdx',
  \label{}
\end{split}  \]
which leads to
\begin{equation} \begin{split}
  &\psi(t,x)=\\
  &=\left[2\pi i\hbar(t-t')\right]^{-1/2}\exp{\left[\frac{i}{\hbar}\int_{t'}^t V(s)ds\right]}\int_{-\infty}^\infty\psi(t',x')\exp{\left[\frac{im(x-x')^2}{2(t-t')}\right]}dx'.
  \label{psi_solution}
\end{split}  \end{equation}

Suppose that, for a fixed $t'$, the function $\psi(t',x)$ is periodic in $x$ with the period $X$,
namely
\begin{equation} \begin{split}
  \psi(t',x+X)=\psi(t',x).
  \label{psi_periodicity}
\end{split}  \end{equation}
It is straightforward to prove that, at any other time $t$, the function $\psi(t,x)$ is also periodic in $x$ with the same period $X$.
Indeed, from \eqref{psi_solution},
\[ \begin{split}
  \psi(t,x+X)&=\left[2\pi i\hbar(t-t')\right]^{-1/2}\exp{\left[\frac{i}{\hbar}\int_{t'}^t V(s)ds\right]}
  \\
  &
  \times
  \int_{-\infty}^\infty\psi(t',x')\exp{\left[\frac{im(x+X-x')^2}{2(t-t')}\right]}dx',
  \label{}
\end{split}  \]
which, after the change of variable $x'\mapsto x'+X$ and use of  \eqref{psi_periodicity}, gives
\[ \begin{split}
  \psi(t,x+X)=\psi(t,x).
  \label{}
\end{split}  \]
Applying this result to the case $X<2\pi\hbar/p$, we obtain the statement.
\end{proof}

\begin{remark}{\rm
When we take as superoscillating initial datum $F_n(x,a)$, as we did in Chapter 5, the limit solution in the case of free particle is
$$
\lim_{n\to\infty} \psi_n^{\rm free}(x,t)=e^{iax-ia^2t/2},
$$
while in the case of the harmonic oscillator is
$$
\lim_{n\to \infty}\psi_n^{\rm ho}(t,x)=(\cos t)^{-1/2} \exp({ -(i/2) (x^2+a^2)\tan t +iax/\cos t}).
$$
 For $|x|\ll 1$ and $0\leq t \ll 1$, we have
  $$(\cos t)^{-1/2}\approx 1,\qquad \tan t\approx t, \qquad x^2\approx 0,
  $$
  so
  $$
  (\cos t)^{-1/2} \exp ( -(i/2)(x^2+a^2)\tan t +iax/\cos t)\approx \exp({iax-ia^2t/2}),
  $$
and a similar result holds by taking a more general superoscillating sequence $Y_n(x,a)$ as initial datum.}
\end{remark}


\chapter{Superoscillating functions and operators}\label{sec4}

In this chapter, we consider sequences of operators $F_n(T,a)$ where $T$ is a self-adjoint operator on a Hilbert space $\mathcal{H}$, and $F_n$ is a superoscillating sequence of the form \eqref{eq1}.
 For any fixed $n\in \mathbb{N}$ and for
  any  self-adjoint operator $T$ in $\mathcal{H}$ we can define,
by the spectral theorem,  the  sequence of operators
\begin{equation}\label{yydgfyf}
F_n(T,a)=\int_{\sigma(T)}F_n(\lambda,a) \ E(d\lambda),
\end{equation}
where $E(d\lambda)$ is the spectral measure associated with $T$.\\
From the quantum mechanics point of view, the most interesting case is when $\sigma(T)=\mathbb{R}$ and $T$ is the momentum operator.
When considering the sequence of functions $F_n$, for every fixed $n\in \mathbb{N}$, we have
$$
\sup_{x\in \mathbb{R}}|F_n(x,a)|=a^n,
$$
so the operators $F_n(T,a)$, for any fixed $n$, turn out to be bounded operators on $\mathcal{H}$ and their norm is $\|F_n(T,a)\|=a^n$.
However, when we take the limit for $n\to \infty$ the norm $\|F_n(T,a)\|$ goes to infinity. This is the reflection of the fact that
on the compact sets of $\mathbb{R}$, the sequence of functions $F_n(x,a)$ converges to $F(x,a)=e^{iax}$ for $x\in\mathbb{R}$,
but for $a>1$ we have
$$
\sup_{n\in \mathbb{N},\ x\in \mathbb{R} }|F_n(x,a)|=\infty.
$$
Thus, this is the difficulty  when considering the limit for $n\to \infty$ of the sequence of the operators $F_n(T,a)$.
\\
The problems we want to solve can be formulated as follows.
\\
 For $a>1$ we consider the sequence of functions $F_n(x,a)$ and  a self-adjoint operator $T$ on a Hilbert space $\mathcal H$.
\begin{itemize}
 \item A first problem is to define the sequence of operators $F_n(T,a)$ when the spectrum of $T$ is limited to a superoscillation set of $F_n$ and compute the limit for $n\to \infty$.
\item A second problem is to consider a suitable variation of the definition $F_n(T,a)$, when the spectrum of $T$ is unbounded, in order for this new sequence to converge to a bounded operator.
\end{itemize}
We will show how to solve both problems using the von Neumann spectral theorem for unbounded self-adjoint operators on a Hilbert space.

\section{A quick review on operators}
We denote by $\mathcal{D}(T)$ the domain of the linear operator $T$ and we will always assume that $\mathcal{D}(T)$ is dense in the Hilbert
space $\mathcal{H}$. We equip $\mathcal H$ with the scalar product $(\cdot,\cdot)$ and we denote by $\|\cdot\|$ the norm in $\mathcal{H}$. By $E$ we will denote the countably additive self-adjoint spectral measure defined on the Borel sets of the complex plane.
\\ We now recall
 the spectral theorem for self-adjoint operators. For more details, we refer the reader to \cite[Theorem 3, p. 1192]{dsII}.

\begin{theorem}
Let $T$ be a self-adjoint operator on the Hilbert space $\mathcal{H}$. Then its spectrum is real and there exists a unique
countably additive self-adjoint spectral measure $E$ defined on the Borel sets of the complex plane, vanishing on the complement
of the spectrum $\sigma(T)$ of $T$. The relation between the operator $T$ and the spectral measure $E$ is given by
$$
\mathcal{D}(T)=\{ \psi \in \mathcal{H}\ :\ \int_{\sigma(T)} |\lambda |^2\, (E(d\lambda) \psi,\psi)<\infty\},
$$
$$
T\psi=\lim_{n\to \infty}\int_{-n}^n\lambda E(d\lambda)\psi,\ \ \ \ \ \ {\it for \ all\  }\ \  \psi\in \mathcal{D}(T).
$$
\end{theorem}
The unique spectral measure (appearing in the spectral theorem) associated with the self-adjoint operator $T$ is called resolution of the identity.
For every bounded Borel function $f$ defined on the spectrum of $\sigma(T)$ of a self-adjoint operator $T$,
we may define the bounded normal operator $f(T)$ by the equation
$$
f(T)=\int_{\sigma(T)} f(\lambda)\,  E(d\lambda).
$$

Using the spectral theorem we give the following definition which holds also for Borel functions $f$ that are not bounded.
\begin{definition}\label{good}
Let $E$ be the resolution of the identity for the self-adjoint operator $T$ and let $f$ be a complex Borel function defined almost everywhere
on the real axis. Consider the sequence of functions $f_n$ defined by
$$
f_n(\lambda)=\left\{\begin{array}{l}
\displaystyle f(\lambda),   \quad{\rm if}\ |f(\lambda)|\leq n,\\
0,\ \ \  \ \  \quad{\rm if}\ |f(\lambda)| > n.\\
\end{array}
\right.
$$
We define $f(T)$ as follows:
$$
\mathcal{D}(f(T))=\{\psi \in \mathcal{H}\ :\ \lim_{n\to \infty}f_n(T)\psi \ \ \ {\it exists}\ \},
$$
$$
f(T)\psi :=\lim_{n\to \infty}f_n(T)\psi,\ \ \ \ {\it for \ all\ }\ \  \psi\in \mathcal{D}(f(T)).
$$

\end{definition}
From the definition it is evident that if $f(\lambda)=\lambda$ we have $f(T)=T$ but it is not clear what happens in general, for example when
$f(\lambda)$ is a polynomial. Next theorem shows that the Definition \ref{good} is a good definition (see \cite[Theorem 6, p. 1196]{dsII}).

\begin{theorem}
Let $E$ be the resolution of the identity for the self-adjoint operator $T$ and let $f$ be a complex Borel function defined almost everywhere
on the real axis. Then $f(T)$ is a closed operator with dense domain, moreover
$$
\mathcal{D}(f(T))=\{\psi \in \mathcal{H}\ : \ \int_{\sigma(T)}|f(\lambda) |^2\, (E(d\lambda) \psi,\psi)<\infty \},
$$
$$
(f(T)\psi, \phi)=\int_{\sigma(T)}f(\lambda)\, (E(d\lambda) \psi,\phi),\ \ \ \psi\in \mathcal{D}(f(T)),\ \ \phi\in \mathcal{H},
$$
and
$$
\|f(T)\psi\|^2=\int_{\sigma(T)}|f(\lambda) |^2\, (E(d\lambda) \psi,\psi),\ \ \ \psi\in \mathcal{D}(f(T)).
$$
\end{theorem}
In the sequel, the following result  will be important (see \cite[Theorem 9, p. 1200]{dsII}):
\begin{theorem}\label{nornafdit}
Let $T$ be a self-adjoint operator and let $f$ be
 a complex Borel function defined almost everywhere
on the real axis. Then we have
$$
\|f(T)\|={\rm ess}\sup_{\lambda\in \sigma(T)}|f(\lambda)|.
$$
\end{theorem}

\section{Superoscillations and operators}
Using the spectral theorem, we can now give the precise definition of the operators we will study. Note that, in this section, it is not necessary to consider $a>1$, so $a$ will vary in $\mathbb R$, unless otherwise specified.
\begin{definition}
 Let $T$ be a self-adjoint operator in a Hilbert space $\mathcal{H}$, let $F_n$  be the function defined in (\ref{ef}) and set $F(x,a)=e^{iax}$.
Using the spectral theorem we define the  operators
\begin{equation}\label{deffN}
\begin{split}
&F_n(T,a)=\int_{\sigma(T)}F_n(\lambda,a) \ E(d\lambda),\\
&F(T,a)=\int_{\sigma(T)}F(\lambda,a) \ E(d\lambda).
\end{split}
\end{equation}
\end{definition}
\begin{proposition}
 Let $T$ be a self-adjoint operator in a Hilbert space $\mathcal{H}$.
 Then for every fixed $n\in \mathbb{N}$ and $a\in \mathbb{R}$  the operators $F_n(T,a)$ defined in
 (\ref{deffN}) are bounded operators on $\mathcal{H}$.
\end{proposition}
\begin{proof}
It follows from the spectral theorem and from the fact that the functions $F_n(\lambda,a)$ and $F(\lambda,a)$ are bounded functions on $\mathbb{R}$ for every fixed $n\in \mathbb{N}$ and $a \in\mathbb R$.
\end{proof}

\noindent
\begin{proposition} Let $T$ be a self-adjoint operator on a Hilbert space $\mathcal{H}$. Then
$F(T,a)$ defined in (\ref{deffN})
is a strongly continuous one-parameter group with respect to $a\in \mathbb{R}$, that is, it satisfies the relations
$$
F(T,a+b)=F(T,a)F(T,b),\ \ \ \ {\it for \ all} \ \ a,b\in \mathbb{R},
$$
and
$$
\lim_{a\to a_0}F(T,a)\psi=F(T,a_0)\psi \ \ \ {\it if}\ \ \  \psi\in \mathbb{H}.
$$
\end{proposition}
\begin{proof}
It follows by standard arguments from the spectral theorem.
\end{proof}
The following proposition shows  that  the problem of the convergence of the sequence of operators $F_n(T,a)$ to the operator $F(T,a)$
is a delicate issue.

\begin{proposition}\label{norma}
Let $T$ be a self-adjoint operator on a Hilbert space $\mathcal{H}$ and let
$F_n(T,a)$ be the sequence defined in (\ref{deffN}). Then we have:
\begin{itemize}
\item[(1)]
For any fixed $n\in \mathbb{N}$  the family of operators
$F_n(T,a)$
is not a group with respect to the parameter $a\in \mathbb{R}$.
\item[(2)]
For every $n\in \mathbb{N}$, $a\in \mathbb{R}$, and for every arbitrary compact set $K \subset \sigma(T)$ we have
\begin{equation}\label{norFn}
\|F_n(T,a)\|= \sup_{\lambda \in K \subset \sigma(T)}
\Big(\cos^2 \Big(\frac{\lambda}{n}\Big)+a^2\sin^2 \Big(\frac{\lambda}{n}\Big)\Big)^{n/2},
\end{equation}
and
\begin{equation}\label{norlim}
\lim_{n\to \infty}\|F_n(T,a)\|=1.
\end{equation}
\item[(3)]
In the case $\sigma(T)=\mathbb{R}$, $a\in \mathbb{R}$, we have
\begin{equation}\label{norFnII}
\|F_n(T,a)\|=a^n.
\end{equation}
\end{itemize}
\end{proposition}
\begin{proof}
To show that
the property $F_n(\lambda,a+b)=F_n(\lambda,a)F_n(\lambda,b)$ is not satisfied, we write
$$
F_n(\lambda, a)=
\Big(\cos\Big(\frac{\lambda}{n}\Big)+ia\sin \Big(\frac{\lambda}{n}\Big) \Big)^n,
$$
then, it is immediate that
$$
\cos\Big(\frac{\lambda}{n}\Big)+i(a+b)\sin \Big(\frac{\lambda}{n}\Big) \not=
\Big(\cos\Big(\frac{\lambda}{n}\Big)+ia\sin \Big(\frac{\lambda}{n}\Big) \Big)
\Big(\cos\Big(\frac{\lambda}{n}\Big)+ib\sin \Big(\frac{\lambda}{n}\Big) \Big)
$$
so
$F_n(\lambda,a+b)\not=F_n(\lambda,a)F_n(\lambda,b)$. Using the spectral theorem, $(1)$ follows.\\
The relation (\ref{norFn}), stated in point $(2)$, follows from Theorem \ref{nornafdit}:
$$
\|f(T)\|=  \sup_{\lambda \in K \subset \sigma(T)} |F_n(\lambda,a)|
$$
and the fact that
$$
|F_n(\lambda,a)| =\Big(\cos^2 \Big(\frac{\lambda}{n}\Big)+a^2\sin^2 \Big(\frac{\lambda}{n}\Big)\Big)^{n/2}.
$$
Simple computations show that
$$
\lim_{n\to \infty}\Big(\cos^2 \Big(\frac{\lambda}{n}\Big)+a^2\sin^2 \Big(\frac{\lambda}{n}\Big)\Big)^{n}=1
$$
so (\ref{norlim}) holds.
Finally, relation (\ref{norFnII}) in point $(3)$ follows from the fact that
$$
\sup_{\lambda\in \mathbb{R}}|F_n(\lambda,a)| =a^n.
$$
\end{proof}
 Proposition \ref{norma} implies that, when $a>1$, we have
$$
\sup_{\lambda\in \mathbb{R}}|F_n(\lambda,a)| =a^n\to \infty\ \ {\rm as}\ \ n\to \infty.
$$
This is due to the fact that for fixed $\lambda\in\mathbb R$ the sequence
$$
|F_n(\lambda,a)|=\Big(\cos^2 \Big(\frac{\lambda}{n}\Big)+a^2\sin^2 \Big(\frac{\lambda}{n}\Big)\Big)^{n/2}
$$
converges to 1, but if we take $\lambda=n\pi/2$ then
$$
|F_n(n\pi/2,a)|=a^{n}.
$$
So we consider the case in which the integral, in the spectral theorem, is computed on  compact sets in $\mathbb{R}\cap \sigma(T)$
 and then we consider the case in which the spectrum in unbounded.

\newpage
{\em The case of compact sets.}

\bigskip
In this case instead of the operators  $F_n(T,a)$ and $F(T)$ defined by (\ref{deffN}), we consider two modified operators:
\begin{definition} Let $T$ be a self-adjoint operator on a Hilbert space $\mathcal{H}$ and let $K$ is a compact set in $\mathbb R$. Then we define
\begin{equation}\label{deffN1}
P_n(T,a)=\int_{\sigma(T) \cap K}F_n(\lambda,a) \ E(d\lambda),\ \ \ \ \ \ \ \ \
P(T,a)=\int_{\sigma(T)\cap K}F(\lambda,a) \ E(d\lambda).
\end{equation}
\end{definition}
The operators $P_n(T,a)$ and $P(T,a)$ of course depend on the set $K$.
In the case in which the operator $T$ has bounded spectrum then, if we choose $K$ such that $\sigma(T) \subset K$, then
$$
F_n(T,a)=P_n(T,a)
$$
and
$$
F(T,a)=P(T,a).
$$
\begin{theorem}
Let $T$ be a self-adjoint operator on a Hilbert space $\mathcal{H}$, and let $P_n(T,a)$ and $P(T,a)$ be the operators defined in (\ref{deffN1}) for $K$ a sufficiently large compact set in $\mathbb R$. Then we have
$$
\|[P_n(T,a)-P(T,a)]\psi\|^2\to 0,\ \ for \ all\ \ \psi\in \mathcal{H}.
$$
\end{theorem}
\begin{proof}
We recall that
$$
|F_n(\lambda,a) -F(\lambda,a)|^2=
1+\Big(\cos^2 \Big(\frac{\lambda}{n}\Big)+a^2\sin^2 \Big(\frac{\lambda}{n}\Big)\Big)^{n}
$$
$$
-2\Big(\cos^2 \Big(\frac{\lambda}{n}\Big)+a^2\sin^2 \Big(\frac{\lambda}{n}\Big)\Big)^{n/2} \cos\Big[n\arctan \Big(a \tan \Big(\frac{\lambda}{n}\Big)\Big) -a \lambda\Big].
$$
From the spectral theorem, for  $\psi\in \mathcal{D}(P(T,a))\cap\mathcal{D}(P_n(T,a))$, we have that
$$
\|[P_n(T,a)-P(T,a)]\psi\|^2=\int_{\sigma(T)\cap K}|F_n(\lambda,a) -F(\lambda,a)|^2\, d(E(d\lambda) \psi,\psi)
$$
 We know that $|F_n(\lambda,a) -F(\lambda,a)|^2\to 0$ as $n\to \infty$ on the compact sets $K$ of $\mathbb{R}$, where the superoscillation phenomenon described in  Theorem \ref{carnot} holds.
Since $K$ can be chosen arbitrarily large, we have
 $$
 |F_n(\lambda,a) -F(\lambda,a )|^2\leq G(\lambda)
 $$
 where
 $$
 G(\lambda):=\sup_{\lambda \in K} |F_n(\lambda,a ) -F(\lambda,a)|^2;
 $$
 note that for $n$ sufficiently large $ G(\lambda)$ does not depend on $n$ and it is integrable.
 The dominated convergence theorem implies that
$$\|[P_n(T,a)-P(T,a)]\psi\|^2\to 0.
$$
Since $F_n(\lambda,a)$ and $F(\lambda,a)$ are bounded functions on any compact set of $\mathbb{R}$ and the spectrum of $T$ is compact
we have $\mathcal{D}(P_n(T))\cap\mathcal{D}(P(T))=\mathcal{H}$.
\end{proof}

\bigskip\noindent
{\em The case in which $\sigma(T)$ is unbounded.}

\bigskip

We now study the case in which the whole unbounded spectrum is taken into account. The previous results
suggest the following definition.
\begin{definition}
Let $T$ be a self-adjoint operator on a Hilbert space $\mathcal{H}$ and  assume that $\sigma(T)\subseteq\mathbb{R}$. Let $\gamma >2$ be a fixed number and define the operators $Q_n(T,\gamma)$ as
\begin{equation}\label{qntg}
Q_n(T,a,\gamma):=\int_{e_n(\gamma)}F_n(\lambda,a) \ E(d\lambda)
\end{equation}
where
\begin{equation}\label{qntgeg}
e_n(\gamma):=\Big[-\frac{n^{1/\gamma}\pi}{\gamma},\ \frac{n^{1/\gamma}\pi}{\gamma}\Big].
\end{equation}
\end{definition}

Observe that $e_n(\gamma)\subset[-\dfrac{n\pi}{2},\ \dfrac{n\pi}{2}]$, for all  $\gamma >2$ and $n\in\mathbb N$. This means that when $n$ goes to infinity
the sequence of operators $Q_n(T,\gamma)$ remains bounded as it is proved in the next result.

\begin{theorem}\label{inddagam}
Let $T$ be a self-adjoint operator on a Hilbert space $\mathcal{H}$ such that $\sigma(T)=\mathbb{R}$. Let $Q_n(T,\gamma)$
be the family of operators defined in (\ref{qntg}).
Then, for every $\gamma >2$, we have
$$
\lim_{n\to \infty}\|Q_n(T,a,\gamma)\|=1.
$$
\end{theorem}
\begin{proof}
For any fixed $n\in\mathbb N$, we have
\[
\begin{split}
\|Q_n(T,a,\gamma)\|&=\sup_{\lambda\in \sigma(T)}|F_n(\lambda,a)\chi_{e_n(\gamma)}|=\sup_{\lambda\in e_n(\gamma)}|F_n(\lambda,a)|
\\
&=\sup_{\lambda\in e_n(\gamma)}\Big(\cos^2 \Big(\frac{\lambda}{n}\Big)+a^2\sin^2 \Big(\frac{\lambda}{n}\Big)\Big)^{n/2}=1,
\end{split}
\]
in fact $\dfrac{\lambda}{n}$ cannot be $\pm \frac{\pi}{2}$, thus
$
\sup_{\lambda\in e_n(\gamma)}|F_n(\lambda,a)|
$
is attained in a point $\lambda_0$ which cannot be $\frac{n\pi}{2}$. So we conclude that
$$
\lim_{n\to \infty}\|Q_n(T,a,\gamma)\|=1
$$
for every $\gamma >2$.
\end{proof}
 Inspired by Definition \ref{good} and by the fact that $\lim_{n\to \infty}\|Q_n(T,\gamma)\|=1$ is independent of $\gamma$, when is $\gamma>2$, we give the following definition.
\begin{definition}\label{cekpo}
Let $T$ be a self-adjoint operator on a Hilbert space $\mathcal{H}$ such that  $\sigma(T)=\mathbb{R}$, and let $\gamma>2$. We set
 \begin{equation}\label{defQN4}
Q_n(T,a):=Q_n(T,a,\gamma),\ \ \ \ \ \ \ \ \ e_n:=e_n(\gamma),
\end{equation}
and we define
$$
\mathcal{D}(Q(T,a)):=\{\psi \in \mathcal{H}\ :\ \lim_{n\to \infty}Q_n(T,a)\psi \ \ \ {\it exists}\ \},
$$
$$
Q(T,a)\psi:=\lim_{n\to \infty}Q_n(T,a) \ \psi, \ \ \ \ \ \psi\in\mathcal{D}(Q(T,a)).
$$
\end{definition}

We can now prove the following:
\begin{theorem}
Let $T$ be a self-adjoint operator on a Hilbert space $\mathcal{H}$, with $\sigma(T)=\mathbb{R}$ and let $Q_n(T,a)$ and $F(T,a)$ be the operators defined in (\ref{defQN4}) and (\ref{deffN}), respectively. Then we have
$$
\|[Q_n(T,a)-F(T,a)]\psi\|^2\to 0,\ \ for \ all\ \ \psi\in \mathcal{H}.
$$
\end{theorem}
\begin{proof} We have
$$
\sup_{\lambda\in e_n}|F_n(\lambda,a) -F(\lambda,a)|^2=
\sup_{\lambda\in e_n} \Big(1+\Big(\cos^2 \Big(\frac{\lambda}{n}\Big)+a^2\sin^2 \Big(\frac{\lambda}{n}\Big)\Big)^{n}
$$
$$
-2\Big(\cos^2 \Big(\frac{\lambda}{n}\Big)+a^2\sin^2 \Big(\frac{\lambda}{n}\Big)\Big)^{n/2} \cos\Big[n\arctan \Big(a \tan \Big(\frac{\lambda}{n}\Big)\Big) -a\lambda\Big]\Big),
$$
and since $\lambda\in e_n$ we also have that
$$
\sup_{\lambda\in e_n}|F_n(\lambda,a) -F(\lambda,a)|^2\to 0\ \ \ {\rm as }\ \ n\to\infty.
$$
We now define the operator
$$
\breve{F}_n(T,a,\gamma )=\int_{ e_n(\gamma)}F(\lambda,a) \ E(d\lambda),
$$
and we observe that, for any $n\in\mathbb N$ fixed
$$
\mathcal{D}(\breve{F}_n(T,a))\cap\mathcal{D}(Q_n(T,a))=\mathcal{H}.
$$
The spectral theorem gives
\[
\begin{split}
&\lim_{n\to\infty}\|[Q_n(T,a)-\breve{F}_n(T,a, \gamma)]\psi\|^2\\
&=\lim_{n\to\infty}\int_{e_n(\gamma)}|F_n(\lambda,a) -F(\lambda,a)|^2\, (E(d\lambda) \psi,\psi),\ \ \ \psi\in \mathcal{H}.
\end{split}
\]
By the dominated convergence theorem we obtain
$$\|[Q_n(T,a)-\breve{F}_n(T,a,\gamma)]\psi\|^2\to 0.$$
Moreover, for $\gamma >2$, we have
$$\|[F(T,a)-\breve{F}_n(T,a,\gamma)]\psi\|^2\to 0,$$
and so by
\[
\begin{split}
\|[Q_n(T,a)-F(T,a)]\psi\|&\leq \|[Q_n(T,a)-\breve{F}_n(T,a,\gamma)]\psi\|
\\
&
+\|[\breve{F}_n(T,a,\gamma)-F(T,a)]\psi\|
\end{split}
\]
the statement follows.

\end{proof}

\begin{remark}{\rm
Our discussion shows that when $n$ is an arbitrary large natural number and $T$ is any self-adjoint operator with unbounded spectrum,
then the operators $F_n(T,a)$ are bounded operators for any $a >1$.
\\
When $n\to \infty $,  and $T$ is any self-adjoint operator with unbounded spectrum
then, if we consider that part of the spectrum that is contained in some compact set, then
 we have $\lim_{n\to \infty}P_n(T,a)=P(T,a)$ for any $a >1 $ and $ P(T,a)$ is a bounded operator which depends on the parameter $a$.
\\
Since we have
$$
\sup_{x\in \mathbb{R}}|F_n(x,a)|=a^n,
$$
in  the case we consider  any self-adjoint operator $T$ with unbounded spectrum,  the family of operators
$Q_n(T,a)$ is such that $\lim_{n\to \infty}Q_n(T,a)$  converges to the bounded linear operator $F(T,a)$.
}
\end{remark}
\begin{remark}
{\rm
If $\hat{P}$ is the momentum operator $-i D_x =\hat{P}$ we have
\begin{equation}\label{fnp}
\begin{split}
&Q_n(\hat{P})=\Big(\frac{a+1}{2}\exp({\frac{i}{n}L\hat{P}})+\frac{1-a}{2}\exp({-\frac{i}{n}L\hat{P}})\Big)^n
\\
&F(\hat{P})=\exp({i aL\hat{P}}),
\end{split}
\end{equation}
and
$$
\lim_{n\to \infty} (Q_n(\hat{P})\phi,\psi)=(F(\hat{P})\phi,\psi),\ \ \ {\rm  for \ all}\ \ \phi\in \mathcal{D}(\hat{P}),\ \ \psi\in \mathcal{H}.
$$
We finally recall that the problem of approximating a function by superoscillating sequences arises from the fact
that  (\ref{fnp}) can be written as
$$
F_n(\hat{P},a)=\sum_{j=0}^nC_j(n,a)\exp({i(1-2j/n)L\hat{P}}).
$$
The action of $F_n(\hat{P},a)$ on functions $\psi \in \mathcal{S}(\mathbb{R})$ gives
a linear combination of the function $\psi$ computed at the points
 $x+(1-2j/n)L$ and yields:
$$
\phi_n(x):=F_n(\hat{P},a)\psi(x)=\sum_{j=0}^nC_j(n,a)\psi(x+(1-2j/n)L).
$$
}
\end{remark}

\chapter{Superoscillations in $SO(3)$}

In this chapter, we expand on our earlier paper \cite{ACSST4} to show how superoscillations may occur in the framework of group representation theory. Specifically, we will analyze the case of the group $SO(3)$, though our results should extend to $SO(n)$ for any integer $n$. In our original \cite{ACSST4}, we considered the situation for $2\ell$ particles with equal spin, while here we will show how this restriction is not necessary, and we will describe the necessary modification to consider the general case in which we have $\ell+m$ particles with spin $\frac 12$ and $\ell-m$ particles with spin $-\frac 12$. The main byproduct of superoscillations in this context is a new asymptotic relationship involving the well known Wigner $d$-functions (see \cite{dennis2004} as well as \cite{sakurai}).

\section{The weak value of the operator  $\exp(i\hat{\mathcal{L}}_z\upn\delta\varphi)$}

Before to treat the topic systematically, we want to offer a heuristic argument which suggests the existence of superoscillations in $SO(3)$. This argument comes from the same general reasoning that is applied when one computes the weak value of operators in the time-symmetric formulation of quantum mechanics  \cite{PT-nov-2010} and that is discussed in Chapter2.

Let $\hat{L}_x$, $\hat{L}_y$, $\hat{L}_z$ be the quantum operators associated with  the components of the angular momentum. These three operators can be seen as generators of $SO(3)$.
Let $\ket{\psi_{\rm in}}$ and $\ket{\psi_{\rm fin}}$ be the states defined by
\begin{equation}
\hat{A}|\psi_{\rm fin}\rangle:=(\hat{L}_z\cos\theta-\hat{L}_x\sin\theta) | \psi_{\rm fin}\rangle = \ell |\psi_{\rm fin}\rangle
\label{A}
\end{equation}
\begin{equation}
\hat{B}|\psi_{\rm in}\rangle:=(\hat{L}_z\cos\theta+\hat{L}_x\sin\theta) | \psi_{\rm in}\rangle = \ell |\psi_{\rm in}\rangle
\label{B}
\end{equation}
where $\ell=\max|\hat{L}_z|$ is related to the total angular momentum and can be computed from the relation $$\hat{L}^2|\psi_{\rm in}\rangle=\ell(\ell+1)|\psi_{\rm in}\rangle$$
or, analogously, from $\hat{L}^2|\psi_{\rm fin}\rangle=\ell(\ell+1)|\psi_{\rm fin}\rangle$.
Since we have
\begin{equation}
{\hat{L}_z}=\frac{\hat{L}_z\cos\theta -\hat{L}_x\sin\theta+\hat{L}_z\cos\theta+\hat{L}_x\sin\theta}{2\cos\theta}=\frac{\hat{A}+\hat{B}}{2\cos\theta},
\label{lz}
\end{equation}
we can compute the weak value of $\la \hat{L}_z\ra_w$:
\begin{eqnarray}
\la \hat{L}_z\ra_w&= & \frac{\langle \psi_{\rm fin}|\hat{L}_z|\psi_{\rm in}\rangle}{\langle \psi_{\rm fin}|\psi_{\rm in}\rangle}=\frac{1}{2\cos\theta}\frac{\langle \psi_{\rm fin}|\hat{A}+\hat{B}|\psi_{\rm in}\rangle}{\langle \psi_{\rm fin}|\psi_{\rm in}\rangle}
\label{lzw0}=
\frac{\ell}{\cos\theta} .
\end{eqnarray}

Assuming sufficient regularity for the functions on which the operators act,
we can now compute the weak value of $e^{i\hat{L}_z\varphi}$ with the same pre- and post-selection, when
$\varphi$ small, as follows:
\[
 \langle \psi_{\rm fin}| e^{i\hat{L}_z\varphi} |\psi_{\rm in}\rangle\approx  \langle \psi_{\rm fin}| 1+{i\hat{L}_z\varphi} |\psi_{\rm in}\rangle= \langle \psi_{\rm fin}|\psi_{\rm in}\rangle\left(1+i\frac{\langle \psi_{\rm fin}|\hat{L}_z\varphi|\psi_{\rm in}\rangle}{\langle \psi_{\rm fin}|\psi_{\rm in}\rangle}\right).
\]
For small values of $\varphi$, the last bracket can be seen again as a first order approximation of an exponential and so
\[
 \la e^{i\hat{L}_z\varphi}\ra_w\approx 1+i\varphi \la \hat L_z\ra_w\approx e^{i\varphi \la \hat L_z\ra_w}=\exp \left[{i\varphi\frac{\ell}{\cos\theta}}\right],
\]
which suggests the existence of superoscillations in $\varphi$ with a frequency much larger than $\ell$ (the maximal eigenvalue of $\hat{L}_z$) which occurs when $\cos\theta$ is very small.

\bigskip

To show this, using (\ref{lz}), we write:

\begin{equation}
 \langle \psi_{\rm fin}| \exp\left[{i\hat{L}_z\varphi ({2\cos\theta})}\right] |\psi_{\rm in}\rangle=\langle \psi_{\rm fin}| \exp\left[{i}(\hat{A}+\hat{B})\varphi\right] |\psi_{\rm in}\rangle,
\label{note}
\end{equation}
where
\[
\hat{A}:=\hat{L}_z\cos\theta-\hat{L}_x\sin\theta, \qquad \hat{B}:=\hat{L}_z\cos\theta+\hat{L}_x\sin\theta.
\]
We now use the Baker-Campbell-Hausdorff formula for a pair of operators $\hat{A}$ and $\hat{B}$, that is:
\begin{equation}\label{BCH}
\begin{split}
&\exp({\hat{A}\varphi})\exp( e^{\hat{B}\varphi}) \\
&= \exp( {(\hat{A}+\hat{B})\varphi+\frac{1}{2}[\hat{A},\hat{B}]\varphi^2+\frac{1}{12}
[[\hat{A},[\hat{A},\hat{B}]]-[\hat{B},[\hat{A},\hat{B}]]\varphi^3+ \ldots})
\end{split}
\end{equation}
If $[\hat{A},\hat{B}]$ and all the various commutators involving $\hat{A}$  and $\hat{B}$ can be neglected, then we may rewrite (\ref{note}) as:
\begin{equation}
\langle \psi_{\rm fin}| \exp\left[{i\hat{L}_z\varphi ({2\cos\theta})}\right] |\psi_{\rm in}\rangle=\langle \psi_{\rm fin}| e^{{i}\hat{A}\varphi}e^{{i}\hat{B}\varphi} |\psi_{\rm in}\rangle=e^{2i\ell \varphi},
\label{note2}
\end{equation}
where, as above, we act with $\hat{B}$ to the right and with $\hat{A}$ to the left.
In the sequel we assume that $\ell$ is large and we consider small variations of $\varphi$ so that
$\varphi=\varphi'/\ell$, with $\varphi'$ bounded. The first commutator under consideration is
\[
\left[\hat{A}\frac{\varphi'}{\ell},\hat{B}\frac{\varphi'}{\ell}\right]
=\frac{2}{\ell^2}\sin\theta\cos\theta \varphi'^2\hat{L}_y\approx \frac{\cos\theta}{\ell},
\]
 where we used the fact that, in general, $|\hat{L}_y|$ is bounded by $\ell$.  In the case of
interest $\hat{L}_y$ is almost orthogonal to both angular momenta $\hat{A}$ and $\hat{B}$ defining the initial and final state. Thus its value is
$O(1)$ and the commutator is in fact approximated by $1/\ell^2$.
Morever, the various $k$-fold commutators in the Baker-Campbell-Hausdorff formula (\ref{BCH}) will always give just one angular momentum operator with
some pre-factor depending on $\theta$ which is bounded by one and with $1/\ell^k$ suppressions.
 The series in formula (\ref{BCH}) is strongly convergent, thus we can neglect all the commutators in the limit, when $\ell$ becomes large. This then recoups (\ref{note2}) which exhibits  extremely strong sensitivity to small $\varphi\cos{\theta}$ rotations or, in other words, superoscillations.

\bigskip

We now recall some notations and properties of quantum operators for the case of spin-$\frac{1}{2}$.
The following well-known matrices $\hat{\textrm{S}}_x$, $\hat{\textrm{S}}_y$, $\hat{\textrm{S}}_z$ represent the following spin-$\frac{1}{2}$ operators:
\[
 \hat{\textrm{S}}_x=\frac{1}{2}
\left[\begin{matrix}
0&1\\
1&0\\
 \end{matrix}
\right],
\quad
\hat{\textrm{S}}_y=\frac{1}{2}
\left[\begin{matrix}
0&-i\\
i&0\\
 \end{matrix}
\right],\quad
\hat{\textrm{S}}_z=\frac{1}{2}
\left[\begin{matrix}
1&0\\
0&-1\\
 \end{matrix}
\right].
\]
Easy computations show that $\hat{\textrm{S}}_x$, $\hat{\textrm{S}}_y$, $\hat{\textrm{S}}_z$ satisfy the commutation relations
$$
[\hat{\textrm{S}}_\alpha,\hat{\textrm{S}}_\beta]=i\varepsilon_{\alpha\beta\gamma}\hat{\textrm{S}}_\gamma
$$
and
\[
(\hat{\textrm{S}}_\alpha)^2=\frac{1}{4}
\left[\begin{matrix}
1&0\\
0&1\\
 \end{matrix}
\right].
\]
We also have the so-called raising and lowering operators which act as follows:
\begin{equation}
\hat{\textrm{S}}_+=\hat{\textrm{S}}_x+i\hat{\textrm{S}}_y,\qquad \hat{\textrm{S}}_-=\hat{\textrm{S}}_x-i\hat{\textrm{S}}_y,
\label{raiselower}
\end{equation}
which satisfy the relations
\begin{equation}\label{loweincr}
\hat{\textrm{S}}_+
\left[
\begin{matrix}
1\\
0
\end{matrix}\right]
=
\hat{\textrm{S}}_-
\left[
\begin{matrix}
0\\
1
\end{matrix}\right]
=
0,
\,\,\,\, \,\,\hat{\textrm{S}}_+
\left[
\begin{matrix}
0\\
1
\end{matrix}\right]
=
\left[
\begin{matrix}
1\\
0
\end{matrix}\right],\,\,\,\,\,\,
\hat{\textrm{S}}_-
\left[
\begin{matrix}
1\\
0
\end{matrix}\right]
=
\left[
\begin{matrix}
0\\
1
\end{matrix}\right].
\end{equation}

 Let us now consider a system of $N$  non interacting spin-$\frac{1}{2}$ particles with associated commuting operators $$\hat{\textrm{S}}_\alpha^{(1)},\ldots,\hat{\textrm{S}}_\alpha^{(j)},\ldots, \hat{\textrm{S}}_\alpha^{(N)}$$ (note that $\hat{\textrm{S}}_\alpha^{(j)}$ operates only on the particle $j$). The total spin operator for all $N$ particles is defined by the operator:
$$
\hat{\mathcal{L}}_\alpha\upn:=\sum_{j=1}^N \hat{\textrm{S}}_\alpha^{(j)}.
$$
One can verify that
the commutation relations between the individual spin-$\frac{1}{2}$ operators, namely $\hat{\textrm{S}}_\alpha$, translate into the same commutation relations
for the $N$ particle system, $\hat{\mathcal{L}}_\alpha\upn$:
\[
 [\hat{\mathcal{L}}_\alpha\upn,\hat{\mathcal{L}}_\beta\upn]=i\varepsilon_{\alpha\beta\gamma}\hat{\mathcal{L}}_\gamma\upn.
\]

\begin{proposition}\label{spin-elle}
In the product wavefunction $|\psi\upn\rangle$, with each individual particle $j$ in the state $|\hat{\textrm{S}}_z^{(j)}=\frac{1}{2} \rangle$:
\begin{equation}
|\psi\upn\rangle = | \hat{\textrm{S}}_z^{(1)}=\frac{1}{2}\rangle \ldots  | \hat{\textrm{S}}_z^{(N)}=\frac{1}{2}\rangle=\prod_{{j=1}}^N |{\hat{\textrm{S}}_z^{(j)}=\frac{1}{2}} \rangle
\label{Nspin}
\end{equation}
we have
\begin{equation}
\hat{\mathcal{L}}_z\upn|\psi\upn\rangle = \frac{N}{2}|\psi\upn\rangle
\label{sz}
\end{equation}
and
\begin{equation}
(\hat{\mathcal{L}}\upn)^2\ \ | \psi\upn\rangle = \frac{N}{2}
\Big(\frac{N}{2}+1\Big)\, |\psi\upn\rangle .
\label{sz2}
\end{equation}

\end{proposition}
\begin{proof}
The relation (\ref{sz}) follows directly from the definition and hence
$$(\hat{\mathcal{L}}_z\upn)^2\ |\psi\upn\rangle=\Big(\frac N2\Big)^2|\psi\upn\rangle.$$
Formula (\ref{sz2}) then follows from:
\begin{equation}
((\hat{\mathcal{L}}_x\upn)^2+
(\hat{\mathcal{L}}_y\upn)^2) \, |\psi\upn \rangle= \frac N2 \, |\psi\upn\rangle .
\label{mixed}
\end{equation}
To verify this relation
we expand $\hat{\mathcal{L}}_x\upn$ and $\hat{\mathcal{L}}_y\upn$ in terms of the individual raising and lowering operators $\hat{\textrm{S}}_+^{(j)}$ and $\hat{\textrm{S}}_-^{(j)}$, see (\ref{raiselower}). The mixed terms in (\ref{mixed}) can be written as
\[
 \hat{\textrm{S}}_x^{(1)}\hat{\textrm{S}}_x^{(2)}+\hat{\textrm{S}}_y^{(1)}\hat{\textrm{S}}_y^{(2)}=\hat{\textrm{S}}_+^{(1)}\hat{\textrm{S}}_-^{(2)}+\hat{\textrm{S}}_-^{(1)}\hat{\textrm{S}}_+^{(2)}
\]
and similarly for the other cases. This expression
clearly vanishes due to the presence of $\hat{\textrm{S}}_+^{(1)}$ or $\hat{\textrm{S}}_+^{(2)}$ which operate on $|\hat{\textrm{S}}_z^{(1)}=\frac{1}{2}\rangle$ and $|\hat{\textrm{S}}_z^{(2)}=\frac{1}{2}\rangle$ respectively, thereby yielding $0$  by (\ref{loweincr}). This leaves only the diagonal terms so that
\[
\left[(\hat{\mathcal{L}}_x\upn)^2+(\hat{\mathcal{L}}_y\upn)^2\right]|\psi\upn\rangle=\sum_{j=1}^N\left[
(\hat{\textrm{S}}_x^{(j)})^2+(\hat{\textrm{S}}_y^{(j)})^2\right]|\psi\upn\rangle=\frac N2|\psi\upn\rangle
\]
since
$$(\hat{\textrm{S}}_x^{(j)})^2|\psi\upn\rangle=\frac{1}{4}|\psi\upn\rangle\qquad {\rm and}\qquad (\hat{\textrm{S}}_y^{(j)})^2|\psi\upn\rangle=\frac{1}{4}|\psi\upn\rangle$$ and we have $N$ such terms.
\end{proof}
\begin{remark}
 {\rm From now on we will consider systems with an even number of particles and we set $\frac N2=\ell$.}
\end{remark}
\begin{remark}
 {\rm The calculations we have done in this chapter have assumed that we are looking at the system of $2\ell$ particles, all of which have spin-$\frac{1}{2}$ and this has allowed us to calculate the eigenvalues for $(\hat{\mathcal{L}}\upn)^2$, $\hat{\mathcal{L}}_z\upn$ and $(\hat{\mathcal{L}}_x\upn)^2+(\hat{\mathcal{L}}_y\upn)^2$.
 Let us now repeat the same calculations for a system with $\ell+m$ particles in the state $|\hat{\textrm{S}}_z^{(j)}=\frac{1}{2} \rangle$ and $\ell-m$ particles $|\hat{\textrm{S}}_z^{(j)}=-\frac{1}{2} \rangle$ (therefore $N=2\ell$). First of all we notice that in Proposition \ref{spin-elle} we obtain
 \begin{equation}
\hat{\mathcal{L}}_z\upn|\psi\upn\rangle = m |\psi\upn\rangle
\end{equation}
and therefore
\begin{equation}
(\hat{\mathcal{L}}_z\upn)^2|\psi\upn\rangle = m^2 |\psi\upn\rangle .
\end{equation}
Since $$(\hat{\mathcal{L}}\upn)^2= (\hat{\mathcal{L}}\upn_x)^2+(\hat{\mathcal{L}}\upn_y)^2+(\hat{\mathcal{L}}\upn_z)^2$$ and $$((\hat{\mathcal{L}}\upn_y)^2+(\hat{\mathcal{L}}\upn_z)^2)|\psi\upn\rangle =\ell |\psi\upn\rangle, $$ we conclude that
$$(\hat{\mathcal{L}}\upn)^2|\psi\upn\rangle =(\ell+m^2)|\psi\upn\rangle .$$
 Clearly when $m=\ell$ we obtain the previous result.}
\end{remark}

Our next goal is to compute the weak value of the operator $\exp(i\hat{\mathcal{L}}_z\upn\delta\varphi)$ for a small angle $\delta\varphi$.

We will again utilize the $N$ spin-$\frac{1}{2}$ particles, but now we use the initial state:
\[
| \psi_{\rm in}\upn\rangle = | \hat{\mathcal{L}}_z\upn \cos\theta+\hat{\mathcal{L}}_x\upn \sin\theta=\ell\rangle,
\]
and the final state:
\[
| \psi_{\rm fin}\upn\rangle = | \hat{\mathcal{L}}_z\upn \cos\theta - \hat{\mathcal{L}}_x\upn \sin\theta=\ell\rangle .
\]
Then
$| \psi_{\rm in}\upn\rangle$ and $| \psi_{\rm fin}\upn\rangle$ can be realized in two different ways.  Consider first $| \psi_{\rm in}\upn\rangle$
which is the eigenfunction of the angular momentum operator
$$\hat{\mathcal{L}}_{\theta}\upn=\hat{\mathcal{L}}_z\upn \cos\theta+\hat{\mathcal{L}}_x\upn \sin\theta ,$$
i.e. the result of rotating $\hat{\mathcal{L}}\upn$ of the original spin-$\frac{1}{2}$ particles (\ref{Nspin}) by an angle $\theta$  close to $\pi/2$ with respect to the $y$-axis, which is realized by the operation $\exp({i\hat{\mathcal{L}}_y\upn\theta})$.  In the spin-$\frac{1}{2}$ product basis, this is obtained using $\exp({i\sum_{j=1}^N \hat{\textrm{S}}_y^{(j)}\theta})$.  This amounts to rotating each of the individual spin-$\frac{1}{2}$ states $|{\hat{\textrm{S}}_z^{(j)}=\frac{1}{2}} \rangle=\left[
\begin{matrix}
1\\
0
\end{matrix}\right]$ in the same way, giving
\[
|\psi_{\rm in}\rangle=\left[
\begin{matrix}
\cos\frac{\theta}{2}\\
\ \ \\
\sin\frac{\theta}{2}
\end{matrix}\right]
\]
for each of the $N$ spin-$\frac{1}{2}$ particles. Similarly, $| \psi_{\rm fin}\upn\rangle$, the corresponding eigenstate of $\hat{\mathcal{L}}_z\upn \cos\theta - \hat{\mathcal{L}}_x\upn \sin\theta$ with maximum eigenvalue $m=\ell$, is obtained by rotating $| \psi_{\rm in}\upn\rangle$ by $\phi=\pi$ around the $z$-axis.  In the $z$-basis of the original $N$ spin-$\frac{1}{2}$ particles, the rotation of each individual spin-$\frac{1}{2}$ by $\phi$ modifies the above $|\psi_{\rm in}\rangle$ to:
\[
\left[
\begin{matrix}
\cos\frac{\theta}{2}e^{i\frac{\phi}{2}}\\
\ \ \\
\sin\frac{\theta}{2}e^{-i\frac{\phi}{2}}
\end{matrix}\right].
\]
Thus, $|\psi_{\rm fin}\rangle$ can be written as (apart from an overall phase of $e^{i\frac{\pi}{2}}$):
\[
|\psi_{\rm fin}\rangle=\left[
\begin{matrix}
\cos\frac{\theta}{2}\\
\ \  \\
-\sin\frac{\theta}{2}
\end{matrix}\right].
\]
\noindent
We now prove the following result:
\begin{theorem}
The weak value of the operator
$\exp(i\hat{\mathcal{L}}_z\upn\delta\varphi)$ for a small angle $\delta\varphi$
when we take $N=2\ell$ spins is:
\begin{equation}
\left[\exp({i\hat{\mathcal{L}}_z\upn\frac{\delta\varphi'}{2\ell}})\right]_w=\dfrac{\left(\cos^2(\theta/2)
\exp(i\frac{\delta\varphi'}{2\ell})-\sin^2(\theta/2)
\exp(-i\frac{\delta\varphi'}{2\ell})\right)^{2\ell}}{\left(\cos\theta\right)^{2\ell}}.
\end{equation}
\end{theorem}
\begin{proof}
To begin with, we note that the scalar product of the initial and final state for a single spin-$\frac{1}{2}$ particle is:
\[
 \langle\psi_{\rm fin}|\psi_{\rm in}\rangle =\cos^2(\theta/2)-\sin^2(\theta/2)=\cos\theta
\label{scalarproduct}
\]
so that the scalar product of the initial and final state for the system of the $N$ independent spin-$\frac{1}{2}$ particles is:
\[
 \langle\psi_{\rm fin}\upn|\psi_{\rm in}\upn\rangle =\left[\cos^2(\theta/2)-\sin^2(\theta/2)\right]^N=(\cos\theta)^N=(\cos\theta)^{2\ell}.
\label{big-scalar-prod}
\]
We now set $\delta\varphi'\!=\!\ell\delta\varphi$ and compute the weak value:
\begin{equation}\label{wvlargespin1}
\begin{split}
\left[\exp({i\hat{\mathcal{L}}_z\upn\frac{\delta\varphi'}{\ell}})\right]_w &= \frac{\langle\psi_{\rm fin}\upn|\exp({i\hat{\mathcal{L}}_z\upn\frac{\delta\varphi'}{\ell}})|
\psi_{\rm in}\upn\rangle}{\langle\psi_{\rm fin}\upn|\psi_{\rm in}\upn\rangle}\\
&=
\frac{ {\prod_{k=1}^N \langle}\psi_{\rm fin}^{(k)}|
~\exp({i\frac{\delta\varphi'}{\ell}} \sum_{n=1}^N \hat{\textrm{S}}_z^{(n)})
 ~ \prod_{{j=1}}^N |\psi_{\rm in}^{(j)}\rangle}
{ {\prod_{k=1}^N \langle}\psi_{\rm fin}^{(k)}|\prod_{{j=1}}^N |\psi_{\rm in}^{(j)}\rangle}.
\end{split}
\end{equation}
Since $$\langle\psi_{\rm fin}^{(k)}|\hat{\textrm{S}}_z^{(n)}|\psi_{\rm in}^{(j)}\rangle=\delta_{k,j},$$
it suffices to calculate one of the products and take the result to the $N$-th power. Thus we can rewrite formula  \eqref{wvlargespin1} as:
\begin{equation}
\prod_{j=1}^N \frac{\langle\psi_{\rm fin}^{(j)}| \exp({i\frac{\delta\varphi'}{\ell}} \hat{\textrm{S}}_z^{(j)})|\psi_{\rm in}^{(j)} \rangle}{\langle\psi_{\rm fin}^{(j)}|\psi_{\rm in}^{(j)} \rangle}=
\frac{\left\{ \langle\psi_{\rm fin}| \exp({i\frac{\delta\varphi'}{\ell}} \hat{\textrm{S}}_z)|\psi_{\rm in}\rangle\right\}^N}{(\cos\theta)^N}.
\label{bigspinbbis}
\end{equation}
Since the eigenvalues of $\hat{\textrm{S}}_z^{(j)}$ in the state $[1\ 0]^T$ and $[0\ 1]^T$ (the superscript $T$ denotes the transpose) are respectively $\frac{1}{2}$ and  $-\frac{1}{2}$, we deduce that for $|\psi_{\rm in}\rangle$:
\[
\exp\left({i\frac{\delta\varphi'}{\ell}} \hat{\textrm{S}}_z\right)|\psi_{\rm in}\rangle=\exp({i\frac{\delta\varphi'}{\ell}} \hat{\textrm{S}}_z)
\left[
\begin{matrix}
 \cos(\theta/2)\\
\sin(\theta/2)
\end{matrix}
\right]
=
\left[
\begin{matrix}
 \cos(\theta/2) \exp({i\frac{\delta\varphi'}{2\ell}})\\
\sin(\theta/2) \exp({-i\frac{\delta\varphi'}{2\ell}})
\end{matrix}
\right].
\label{rot-phi}
\]
Computing the scalar product with the post-selected state, we have:
\begin{eqnarray}
\langle\psi_{\rm fin}| \exp\left({i\frac{\delta\varphi'}{\ell}} \hat{\textrm{S}}_z\right)|\psi_{\rm in}\rangle &= & [\cos(\theta/2)\ \  -\sin(\theta/2 )]
\left[
\begin{matrix}
 \cos(\theta/2) \exp({i\frac{\delta\varphi'}{2\ell}})\\
\sin(\theta/2) \exp({-i\frac{\delta\varphi'}{2\ell}})
\end{matrix}
\right]\nonumber\\
& =&\cos^2(\theta/2)e^{i\frac{\delta\varphi'}{2\ell}}-\sin^2(\theta/2)e^{-i\frac{\delta\varphi'}{2\ell}}.
\end{eqnarray}
When we take $N=2\ell$ spins, we obtain the weak value:
\begin{equation}\label{8.20}
\left[\exp\left({i\hat{\mathcal{L}}_z\upn\frac{\delta\varphi'}{2\ell}}\right)\right]_w=\dfrac{\left(\cos^2(\theta/2)
\exp({i\frac{\delta\varphi'}{2\ell}})-\sin^2(\theta/2)
\exp({-i\frac{\delta\varphi'}{2\ell}})\right)^{2\ell}}{\left(\cos \theta \right)^{2\ell}}.
\end{equation}
\end{proof}
\begin{remark}{\rm Let us consider the case in which we have $\ell+m$ spin $\frac 12$ particles and $\ell+m$ spin $-\frac 12$ particles.
The state corresponding to $-\frac 12$ particle can be represented as the vector $[0 \ 1]^T$ which after a $\theta/2$ rotation gives us the initial state $|\psi_{\rm in}\rangle =[-\sin(\theta/2)\ \cos(\theta/2)]^T$. Arguing as in the proof of the previous Theorem we obtain that, up to
a phase change factor, $|\psi_{\rm fin}\rangle =[-\sin(\theta/2)\ -\cos(\theta/2)]^T$. This implies that $$\langle \psi_{\rm in}|\psi_{\rm fin}\rangle =\sin^2(\theta/2)-\cos^2(\theta/2)=-\cos \theta .$$ If we replace $\langle \psi_{\rm in}|\psi_{\rm fin}\rangle$ with $\langle \psi^N_{\rm in}|\psi^N_{\rm fin}\rangle$ we obtain $$\langle \psi^N_{\rm in}|\psi^N_{\rm fin}\rangle =(-1)^{\ell -m}(\cos \theta)^{2\ell}.$$ The rest of the computations carry over in the same way and demonstrate the same superoscillatory phenomenon even in this case.
}
\end{remark}
\begin{remark}{\rm
A key observation is that (\ref{8.20}) has the same form as the
standard superoscillating sequence
\begin{equation}
F_n(x,a)=\left[\cos \Big(\frac{2\pi x}{n}\Big)+ia\sin \Big(\frac{2\pi x}{n}\Big)\right]^n
\end{equation}
where $a:=\frac{1}{\cos \theta}$.
Since $\cos\theta$ can be
arbitrarily small, depending on the pre- and post-selection,  the weak value of $\hat{\mathcal{L}}_z\upn$ (which is $\frac{\ell}{\cos\theta}$) can be arbitrarily large, well outside the spectrum of $\hat{\mathcal{L}}_z\upn$ which is $[-\ell,\ell]$ and thus we see that the largest wavelength in the expansion is one.
However, as we have shown previously, as long as $|x|<\sqrt n$, $F_n(x,a)$ can be approximated as $F_n(x,a)\approx e^{i2\pi a x}$. When $\ell$ becomes large, this expression behaves like $\exp(ia\delta\varphi')$ which, recalling that $\delta\varphi'=\ell\delta\varphi$ and $\delta\varphi$ is small, is just $\exp(i\frac{\ell}{\cos\theta}\delta\varphi)$.
}
\end{remark}

\bigskip


The $SO(3)$ superoscillatory phenomenon discussed before depends on the choice of $|\psi_{\rm fin}\upn\rangle$.  In the example above, we generated $|\psi_{\rm fin}\upn\rangle$ by rotating $|\psi_{\rm in}\upn\rangle$ by an angle $\varphi_0:=\varphi=\pi$ around which we performed a further small rotation $\delta\varphi$. In general, this phenomenon depends on the choice of $\varphi_0$.

\begin{remark}{\rm
To better understand the conditions under which superoscillations occur, we consider the case $|\psi_{\rm fin}\upn\rangle=|\psi_{\rm in}\upn\rangle$, corresponding to $\varphi=0$.
In this case we can repeat the argument given before where $|\psi_{\rm fin}\rangle$ is now given by the vector $[\cos(\theta/2)\ \ \sin (\theta/2)]^T$ so that $\langle \psi_{\rm fin}\upn|
\psi_{\rm in}\upn\rangle=1$. Therefore
the weak value is given by
\[
\left[\exp({i\hat{\mathcal{L}}_z\upn\frac{\delta\varphi'}{\ell}})\right]_w=
 \left(\cos^2(\theta/2)\exp(i\frac{\delta\varphi'}{4\ell})+\sin^2(\theta/2)
 \exp(-i\frac{\delta\varphi'}{4\ell})\right)^{2\ell}
 \]
 \[\approx \left(1+i\frac{\delta\varphi'}{2\ell}\cos\theta\right)^{2\ell}\approx \exp({i\cos\theta\frac{\delta\varphi'}{2}}).
\label{lzweak}
\]
Since $\cos\theta\leq 1$, it is evident that superoscillations do not occur.
}
\end{remark}

\bigskip
Next goal is to consider the case in which we take a small rotation $\delta\varphi$ centered around $\varphi_0$ and, for simplicity, we will still consider $|\psi_{\rm in}\rangle=|\psi_{\rm fin}\rangle$. The computations we have done before show that
\[
\begin{split}
&\langle \psi_{\rm fin} | \exp\left[{\frac{i}{2}\hat{\mathcal{L}}_z\upn (\varphi_0+\frac{\delta\varphi'}{\ell}})\right]|\psi_{\rm in}\rangle\\
 &=\cos^2(\theta/2) \exp({\frac{i}{2}(\varphi_0+\frac{\delta\varphi'}{\ell})} ) +\sin^2(\theta/2) \exp({-\frac{i}{2}(\varphi_0+\frac{\delta\varphi'}{\ell})})\\
& =\cos(\frac{\varphi_0}{2})\cos(\frac{\delta\varphi'}{2\ell}) -\sin(\frac{\varphi_0}{2})\sin(\frac{\delta\varphi'}{2\ell}) \\ &+i\cos\theta[\sin(\frac{\varphi_0}{2})\cos(\frac{\delta\varphi'}{2\ell}) +\cos(\frac{\varphi_0}{2})\sin(\frac{\delta\varphi'}{2\ell})]\\
\end{split}
\]
and since $\delta\varphi'$ is small, we can approximate the above expression as
\[
 \cos(\frac{\varphi_0}{2})-\frac{\delta\varphi'}{2\ell}\sin(\frac{\varphi_0}{2}) +i\cos\theta(\sin(\frac{\varphi_0}{2})+\frac{\delta\varphi'}{2\ell}\cos(\frac{\varphi_0}{2})).
\]
By taking the $2\ell$-th power we obtain
\[
 \left(\cos\left(\frac{\varphi_0}{2}\right)+i\cos\theta\sin\left(\frac{\varphi_0}{2}\right)\right)^{2\ell}
\left[1+ \frac{\delta\varphi'}{\ell} \frac{i\cos\theta\cos(\frac{\varphi_0}{2})-\sin(\frac{\varphi_0}{2})}{\cos(\frac{\varphi_0}{2})+i\cos\theta\sin(\frac{\varphi_0}{2})} \right]^{2\ell}
\]
For $\ell$ large, the right hand side converges to
$\exp({i(a+ib)(\delta\varphi'/\ell)})$ where
$a$ corresponds to the imaginary part of
\[
 \frac{i\cos\theta\cos(\frac{\varphi_0}{2})-\sin(\frac{\varphi_0}{2})}{\cos(\frac{\varphi_0}{2})+i\cos\theta\sin(\frac{\varphi_0}{2})}
\]
namely
\[
 \frac{\cos\theta}{\cos^2(\frac{\varphi_0}{2})+\cos^2\theta\sin^2(\frac{\varphi_0}{2})}.
\]
When $\varphi_0$ approximates $0$, then $a$ approximates $\cos\theta$ and thus there are no superoscillations. When if $\varphi_0$ is near $\pi$ then $a\approx 1/\cos\theta$ and the superoscillation phenomenon occurs. This is fully consistent with the previous analysis.

\bigskip
To conclude this section, we note that the same argument can be applied to the case in which $\psi_{\rm fin}$ is associated with the vector $[\cos(\theta/2) \ \ -\sin(\theta/2)]^T$. In this case, we obtain that the frequency is, up to the factor $\delta\varphi/\ell$, given by
\[
 \frac{\cos\theta}{\sin^2(\varphi_0)+\cos^2\theta\cos^2(\varphi_0)},
\]
which reproduces the superoscillations when
\[
|\cos(\varphi_0)|< \frac{1}{\sqrt{1+\cos\theta}}.
\]

\section{Asymptotic expansion for the Wigner functions}
In this section we use the preceding discussions  to show a new interesting relationship for Wigner's $d^{(j)}_{\ell m}$ matrices.  These matrices are the representation of the rotation group in the standard $|\ell,m\rangle$ basis defined by
\begin{equation}
 \exp({-i\hat{L}_y\theta})\ | \ell ,m \rangle =\sum_{m'} | \ell,m' \rangle d^{(\ell)}_{m' m}(\theta).
\label{dmatrix}
\end{equation}
The $d^{(\ell)}_{m' m}(\theta)$ are orthogonal functions on the sphere, and for any given $\theta$, form a $(2\ell+1)(2\ell+1)$ unitary matrix.
In general, $d^{(\ell)}_{m' m}(\theta)$ are polynomials in $\cos \frac{\theta}{2}$ and $\sin \frac{\theta}{2}$.

These polynomials simplify into just one term in the special case $m=\ell$ (see e.g. \cite{sakurai}):
\[
d^{(\ell)}_{m' \ell}(\theta)=(-1)^{\ell -m'}
\left[\frac{(2\ell)!}{(\ell+m' )! (\ell -m')! } \right]^{\frac 12}\left[\cos\left(\frac{\theta}{2}\right)\right]^{\ell +m'}
\left[\sin\left(\frac{\theta}{2}\right)\right]^{\ell -m'}.
\]
Thus $[d^{(\ell)}_{m' \ell}(\theta)]^2$
is the $(j-m')$-th term in the binomial expansion of
\[
\left[\cos^2\left(\frac{\theta}{2}\right)+\sin^2\left(\frac{\theta}{2}\right)\right]^{2\ell}.
\]
As a consequence we have, as expected, $$\sum_{m'} | d^{(\ell)}_{m' m}(\theta)|^2=\bra {\psi_{\rm in}}\psi_{\rm in}\rangle =1.$$
By computing the scalar product
$\bra{\psi_{\rm fin}}\psi_{\rm fin}\rangle$
in two different ways, gives us the following (probably very well known) result:
\begin{proposition}
 With the notations above, for any angle $\theta$ we have:
\[
 \sum_{m'} (-1)^{m'} d^{(\ell)}_{m' \ell}(\theta)d^{(\ell)}_{m' \ell}(\theta)=(\cos\theta)^{2\ell}.
\]
\end{proposition}
\begin{proof}
We consider the initial and final state $|\psi_{\rm in}\rangle$, $|\psi_{\rm fin}\rangle$ introduced above.
 We first note that
\[
| \psi_{\rm in}\upn\rangle=\sum_{m'}\, | \ell ,m'\rangle d^{(\ell)}_{m' \ell}(\theta)\qquad
\]
and, due to the $\pi$ rotation around the $z$-axis (leading to the $(-1)^{m'}$ factor):
\[
|\psi_{\rm fin}\upn\rangle=\sum_{m'} \, | \ell ,m'\rangle (-1)^{m'} d^{(\ell)}_{m' \ell}(\theta)
\]
and so the scalar product of the initial and final state $ \langle\psi_{\rm fin}\upn|\psi_{\rm in}\upn\rangle$, see (\ref{big-scalar-prod}), is given by:
\[
 \langle \psi_{\rm fin}\upn| \psi_{\rm in}\upn \rangle =\sum_{m'} (-1)^{m'}d^{(\ell)}_{m' \ell}(\theta)
d^{(\ell)}_{m' \ell}(\theta)=(\cos\theta)^{2\ell}.
\]
The same scalar product  can also be computed directly:
\[
\begin{split}
 \sum_{m'} (-1)^{m'} &d^{(\ell)}_{m' \ell}(\theta)d^{(\ell)}_{m' \ell}(\theta)\\
 &=\sum_{m'}(-1)^{m'} \left[\frac{(2\ell)!}{(\ell+m' )! (\ell -m')! } \right]\left[\cos\left(\frac{\theta}{2}\right)\right]^{\ell +m'}
\left[\sin\left(\frac{\theta}{2}\right)\right]^{\ell -m'}\\
 &=\left[\cos^2\left(\frac{\theta}{2}\right)-\sin^2\left(\frac{\theta}{2}\right)\right]^{2\ell}=(\cos\theta)^{2\ell},\\
 \end{split}
\]
and the statement follows.
\end{proof}
More interesting and novel, however, is the following asymptotic formula which concludes this section.
\begin{theorem}
 With the notations above, for any angle $\theta$ and any small value of $\delta\varphi$ we have, for large $\ell$:
\[
 \sum_{m'}(-1)^{m'}d^{(\ell)}_{m' \ell}(\theta)d^{(\ell)}_{m' \ell}(\theta) e^{im'\delta\varphi}=\left[\frac {1}{a} \left( \cos\left(\frac{\delta\varphi}{2}\right)+ia \sin\left(\frac{\delta\varphi}{2}\right) \right) \right]^{2\ell},
\]
where $  a=\frac{1}{\cos \theta}$, which, again, is of the superoscillating form (\ref{eq1}).
\end{theorem}
\begin{proof}
 The weak value of $e^{i\hat{\mathcal{L}}_z\upn \delta\varphi}$ is given by (\ref{lzweak}).  On the other hand we have
\[
 \frac{ \langle\psi_{\rm fin}| e^{i\hat{\mathcal{L}}_z\upn \delta\varphi}|\psi_{\rm in}\rangle}{\langle \psi_{\rm fin}\upn| \psi_{\rm in}\upn\rangle} =\frac{\sum_{m'}(-1)^{m'}d^{(\ell)}_{m' \ell}(\theta)d^{(\ell)}_{m' \ell}(\theta)e^{im'\delta\varphi}}{(\cos\theta)^{2\ell}}.
\]
from which the assertion follows.

\bigskip

We can also directly compute:
\[
\begin{split}
&\sum_{m'}(-1)^{m'}d^{(\ell)}_{m' \ell}(\theta)d^{(\ell)}_{m' \ell}(\theta)e^{im'\delta\varphi}\\
&=\sum_{m'}  \left[\frac{(2\ell)!}{(\ell+m' )! (\ell -m')! } \right] \left(e^{i\delta\varphi/2}\cos^2\left(\frac{\theta}{2}\right)\right)^{\ell+m'} \left(e^{-i\delta\varphi/2}\sin^2\left(\frac{\theta}{2}\right)\right)^{\ell-m'}\\
&=\left[e^{i\delta\varphi/2}\cos^2\left(\frac{\theta}{2}\right) - e^{-i\delta\varphi/2}\sin^2\left(\frac{\theta}{2}\right)\right]^{2\ell}\\
&=\left[e^{i\delta\varphi/2}\frac{\cos^2\left(\frac{\theta}{2}\right)}{\cos\theta} - e^{-i\delta\varphi/2}\frac{\sin^2\left(\frac{\theta}{2}\right)}{\cos\theta}\right]^{2\ell}(\cos\theta)^{2\ell}\\
&=\left[\cos\left(\frac{\delta\varphi}{2}\right)+i\frac{1}{\cos\theta}\sin\left(\frac{\delta\varphi}{2}\right)\right]^{2\ell}(\cos\theta)^{2\ell}
\end{split}
\]
yielding the same result as above.
\end{proof}




\begin{thebibliography}{99}


\bibitem{aav} Y. Aharonov, D. Albert, L. Vaidman, {\em How the Results of a Measurement of a component of a spin $\frac12$ particle can turn out to be $100$?},  Phys. Rev. Lett. {\bf 60} (1988), 1351--1354.

\bibitem{superosc0} Y. Aharonov, E. Ben-Reuven, S. Popescu, D. Rohrlich, {\em Perturbative induction of vector potentials}, Tel Aviv University preprint TAUP 184790 (1990).

\bibitem{abl} Y. Aharonov, P.G. Bergmann, J.L. Lebowitz, {\em Time Symmetry in the quantum process of measurement},   Phys. Rev. B, {\bf 134} (1964) 1410--6.
Reprinted in Y. Aharonov, P.G. Bergmann, J.L. Lebowitz, Quantum Theory and Measurement ed
J. A. Wheeler and W. H. Zurek (Princeton, NJ: Princeton University Press), (1983), 680--686.



\bibitem{ab} Y.~Aharonov, D.~Bohm {\em Time in the Quantum Theory and the Uncertainty Relation for Time and Energy}, Phys. Rev. {\bf 122}, 1649 (1961); reprinted in
{\it Quantum Theory and Measurement}, eds. J.~A.~Wheeler, W.~H.~Zurek, (Princeton University Press), 1983.




\bibitem{botero} Y. Aharonov, A. Botero, {\em Quantum averages of weak values}, { Physical Review A}, {\bf 72} (2005), Art. No. 052111.

\bibitem{at2} Y. Aharonov, A. Botero, S. Popescu, B. Reznik, J. Tollaksen,
{\em Revisiting Hardy's Paradox: Counterfactual Statements, Real Measurements, Entanglement and WVs},
{Phys Lett A}, {\bf 301} (2002), 130--138.

\bibitem{ACSST4}
Y. Aharonov,  F. Colombo,  S. Nussinov, I. Sabadini, D.C. Struppa, J. Tollaksen, {\em Superoscillations phenomenona in SO(3)}, Proc. Royal Soc. A., {\bf 468} (2012), 3587--3600.


\bibitem{acsst}
Y. Aharonov,  F. Colombo,  I. Sabadini, D.C. Struppa, J. Tollaksen, {\em Some mathematical properties of  superoscillations},
{ J. Phys. A}, {\bf 44} (2011), 365304.

\bibitem{ACSST2}
Y. Aharonov,  F. Colombo,  I. Sabadini, D.C. Struppa, J. Tollaksen, {\em On some operators associated to superoscillations},
 {Complex  Anal. Oper. Theory},  {\bf 7} (2013),  1299--1310.

\bibitem{ACSST3}
Y. Aharonov,  F. Colombo,  I. Sabadini, D.C. Struppa, J. Tollaksen, {\em On the Cauchy problem for the Schr\"odinger equation
with superoscillatory intitial data}, J. Math.  Pure Appl.,  (9) {\bf 99} (2013),  165--173.



\bibitem{acsstYa80} Y. Aharonov,  F. Colombo,  I. Sabadini, D.C. Struppa, J. Tollaksen,
 {\em  On superoscillations longevity: a windowed Fourier transform approach}, in D. Struppa, J. Tollaksen (eds)
  {\em  Quantum Theory: A Two-Time Success Story}, Springer 2013, 313--325.

  \bibitem{acsst6}
Y. Aharonov,  F. Colombo,  I. Sabadini, D.C. Struppa, J. Tollaksen,
{\em Superoscillating sequences as solutions of
generalized Schr\"odinger equations}, J. Math. Pure Appl. {\bf 103} (2015), 522--534.

\bibitem{acsst7}
Y. Aharonov,  F. Colombo,  I. Sabadini, D.C. Struppa, J. Tollaksen,
{\em Evolution of superoscillatory  data},  J. Phys. A.
 {\bf 47} (2014), 205301.


\bibitem{adz} Y. Aharonov, L. Davidovich, N. Zagury, {\em Quantum Random Walks}, Phys. Rev., {\bf A48} (1993), 1687.

    \bibitem{at} Y. Aharonov, S. Massar, S. Popescu, J. Tollaksen,  L. Vaidman, {\em Adiabatic Measurements on Metastable Systems}, Phys. Rev. Lett., {\bf 77},  (1996), p. 983.

\bibitem{PT-nov-2010} Y. Aharonov, S. Popescu, J. Tollaksen,  {\em  A time-symmetric formulation of quantum mechanics},
  {Physics Today}, November 2010, p. 27.

\bibitem{eachmoment} Y. Aharonov, S. Popescu, J. Tollaksen, Each moment of time is a new universe, in  D. Struppa, J. Tollaksen (Eds), {\em Quantum Theory: a Two-Time Success Story. Yakir Aharonov Festschrift}. Springer, Milan, 2013.


\bibitem{aharonov_book} Y. Aharonov, D. Rohrlich, {\em Quantum Paradoxes: Quantum Theory for the Perplexed},
Wiley-VCH Verlag, Weinheim, 2005.

\bibitem{townes} Y. Aharonov, J. Tollaksen,  {\em New insights on Time-Symmetry in Quantum
Mechanics}, in Visions of Discovery: New Light on Physics, Cosmology And Consciousness, ed. R.
Y. Chiao, M. L. Cohen, A. J. Leggett, W. D. Phillips, and C. L. Harper, Jr. Cambridge: Cambridge
University Press, 2010.


\bibitem{av} Y. Aharonov, L. Vaidman, {\em Properties of a quantum system during the time interval between two measurements}, Phys. Rev. A, {\bf 41}, (1990), 11--20.

\bibitem{jmav} Y. Aharonov and L. Vaidman, {\em Complete Description of a quantum system at a given time},
J. Phys. A {\bf 24} (1991), 2315.


\bibitem{av2v} Y. Aharonov, L. Vaidman, {\em The Two-State Vector Formalism of Quantum Mechanics: an Updated Review}, in   {\em	 Time in Quantum Mechanics}, J. Muga, R. Sala Mayato and I. Egusquiza (Eds), Springer, Berlin, 2002.


\bibitem{Ahnert} S.E. Ahnert, M.C. Payne, {\em Linear optics implementation of weak values in Hardy's paradox},
Phys. Rev. A, {\bf 70} (2004), 042102.

\bibitem{balser 2004} W. Balser, {\em Summability of formal power-series solutions of partial differential equations with constant coefficients}, J. Math. Sci., {\bf 124} (2004), 5085--5097.


\bibitem{bd} C.A. Berenstein, M.A. Dostal, {\em Analytically Uniform Spaces and their Applications to Convolution Equations}, Lecture Notes in Mathematics, {\bf 256}, Springer Verlag, 1972.

\bibitem{bs} C.A. Berenstein, D. C. Struppa, {\em Dirichlet series and convolution equations},
Publ. RIMS, Kyoto Univ. {\bf 24} (1988), 783--810.

\bibitem{bs1} C.A. Berenstein, D. C. Struppa,
{\em Complex analysis and convolution equations},
   Current problems in mathematics. Fundamental directions, Vol.\
      54, {Itogi Nauki i Tekhniki},
      {Akad. Nauk SSSR Vsesoyuz. Inst. Nauchn. i Tekhn. Inform.},
      Moscow, {1989},
   {5--111}.

\bibitem{bs2} C. A. Berenstein, D.C. Struppa, {\em Interpolation and Dirichlet series: a new approach}, in: Geometrical and Algebraical Aspects in Several Complex Variables, (Cetraro 1989), 33-45, EditEl Rende, 1991.

\bibitem{bs3} C. A. Berenstein, D.C. Struppa, {\em Sheaves of holomorphic functions with growth conditions}, in: $\mathcal D$-modules and Microlocal Geometry (Lisbon, 1990), 63--74, Berlin, 1993.


\bibitem{bt1} C. A. Berenstein, B. A. Taylor,  {\em  A new look at interpolation theory for entire functions of one variable}, Adv. Math. {\bf 33} (1979),  109--143.

\bibitem{bt} C. A. Berenstein, B. A. Taylor,  {\em Interpolation problems in $\cc^n$ with applications
to harmonic analysis}, {J. Anal. Math.}, {\bf 38} (1980), 188--254.

    \bibitem{orb-ang-light2} G.C.G. Berkhout, M.P.J. Lavery, J. Courtial, Beijersbergen M.W., Padgett M.J., Efficient Sorting of Orbital Angular Momentum States of Light, {\em Physical Review Letters}, {\bf 105} (2010), 153601.


\bibitem{berry2} M.V. Berry, {\em Faster than Fourier},  in
Quantum Coherence and Reality; in celebration of the 60th Birthday of Yakir Aharonov, J.S. Anandan and J.L. Safko eds., World Scientific, Singapore, 55--65, 1994.

\bibitem{berry} M.V. Berry,  {\em Evanescent and real waves in quantum billiards and Gaussian beams},  J. Phys. A. {\bf 27} (1994), 391 .


\bibitem{berry-noise-2013}  M. Berry, {\em Exact nonparaxial transmission of subwavelength detail using superoscillations}, {J. Phys. A} {\bf 46}, (2013), 205203.

\bibitem{dennis2008} M. Berry, M.R. Dennis, {\em Superoscillation in speckle patterns}, Journal of Physics A: Mathematical and General, (2009).

\bibitem{b1} M.V. Berry, M.R. Dennis,  {\em  Natural superoscillations in monochromatic waves in D dimensions},
J.  Phys. A {\bf 42} (2009), 022003.

\bibitem{b6} M.V. Berry, M.R. Dennis, B. McRoberts, P. Shukla, {\em Weak value distributions for spin$\frac 12$}, {J. Phys. A}, {\bf 44} (2011), 205301.


\bibitem{b4} M.V. Berry, S. Popescu,
{\em Evolution of quantum superoscillations, and optical superresolution without evanescent waves},
{J.~Phys.~A}, {\bf 39} (2006), 6965-6977.

\bibitem{b5} M.V. Berry, P. Shukla, Pragya, {\em Pointer supershifts and superoscillations in weak measurements}, {\em J. Phys A}, {\bf 45} (2012), 015301.

\bibitem{Bliokh2} Y.K. Bliokh, {\em Geometrical optics of beams with vortices: Berry phase and orbital angular momentum Hall effect}, {Phys. Rev. Lett.} {\bf 97} (4),  Article Number: 043901, (2006).


\bibitem{Bliokh} Y.K. Bliokh, I.V. Shadrivov, Y.S. Kivshar, {\em Goos-Hanchen and Imbert-Fedorov shifts of polarized vortex beams}, { Optics Letters} {\bf 34} (3) (2009), 389--391.

\bibitem{bondcahn} F.E. Bond, C.R. Cahn, {\em On sampling the zeros of bandwidth limited signals}, {\em IRE Trans. Infor. Theory} {\bf 4} (3) (1958), 110--113.


\bibitem{new} A. Botero, {\em Sampling Weak Values: A Non-Linear Bayesian Model for Non-Ideal Quantum Measurements}, PhD dissertation, The University of Texas at Austin, 1999, available on arXiv:quant-ph/0306082.

\bibitem{popescu} R. Brout, S. Massar, R. Parentani, S. Popescu, Ph. Spindel, {\em Quantum back reaction on a classical field}, {Phys. Rev. D} {\bf 52} (1995), 1119.

\bibitem{bru04}
N. Brunner, V. Scarani, M. Wegm\"{u}ller, M. Legr\'{e}, N. Gisin, {\em Direct Measurement of Superluminal Group Velocity and Signal Velocity in an Optical Fiber},
Phys. Rev. Lett. {\bf 93}, 203902 (2004).

\bibitem{brunner} { N. Brunner, C. Simon, {\em Measuring small longitudinal phase shifts: weak measurements or standard interferometry?}, Phys. Rev. Lett. 105, 010405 (2010) }



\bibitem{harmonic} R. Buniy, F. Colombo,  I. Sabadini, D.C. Struppa,
{\em Quantum Harmonic Oscillator  with superoscillating initial datum}, J. Math. Phys., {\bf 55}  (2014), 113511.

\bibitem{cheon} T. Cheon, S. Poghosyan, {\em Weak value expansion of quantum operators and its application in stochastic matrices},  arXiv:1306.4767, 2013.


\bibitem{sciencewm}
A. Cho, {\em Particle Physicists New Extreme Teams}, Science, {\bf 333} (2011).


\bibitem{cho10}
Y.-W. Cho, H.-T. Lim, Y.-S. Ra, Y.-H. Kim, {\em Weak Value Measurement with an Incoherent Measuring Device}, New J. Phys. {\bf
12}, 023036 (2010).

\bibitem{ccat}
T. Denkmayr, H. Geppert, S. Sponar, H. Lemmel, A. Matzkin, J. Tollaksen,
 Y. Hasegawa, {\em Observation of a quantum Cheshire Cat in a matter-wave interferometer
experiment}, Nature Communications, 5, 2014.

\bibitem{dennis2004} M.R. Dennis, {\em Canonical representation of spherical functions: Sylvester's theorem, Maxwell's multipoles and Majorana's sphere}, J. Phys. A, {\bf 37} (2004), 9487--9500.

\bibitem {dennis 2008} M.R. Dennis, A. Hamilton, J. Courtial, {\em Optics Letters}, {\bf 33} (2008), 2976-8.

 \bibitem{howell}
P. B. Dixon, D. J. Starling, A. N. Jordan, J. C. Howell, {\em Ultrasensitive Beam Deflection Measurement via Interferometric Weak Value Amplification}, {Phys. Rev. Lett.} {\bf 102} (2009), 173601.

\bibitem{howell2}
P. B. Dixon, D. J. Starling, A. N. Jordan, J. C. Howell, {\em Optimizing the signal-to noise ratio of a beam-deflection measurement with interferometric weak values},  { Phys. Rev. A} {\bf 80}, 041803(R) (2009).


\bibitem{dsII}
N. Dunford, J. Schwartz, {\it Linear operators, part II: general
theory}, J. Wiley and Sons (1988).


\bibitem{e1} L. Ehrenpreis, {\em Solutions of some problems of division I}, { Am. J. Math.},
{\bf 76} (1954), 883--903.

\bibitem{e2} L. Ehrenpreis, {\em Solutions of some problems of division II}, { Am. J. Math.},
{\bf 77} (1955), 286--292.

\bibitem{e3} L. Ehrenpreis, {\em Solutions of some problems of division III}, { Am. J. Math.},
{\bf 78} (1956), 685--715.

\bibitem{ehrenpreis} L. Ehrenpreis, {\em Fourier Analysis in Several Complex Variables}, Wiley Interscience, New York 1970.


\bibitem{kempf1}
P. J. S. G. Ferreira, A. Kempf,
{\em Unusual properties of superoscillating particles},
{J. Phys. A}, {\bf 37} (2004), 12067-76.

\bibitem{kempf2}
P. J. S. G. Ferreira, A. Kempf,
{\em Superoscillations: faster than the Nyquist rate},
{IEEE trans. Signal. Processing}, {\bf 54} (2006), 3732-40.

\bibitem{feynman} R.P. Feynman, {\em Space-Time Approach to Non-Realistic Quantum Mechanics}, Rev. Mod. Phys. 20, 367 (1948).

\bibitem{fh} R.P. Feynman, A. R. Hibbs, {\em Quantum Mechanics and Path Integrals},
McGraw-Hill, New York, 1965.



\bibitem{gabor} D. Gabor, {\em Theory of communication}, J. IEE, {\bf 93} (1946), 429--457.

\bibitem{gallier} J. Gallier, {\em Discrete Mathematics}, Springer Universitext, Springer, 2011.

\bibitem{gy} I. M. Gelfand, A. M. Yaglom, {\em Integration in Functional Spaces}, J. Math. Phys. 1, 48 (1960)

\bibitem{Goerge}
R.E. George, L.M. Robledo, O.J.E. Maroney, M.S. Blokb, H. Bernien, M.L. Markham, D.J. Twitchen, J.J.L. Morton, A.D. Briggs, R. Hanson, {\em
Opening up three quantum boxes causes classically undetectable wavefunction collapse},
 PNAS \textbf{110}, 3777 (2013)


\bibitem{gog11}
M. E. Goggin, M. P. Almeida, M. Barbieri, B. P. Lanyon, J. L. O'Bryen, A. G. White, G. J. Pryde, {\em Violation of the Leggett-Grag Inequality with Weak Measurements of Photons}, Proc. Nat. Acad. Sci. {\bf 108}, 1256 (2011).


\bibitem{GR} I. S. Gradshteyn, I. M. Ryzhik, {Table of Integrals, Series, and Products}, Academic Press, 2007.

\bibitem{har92}
L. Hardy, {\em Quantum-mechanics, local realistic theories, and lorentz-invariant realistic theories}, Phys. Rev. Lett. {\bf 68}, 2981 (1992).


\bibitem{hog11} J. M. Hogan, J. Hammer, S.-W. Chiow, S. Dickerson,
D. M. S. Johnson, T. Kovachy, A. Sugarbaker, M. A. Kasevich,
{\em Precision angle sensor using an optical lever inside a {Sagnac}
interferometer}, 2011, Opt. Lett.,    volume 36,
    pages  1698-1700.



\bibitem{hormander} L. H\"ormander, {\em The Analysis of Linear Partial Differential Operators I},
Springer Verlag, Berlin Heidelberg, 1983.


\bibitem{kwiat}
O. Hosten and P. Kwiat, {\em Observation of the spin Hall effect of light via weak measurements}, Science {\bf 319}, 787 (2008).

\bibitem{how10}
J. C. Howell, D. J. Starling, P. B. Dixon, P. K. Vudyasetu, and A. N.
Jordan, {\em Interferometric Weak Value Deflections: Quantum and Classical Treatments}, Phys. Rev. A {\bf 81}, 033813 (2010).



\bibitem{ichinobe} K. Ichinobe, {\em Integral representation for Borel sum of divergent solution to a certain non-Kowalevski type equation},
Publ. RIMS, Kyoto University, {\bf 39} (2003), 657--693.

\bibitem{kaneko} A. Kaneko, {\em Introduction to Hyperfunctions}, Kluwer Academic Publishers, 1988.

\bibitem{katostruppa} G. Kato, D.C. Struppa, {\em Fundamentals of Algebraic Microlocal Analysis}, Monographs and Textbooks in Pure and Applied Mathematics, 217, Marcel Dekker, 1999.

    \bibitem{kempe} J. Kempe, {\em Quantum Random Walks: an introductory overview}, Contemporary Physics, {\bf 44} (2003), 307--327; quant-ph/0303081.

        \bibitem{qrw-phase2} T. Kitagawa, M. A. Broome, A. Fedrizzi, M. S. Rudner, E. Berg, I. Kassal, A. Aspuru-Guzik, E. Demler,  A. G. White, {\em Observation of Topologically Protected Bound States in Photonic Quantum Walks}, Nat. Comm. 3, 882 (2012).

\bibitem{qrw-phase} T. Kitagawa, M. S. Rudner, E. Berg, and E. Demler, {\em Exploring Topological Phases with Quantum Walks}, Phys. Rev. A 82, 033429 (2010).

\bibitem{kni90}
J. M. Knight and L. Vaidman, {\em Weak Measurement of Photon Polarization}, Phys. Lett. A {\bf 143}, 357 (1990).


\bibitem{Kolenderski}
P. Kolenderski, U. Sinha, L. Youning, T. Zhao, M. Volpini, A. Cabello,
R. Laflamme, T. Jennewein, {\em
Time-resolved double-slit interference pattern measurement with entangled photons},
 Phys. Rev. A \textbf{86}, 012321 (2012).


\bibitem{sonya} S. Kowalevski, {\em Zur Theorie der partiellen Differentialeichungen}, J. Reine Angew. Math., {\bf 80} (1875), 1--32.

\bibitem{orb-ang-light3} M.P.J. Lavery, G.C.G. Berkhout, J. Courtial, M.J. Padgett, {\em Measurement of the light orbital angular momentum spectrum using an optical geometric transformation}, {Journal of Optics} {\bf 13} (2011), 064006.

\bibitem{leeferreira} D.G. Lee, P.J.S.G. Ferreira, {\em Superoscillations of prescribed amplitude and derivative}, {IEEE Trans. Signal Processing} {\bf 62} (13) (2014), 3371--3378.

\bibitem{leeferreira2} D.G. Lee, P.J.S.G. Ferreira, {\em Superoscillations with optimal numerical stability}, { IEEE Sign. Proc. Letters} {\bf 21} (12) (2014), 1443--1447.

\bibitem{kyle lee} K. Lee, D.C. Struppa, {\em Some combinatorial identities arising from
superoscillatory sequences}, unpublished.

\bibitem{levin} B. Ja. Levin, {\em Distribution of zeros of entire functions}, Translation of Mathematical Monograph, AMS, Providence, Rhode Island, 1964.

\bibitem{lix11} C.-F. Li and X.-Y. Xu and J.-S. Tang and J.-S. Xu and G.-C. Guo,
{\em Ultrasensitive phase estimation with white light},
 2011,
 Phys. Rev. A,
    volume 83,
    pages  044102.

\bibitem{lindberg} J. Lindberg, {\em Mathematical concepts of optical superresolution}, {\em Journal of Optics} {\bf 14} (2012) 083001.

\bibitem{lun09}
J. S. Lundeen, A. M. Steinberg,  {\em Experimental Joint Weak Measurement on a Phtoton Pair as a Probe of Hardy’s Paradox}, Phys. Rev. Lett. {\bf 102},
020404 (2009).


\bibitem{lms} D.A. Lutz, M. Miyake, R. Sch\"afke, {\em On the Borel summability of divergent solutions of the heat equation},
Nagoya Math. J., {\bf 154} (1999), 1--29.

\bibitem{malgrange} B. Malgrange, {\em Existence et approximation des solutions des \'equations aux
   d\'eriv\'ees partielles et des \'equations de convolution},
   {Ann. Inst. Fourier, Grenoble}, {\bf 6},
   (1955--1956), 271--355.


\bibitem{mallat} S. Mallat, {\em A wavelet tour of signal processing}, Academic Press, Inc., San Diego, CA, 1998.

\bibitem{mir07}
R. Mir, J. S. Lundeen, M. W. Mitchell, A. M. Steinberg, J. L.
Garretson, H. M. Wiseman, {\em A Double-Slit “Which-Way” Experiment on the Complementarity – Uncertainty Debate}, New J. Phys. {\bf 9}, 287 (2007).


\bibitem{montroll} E. W. Montroll, {\em Markov Chains, Wiener Integrals, and Quantum Theory},
Commun. Pure and Appl. Math. 5, 415 (1952).

\bibitem{nt} S. Nussinov, J. Tollaksen, {\em Color Transparency in QCD and post-selection in quantum mechanics}, {Phys  Rev D}, {\bf 78} (2008), 036007.


\bibitem{brun} O. Oreshkov, T. A. Brun, {\em Weak Measurements are Universal},
Phys. Rev. Lett. {\bf 95}  (2005), 110409.


\bibitem{palamodov} V. P. Palamodov, {\em Linear Differential Operators with Constant Coefficients},
Springer-Verlag, Berlin, 1970.

\bibitem{qrw-alg} G. D. Paparo, V. Dunjko, A. Makmal, M. A. Martin-Delgado, H. J. Briegel, {\em Quantum Speedup for Active Learning Agents}, Phys. Rev. X 4, 031002 (2014).


\bibitem{par98}
A. D. Parks, D. W. Cullin, D. C. Stoudt, {\em Observation and Measurement of an Optical Aharonov-Albert-Vaidman Effect}, Proc. R. Soc. Lond.
{\bf 454}, 2997 (1998).


\bibitem{Pryde}  G.J. Pryde, J.L. O'Brien, A.G. White, T.C. Ralph, H.M. Wiseman, {\em
 Measurement of quantum weak values of photon polarization},
Phys. Rev. Lett., 94 (22) Art. No. 220405 JUN 10 2005.


\bibitem{superosc-misc2} W. Qiao,
{\em A simple model of Aharonov-Berry's superoscillations},
J. Phys. A {\bf 29}  (1996), 2257-2258.

\bibitem{requicha} A.A.G Requicha, {The zeros of entire functions: Theory and engineering applications}, {\em Proc. IEEE} {\bf 68} (3) (1980), 308--328.

\bibitem{steinber2004} K.J. Resch, J.S. Lundeen, A.M. Steinberg, {\em Experimental Realization of the Quantum Box Problem},
 Physics Letters A, 324 vol.2-3, 125 (2004).




\bibitem{RSH} N.W. M. Ritchie, J. G. Story, R. G. Hulet, {\em Realization of a measurement of a “weak value”, }{  Phys. Rev.Lett.} {\bf 66}, 1107 (1991).



\bibitem{natmat-2012} E.T.F. Rogers, J. Lindberg, T. Roy, S. Savo, J.E. Chad, M.R. Dennis, N.I. Zheludev, {\em Nature Materials}, {\bf 11} (2012), 432--435.

\bibitem{rudin} W. Rudin, {\em Functional Analysis},
Functional analysis. McGraw-Hill Series in Higher Mathematics.
McGraw-Hill Book Co., New York-D\"usseldorf-Johannesburg, 1973.



\bibitem{sakurai}
J. J. Sakurai {\em Modern Quantum Mechanics}, Addison-Wesley Publishing Company, 1995.

\bibitem{schulman} L. S. Schulman, Techniques and Applications of Path Integration, Dover, 2005.

\bibitem{sol04}
D. R. Solli, C. F. McCormik, R. Y. Chiao, S. Popescu, and J. M.
Hickmann, {\em Fast Light, Slow Light, and Phase Singularities: A Connection to Generalized Weak Values}, Phys. Rev. Lett. {\bf 92}, 043601 (2004).


\bibitem{sta09}
D. J. Starling, P. B. Dixon, A. N. Jordan, J. C. Howell, {\em Optimizing the Signal-to-Noise Ratio of a Beam-Deflection Measurement with Interferometric Weak Values}, Phys.
Rev. A {\bf 80}, 041803(R) (2009).


\bibitem{sta10}
D. J. Starling, P. B. Dixon, N. S. Williams, A. N. Jordan, and J. C.
Howell, {\em Precision Frequency Measurements with Interferometric Weak Values}, Phys. Rev. A {\bf 82}, 011802(R) (2010).


\bibitem{struppamemoirs} D.C. Struppa, {\em The fundamental principle for systems of convolution equations},
{Memoirs Amer. Math. Soc.} {\bf 273}, 1983.



\bibitem{sut95}
D. Suter, {\em Weak Measurements and the Quantum Time-Translation Machine in a Classical System}, Phys. Rev. A {\bf 51}, 45 (1995).


\bibitem{sut93}
D. Suter, M. Ernst, R. R. Ernst, {\em Quantum Time-Translation Machine – An Experimental Realization}, Mol. Phys. {\bf 78}, 95 (1993).


\bibitem{taylor} B.A. Taylor, {\em Some locally convex spaces of entire functions},
Proc. Symp. Pure Math., {\bf 11} (1968), 431--467.

\bibitem{jt} J. Tollaksen,  {\em Quantum Reality and Nonlocal Aspects of Time}, 2001 PhD thesis, Boston University, ISBN 978 054 941 2946.

\bibitem{spie-nswm} J. Tollaksen,  {\em Non-statistical weak measurements}, in
{Quantum Information and Computation V}, E. Donkor, A. Pirich, H. Brandt (eds),
Proc. of SPIE Vol. 6573 (SPIE, Bellingham, WA, 2007), CID 6573-33.


\bibitem{at3} J. Tollaksen, {\em Robust Weak Measurements on Finite Samples},
J. Phys. Conf. Series, vol. 70, (2007), 012015, quant-ph/0703038.

\bibitem{jt2010} J. Tollaksen, Y. Aharonov, A. Casher, T. Kaufherr, S. Nussinov,{\em Quantum interference experiments, modular variables and weak measurements}, New Journal of Physics, {\bf 12} (2010), 013023.

\bibitem{tur11} M. D. Turner, C. A. Hagedorn, S. Schlamminger, J. H. Gundlach, {\em Picoradian deflection measurement with an interferometric quasi-autocollimator using weak value amplification}, 2011, Opt. Lett., volume 36,   1479--1481.

\bibitem{vladimirov1} V. Vladimirov, {\em Distributions en physique math\'ematique. (French) [Distributions in mathematical physics]}, Mir, Moscow, 1979.

\bibitem{vNBOOK} J. Von Neumann, {\em Mathematical Foundations of Quantum Mechanics}, Princeton University Press, Princeton, 1983.

\bibitem{wan06}
Q. Wang, F.-W. Sun, Y.-S. Zhang, Jian-Li, Y.-F. Huang, G.C. Guo, {\em Experimental Demonstration of a Method to Realize Weak Measurement of the Arrival Time of a Single Photon},
Phys. Rev. A {\bf 73}, 023814 (2006).


\bibitem{Wiseman} H. M. Wiseman,{\em   Weak values, quantum trajectories, and the cavity-QED experiment on wave-particle correlation},
Phys. Rev. A, {\bf  65} (2002), Art. No. 032111.



\bibitem{orb-ang-light} A. M. Yao, M.J. Padgett, {\em Orbital angular momentum: origins, behavior and applications}, Advances in Optics and Photonics {\bf 3}, (2011), 161--204.

\bibitem{yok09}
K. Yokota, T. Yamamoto, M. Koashi, N. Imoto, {\em Direct Observation of Hardy’s paradox by joint weak measurement with an entangled photon pair}, New J. Phys. {\bf
11}, 033011 (2009).


\bibitem{yosida} K. Yosida, {\em Functional analysis},
   {Grundlehren der Mathematischen Wissenschaften, Band 123},
   {Springer-Verlag},
   {Berlin}, {1978}.


\bibitem{gefen} { O. Zilberberg, A. Romito, Y. Gefen, {\em Charge
sensing amplification via weak values measurement}, Phys. Rev. Lett. 106, 080405 (2011) arXiv:1009.4738. }












\end{thebibliography}
\end{document}